\documentclass[letter,11pt]{article}%
\usepackage[latin9]{inputenc}
\usepackage{geometry}
\usepackage{color}
\usepackage{amsmath}
\usepackage{amsthm}
\usepackage{amssymb}
\usepackage{graphicx}
\usepackage{setspace}
\usepackage{esint}
\usepackage{amsfonts}
\usepackage{comment}
\usepackage{booktabs}
\usepackage{lscape}
\usepackage{chngpage}
\usepackage{array}
\usepackage[titletoc,title]{appendix}
\usepackage[labelsep=period]{caption}
\usepackage{epstopdf}
\usepackage{natbib}
\usepackage[colorlinks,citecolor=blue,urlcolor=blue,filecolor=blue]{hyperref}
\usepackage{diagbox}%
\usepackage{mathtools}
\usepackage{threeparttable}
\usepackage{multirow}
\usepackage{float}
\usepackage{bm}
\usepackage{siunitx}
\DeclarePairedDelimiter\ceil{\lceil}{\rceil}
\DeclarePairedDelimiter\floor{\lfloor}{\rfloor}
\providecommand{\U}[1]{\protect\rule{.1in}{.1in}}
\providecommand{\U}[1]{\protect\rule{.1in}{.1in}}
\providecommand{\U}[1]{\protect\rule{.1in}{.1in}}

\makeatletter
\theoremstyle{plain}

\theoremstyle{definition}

\theoremstyle{plain}

\RequirePackage[colorlinks,citecolor=blue,urlcolor=blue]{hyperref}
\RequirePackage{hypernat}
\providecommand{\U}[1]{\protect\rule{.1in}{.1in}}
\newtheorem{lemma}{Lemma}
\newtheorem{theorem}{Theorem}

\newtheorem{assumption}{Assumption}
\newtheorem{remark}{Remark}

\providecommand{\definitionname}{Definition}
\providecommand{\propositionname}{Proposition}
\providecommand{\theoremname}{Theorem}
\providecommand{\definitionname}{Definition}
\providecommand{\propositionname}{Proposition}
\providecommand{\theoremname}{Theorem}
\providecommand{\definitionname}{Definition}
\providecommand{\propositionname}{Proposition}
\providecommand{\theoremname}{Theorem}
\makeatother
\providecommand{\definitionname}{Definition}
\providecommand{\propositionname}{Proposition}
\providecommand{\theoremname}{Theorem}

\graphicspath{{figures/}}

\usepackage{natbib}
 \bibpunct[, ]{(}{)}{,}{a}{}{,}%
 \def\bibfont{\small}%

\begin{document}

\title{\bf Monte-Carlo Estimation of CoVaR}

\author{
Weihuan Huang\\[-6pt]
{\small School of Management \& Engineering, Nanjing University, Nanjing 210093, China}\\[6pt]
Nifei Lin\\[-6pt]
{\small School of Management, Fudan University, Shanghai 200433, China}\\[6pt]
L. Jeff Hong\\[-6pt]
{\small School of Management and School of Data Science, Fudan University, Shanghai 200433, China}
}

\footnotetext[1]{L. Jeff Hong (corresponding author), email: hong\_liu@fudan.edu.cn }

\date{ }
\maketitle

\singlespacing
\begin{abstract}
\noindent ${\rm CoVaR}$ is one of the most important measures of financial systemic risks. It is defined as the risk of a financial portfolio conditional on another financial portfolio being at risk. In this paper we first develop a Monte-Carlo simulation-based batching estimator of CoVaR and study its consistency and asymptotic normality. We show that the optimal rate of convergence of the batching estimator is $n^{-1/3}$, where $n$ is the sample size. We then develop an importance-sampling inspired estimator under the delta-gamma approximations to the portfolio losses, and we show that the rate of convergence of the estimator is $n^{-1/2}$. Numerical experiments support our theoretical findings and show that both estimators work well. 
\end{abstract}




\onehalfspacing
\section{Introduction}

Financial institutions are inter-connected because they may be counter-parties of the same financial contracts or they may hold the same financial assets. These connections create potential channels to propagate or even to amplify financial risks. Such risks are known as financial systemic risks, and they bring challenges to financial risk management. When financial systems are stable, financial institutions may measure their own risks in isolation, assuming that others will not default on their contracts or sell off their assets. In financial crisis, however, the assumption may not hold. The time that a financial institution is in great need of cash inflow may also be the time that its counter-parties cannot honor their contracts and its assets are devalued significantly. For instance, in the 2007-2009 financial crisis, the risks spread from structured investment vehicles to commercial banks and then to investment banks and hedge funds, led otherwise healthy financial institutions to default and finally caused a global financial crisis. Therefore, to manage financial systemic risk, we need to take into account the inter-connections among financial institutions and better understand how the risk of one financial institution affects the risk of another.

To manage financial systemic risk, the first step is to measure it. The most common risk measure used by financial institutions is the value-at-risk (${\rm VaR}$), which is defined as the upper quantile of an institution's portfolio loss distribution when it is considered in isolation \citep{Jorion}. However, it does not consider the inter-connections among financial institutions and is not an appropriate measure of financial systemic risk. \cite{CoVaR2016} propose a systemic risk measure ${\rm CoVaR}$, which is defined as the VaR of a portfolio loss conditional on another portfolio being at its VaR. They show that ${\rm CoVaR}$ well captures the cross-sectional tail-dependency between the whole financial system and a particular institution, and could predict the 2007-2009 crisis. Since then, ${\rm CoVaR}$ has become one of the most important measures of financial systemic risk. 

However, estimating CoVaR is a challenging problem because CoVaR is a conditional quantile conditioned on a probability zero event, which cannot be observed directly in the data. There is an emerging literature that handles this problem by assuming the loss distributions to follow certain structural models, where the model parameters may be estimated using financial data. The CoVaR may then be calculated given the models. Based on a linear factor model, i.e., the two portfolio losses are in linear relationship, \cite{CoVaR2016} propose a quantile regression approach to estimate CoVaR. They and \cite{girardi2013systemic} also  use GARCH models to capture the dynamic evolution of systemic risk contributions. Copula models are also popular in CoVaR estimations because they are convenient in modeling the dependence between the portfolio losses. For instance, \cite{mainik2014dependence} present analytical results for CoVaR using copulas. \cite{oh2018time} use a new class of copula-based dynamic models for high-dimensional conditional distributions, facilitating the estimation of CoVaR. \cite{karimalis2018measuring} also provide a simple closed-form expression of CoVaR for a broad range of copula families and allows time-varying exposures. One can also make distributional assumptions and use maximum likelihood techniques to estimate CoVaR. For instance, \cite{bernardi2013multivariate} estimate CoVaR using a multivariate Markov switching model with a student-t distribution accounting for heavy tails and nonlinear dependence, and \cite{cao2013multi} estimates a multivariate student-t distribution to calculate the CoVaRs across firms. Recently, \cite{bianchi2022non} show how to develop CoVaR estimators based on models where heavy tails, negative skew, asymmetric dependence and volatility clustering are taken into consideration. \cite{nolde2022extreme} develop a methodology to estimate CoVaR semi-parametrically within the framework of multivariate extreme value theory. 

These model-based approaches are efficient if the models are appropriately specified. Otherwise, they introduce bias that cannot be easily removed. In practice the portfolios of financial institutions are typically quite complicated, including many assets (e.g., financial derivatives) that are nonlinear in underlying risk factors. The aforementioned models may not be able to capture their dependence and may lead to significantly biased CoVaR estimators.

Monte-Carlo simulation is in general a flexible modeling technique that may capture the complex structures and dynamics in portfolio losses. It is widely used in financial engineering and risk management \citep{glasserman2004monte,hull2012risk}, and its usages in estimating and optimizing traditional risk measures, such as VaR and expected shortfalls, are well studied in the literature (see, for instance, \cite{Hong2014} for a comprehensive review of the topic). However, to the best of our knowledge, Monte-Carlo estimation of CoVaR has not been studied carefully. In this paper our goal is to fill this gap and to develop CoVaR estimators that can take advantage of the modeling flexibility of Monte-Carlo simulation and have provable convergence including the rate of convergence and the asymptotically valid confidence intervals. 

To develop Monte-Carlo estimators of CoVaR, we also need to handle the difficulty of conditioning on a probability zero event. We first propose a batching estimator. The idea is to divide the Monte-Carlo observations into multiple batches and use each batch to generate an observation from the conditional distribution where the condition holds approximately true. Once we have the (approximate) observations from the conditional distribution, the conditional quantile (i.e., the CoVaR) may be estimated. The idea of batching in handling the conditioning on a probability zero event is not new. \cite{Hong09} uses the same idea in estimating quantile sensitivity, which is a conditional expectation, instead of a conditional quantile as in CoVaR, conditioning on a probability zero event. To understand the large-sample behaviors of the batching estimator of CoVaR, we analyze its consistency and asymptotic normality. We show that the estimator is strongly consistent if both of the number of batches and the number of observations within a batch go to infinity as the total sample size $n$ goes to infinity. We also show that it is asymptotically normally distributed under mild conditions and the optimal rate of convergence is $n^{-1/3}$. The asymptotic normal distribution of the CoVaR estimator may be used to construct an asymptotically valid confidence interval of the CoVaR.

Although the batching estimator is strongly consistent, its optimal rate of convergence is only $n^{-1/3}$ and it is slower than $n^{-1/2}$, the typical rate of convergence of Monte-Carlo estimators. This slower rate is caused by the conditioning on a probability zero event, and it also implies that a large sample size is often needed to obtain an accurate estimate of CoVaR. To reduce the variance of the estimator and to improve the rate of convergence, we consider a special case where the two portfolio losses are modelled by delta-gamma approximations. Notice that delta-gamma approximations model the portfolio losses as quadratic functions of underlying risk factors, and they are commonly used in approximating losses of complicated financial portfolios \citep{hull2012risk}. Under the delta-gamma approximations, \cite{glasserman2000variance} and \cite{glasserman2002portfolio} use importance sampling (IS) techniques to reduce the variance of the VaR estimator.
In this paper we propose an IS scheme to the last dimension of the risk factors in the delta-gamma approximations so that the probability-zero condition in the definition of CoVaR holds approximately. Utilizing the structure of quadratic functions, we further show that {\it there exists a proper limit of the IS estimator that does not depend on the IS distribution of the last risk factor at all}. Therefore, we arrive at a new CoVaR estimator that is under the original probability distributions instead of the IS distributions. For this reason, we call the new estimator the ``IS-inspired CoVaR estimator". The new IS-inspired estimator not only reduces the estimation variance but also improves the rate of convergence to $n^{-1/2}$, successfully bypassing the difficulty of conditioning on a probability zero event. The idea of using IS to improve the rate of convergence has also been explored by \cite{liu2015simulating} in studying credit risk contributions. The difference is that he estimates conditional expectations under a linear copula model for portfolios of loans, while we estimate conditional quantiles of more complicated quadratic portfolios for portfolios of financial derivatives. This brings different structures and new challenges in developing IS-inspired estimators. We also prove the asymptotic normality of the estimator and develop an asymptotically valid confidence interval.

The rest of this paper is organized as follows: The problem is defined in Section \ref{sec2}. We then introduce the batching estimator and study its asymptotic properties in Section \ref{sec3}. In Section \ref{sec4} we introduce the delta-gamma approximations to portfolio losses and develop the IS-inspired estimator and its asymptotic properties. In Section \ref{sec5} we conduct numerical experiments to understand the performances of the two estimators on different types of problems, followed by conclusions in Section \ref{sec6}.

\section{Problem Definition}\label{sec2}

Let $X$ and $Y$ be two continuous random variables on a probability space $(\Omega, \mathcal{F}, {\rm Pr})$.  In the context of financial systemic risk management, $X$ and $Y$ may denote the losses of two financial portfolios. For example, $X$ and $Y$ may be the losses of the portfolios of two financial institutions, or $X$ may be the loss of the portfolio of a financial institution and $Y$ may be the loss of an index (which also represents a portfolio). Let ${\rm VaR}_\alpha(X)$ be the $\alpha$-VaR (i.e., $\alpha$-quantile) of $X$ with $\alpha\in(0,1)$. It satisfies 
\begin{equation}\label{VaRX}
	{\rm Pr}\left\{ X\leq {\rm VaR}_\alpha(X)\right\} = \alpha. 
\end{equation}
Notice that ${\rm VaR}_\alpha(X)$ means that we are $100\times\alpha\%$ confident that the random loss $X$ will not exceed  ${\rm VaR}_\alpha(X)$ and it measures the tail risk of the portfolio $X$. The concept of VaR was first proposed by J.P. Morgan in early 1990s and has become a widely adopted risk measure in global financial industries since then \citep{Jorion,DuffiePan}. However, VaRs cannot measure financial systemic risk, which caused significant losses and led to the collapse of major financial institutions in the 2007-2009 financial crisis. 

In the wake of the crisis, \cite{CoVaR2016} propose ${\rm CoVaR}_{\alpha,\beta}$ as a measure of financial systemic risk, which satisfies
\begin{equation}\label{defCoVaR}
	{\rm Pr}\left\{ Y\leq {\rm CoVaR}_{\alpha,\beta} | X= {\rm VaR}_\alpha(X)\right\} = \beta, 
\end{equation}
where $\alpha,\beta\in(0,1)$. It is the $\beta$-quantile of the conditional loss distribution of $Y$ conditioning on $ X= {\rm VaR}_\alpha(X)$, and it measures the tail risk of the portfolio $Y$ when the portfolio $X$ is at risk. Notice that ${\rm CoVaR}_{\alpha,\beta}={\rm VaR}_\beta(Y)$ if the two losses $X$ ad $Y$ are independent. However, the losses of financial portfolios are typically positively dependent. Then, ${\rm CoVaR}_{\alpha,\beta}$ is typically significantly larger than ${\rm VaR}_\beta(Y)$, indicating that the tail risk at the time of financial distress is significantly higher than that at the normal time.

Notice that both $X$ and $Y$ may be losses of complicated portfolios and their dependence may be quite difficult to capture using a simple parametric model. Monte-Carlo simulation models may be developed to simulate the dynamics of the portfolios and to generate observations of the losses  \citep{glasserman2004monte,hull2012risk}. Suppose that we have observed an independent and identically distributed (i.i.d.) sample of the losses from the simulation model, denoted by $(X_1,Y_1), (X_2,Y_2), \ldots, (X_n, Y_n)$. In this paper our goal is to develop an estimator of ${\rm CoVaR}_{\alpha,\beta}$ using the sample. Furthermore, because we can often afford a relatively large sample size in Monte-Carlo studies, in this paper we also want to understand the asymptotic properties of the estimator as the sample size $n$ goes to infinity.

Let $f(x,y)$ denote the joint density function of $(X,Y)$. Let $f_X(x)=\int_\mathbb{R}f(x,y){\rm d}y$ and $F_X(x)=\int_{-\infty}^x f_X(u){\rm d}u$ be the marginal density function and cumulative distribution function of $X$, respectively. Based on \cite{durrett2019probability}, we let
\begin{equation}\label{defofCP}
	F_{Y|X}(y|x)={\rm Pr}\{ Y\leq y | X=x \}  = \lim\limits_{\varepsilon\rightarrow 0}\frac{{\rm Pr}\{Y\leq y, |X- x|\leq \varepsilon\}}{{\rm Pr}\left\{ |X- x|\leq \varepsilon\right\}}
	= \int_{-\infty}^{y} \frac{f(x,v)}{f_X(x)} \mathrm{d}v
\end{equation}
be the conditional distribution function. To facilitate the development and the analysis of the CoVaR estimator, we make the following assumption on the distribution of $(X,Y)$.

\begin{assumption}\label{assu:dist}
	Let $\mathcal{X}\subset \mathbb{R}$ be a neighborhood of $x= {\rm VaR}_\alpha(X)$ and $\mathcal{Y}\subset \mathbb{R}$ be a neighborhood of $y= {\rm CoVaR}_{\alpha,\beta}$. Then, $f_X(x)$ and $f(x,y)$ are continuously differentiable and positive in $\mathcal{X}$ and $\mathcal{X}\times\mathcal{Y}$, respectively, and for any $y\in\mathcal{Y}$, $F_{Y|X}(y|x)$ is a twice differentiable function of $x$ in $\mathbb{R}$.
\end{assumption}

Notice that the assumption of continuous and positive density in a neighborhood of the VaR is common in analyzing the properties of VaRs, see \cite{Hong09}. By Assumption \ref{assu:dist}, it is clear that ${\rm VaR}_\alpha(X)$ is the unique value satisfying Equation \eqref{VaRX} and ${\rm VaR}_\alpha(X)= F_X^{-1}(\alpha)$. Furthermore, because  $f(x,y)$ is positive in $\mathcal{X}\times\mathcal{Y}$,  Assumption \ref{assu:dist} also guarantees that ${\rm CoVaR}_{\alpha,\beta}$ is the unique solution of Equation \eqref{defCoVaR}. In fact, we know that, for $x\in\mathcal{X}$, $F_{Y|X}(y|x)$ is a differentiable function of $y$ and the conditional density $f_{Y|X}(y|x)$ satisfies
\begin{equation*}\label{Fpiany}
	f_{Y|X}(y|x) = \frac{\partial}{\partial y} F_{Y|X}(y|x) = \frac{f(x,y)}{f_X(x)} > 0, \text{ for } y\in\mathcal{Y}. 
\end{equation*}
Therefore, for any $x\in\mathcal{X}$, we have an inverse function $y=F_{Y|X}^{-1}(\beta|x)$, and 
\begin{equation}\label{ConQuan}
	{\rm CoVaR}_{\alpha, \beta}= F^{-1}_{Y|X}(\beta\,|\,{\rm VaR}_\alpha(X)). 
\end{equation} 

In this paper, we use the notation $Y_n=O_{\rm Pr}(X_n)$ to denote that, for any $\varepsilon>0$, there exists $c>0$ such that ${\rm Pr}\{|Y_n/X_n|>c\} \leq \varepsilon$ for all $n\in\mathbb{N}$, use the notation w.p.1 to denote ``with probability 1" (also known as ``almost surely"), and use the notation $X_n\Rightarrow X$ to denote that $X_n$ converges in distribution to $X$.

\section{Batching Estimation}\label{sec3}

As pointed out in the Introduction, the difficulty in estimating CoVaR lies in the fact that it is a conditional quantile that conditions on a probability zero event $\{X={\rm VaR}_\alpha(X)\}$. In this section we propose a batching estimator to address this difficulty. The basic idea is to divide the data into multiple batches, use each batch to obtain an observation from the approximate conditional distribution, and then use the multiple observations to estimate the CoVaR. The estimator is straight-forward, but analyzing its asymptotic properties is quite challenging. We present the estimator in detail in Section \ref{subsec:BE:est} and show its strong consistency and asymptotic normality in Sections \ref{subsec:BE:con} and \ref{subsec:BE:nor}, respectively.

\subsection{The Estimator}\label{subsec:BE:est}

We have an i.i.d.\ sample $\{(X_1,Y_1), (X_2,Y_2), \ldots, (X_n,Y_n)\}$ with the sample size $n$, and we take the following three-step approach to estimate the CoVaR. 
\begin{description}
\item[Step 1.] We divide the data into $k$ batches and each batch has $m$ observations with $n=k\times m$, and denote the observations in the $i$-th batch as $\{(X_{i,j}, Y_{i,j})\}_{j=1}^m$, $i=1,2,\ldots,k$. 

\item[Step 2.] For each batch (say $i$-th batch), we sort $X_{i,1},\ldots,X_{i,m}$ from lowest to highest, denoted by $X_{i,(1)}\leq X_{i,(2)}\leq \cdots \leq X_{i,(m)}$, where $X_{i,(j)}$ denotes the $j$-th smallest value, which is also the $j$-th order statistic of the batch. Then, by \cite{Serfling1980}, $X_{i,(\ceil{\alpha m})}$ is a strongly consistent estimator of ${\rm VaR}_\alpha(X)$. Let $\hat Y_i = Y_{i,(\ceil{\alpha m})}$, where $Y_{i,(\ceil{\alpha m})}$ is the corresponding observation of $X_{i,(\ceil{\alpha m})}$. It is important to note that $Y_{i,(j)}$ is {\it not} the $j$-th order statistic of $\{Y_{i,1},\ldots,Y_{i,m}\}$, $(j)$ represents the order of $\{X_{i,1},\ldots,X_{i,m}\}$. Furthermore, let $\hat Y =(Y|X= X_{(\ceil{\alpha m})})$ be the conditional random variable. Notice that $\{\hat Y_1,\ldots,\hat Y_k\}$ is an i.i.d.\ sample of $\hat Y$.

\item[Step 3.] We sort $\hat Y_1,\ldots,\hat Y_k$ from lowest to highest, denoted by $\hat Y_{(1)}\le\hat Y_{(2)}\le\cdots\le\hat Y_{(k)}$. Then, we define the batching estimator of ${\rm CoVaR}_{\alpha,\beta}$ as
\begin{equation*}\label{BE}
	\hat{Y}^{\rm BE}= \hat{Y}_{(\ceil{\beta k})}.
\end{equation*}
\end{description}

Notice that ${\rm CoVaR}_{\alpha,\beta}$ is the conditional quantile of $Y|X={\rm VaR}_\alpha(X)$. To estimate it, the major difficulty is that $\{X= {\rm VaR}_\alpha(X)\}$ is a probability zero event and it cannot be observed in the data. To solve this problem, the batching estimator approximates the set $\{X= {\rm VaR}_\alpha(X)\}$ by the set $\{X= X_{(\ceil{\alpha m})}\}$, which guarantees to have an observation of $Y|X= X_{(\ceil{\alpha m})}$ in each batch. Once there are (approximate) observations, the conditional quantile can be estimated approximately. In the rest of this section we analyze the asymptotic properties of the batching estimator as the sample size $n$ goes to infinity and provide guidelines on how to select the parameters $k$ and $m$.

\subsection{Strong Consistency}\label{subsec:BE:con}

Notice that, by \cite{Serfling1980}, the batching estimator $\hat{Y}^{\rm BE}$ is the $\beta$-inverse of the empirical distribution function of $\hat Y_1,\ldots,\hat Y_k$, defined by
\begin{equation*}
	\hat{F}_{k}(y) = \frac1k \sum_{i=1}^k I\{\hat{Y}_i\leq y\}.
\end{equation*}
To understand the consistency of $\hat{Y}^{\rm BE}$, we first analyze the convergence of $\hat{F}_{k}(y)$ to the conditional distribution function $F_{Y|X}(y\,|\,{\rm VaR}_\alpha(X))$. We divide the error of $\hat{F}_{k}(y)- F_{Y|X}(y\,|\,{\rm VaR}_\alpha(X))$ into two parts: 
\begin{equation}
	\label{2parts}
	\hat{F}_{k}(y)- F_{Y|X}(y\,|\,{\rm VaR}_\alpha(X))\ =\ \underbrace{\vphantom{\frac11}\hat{F}_{k}(y)- {\rm E}\big[ I\{\hat{Y}\leq y\} \big]}_{\rm across-batch~error} \ +\ \underbrace{\vphantom{\frac11}{\rm E}\big[I\{\hat{Y}\leq y\} \big]- F_{Y|X}(y\,|\,{\rm VaR}_\alpha(X)) }_{\rm within-batch~error}. 
\end{equation}
We see that the across-batch error is caused by the variance of $\hat{F}_{k}(y)$, and the within-batch error is the bias of $\hat{F}_{k}(y)$. 
Notice that, the within-batch error only depends on $m$, the number of observations in each batch,  while the across-batch error depends on both $m$ and $k$, the number of batches. 

We follow Equation \eqref{2parts} to analyze the convergence of the two terms separately. In the following two lemmas, we prove that both terms have the desired convergence.

\begin{lemma}\label{WBC}
	Suppose Assumption \ref{assu:dist} holds. Then, ${\rm E}\big[I\{\hat{Y}\leq y\} \big]\rightarrow F_{Y|X}(y\,|\,{\rm VaR}_\alpha(X))$ for any $y\in\mathcal{Y}$ as $m\rightarrow\infty$. 
\end{lemma}

\begin{proof}
By the law of total expectation, we have 
\begin{equation}\label{eqn:WBC1}
	{\rm E}\left[ I\{\hat{Y}\leq y\} \right] = {\rm E}\left[{\rm E}\left[ I\{\hat{Y}\leq y\}\, |\, X_{(\ceil{\alpha m})} \right]\right] = {\rm E}\left[{\rm Pr}\left\{ Y\leq y\, |\, X= X_{(\ceil{\alpha m})} \right\}\right] = {\rm E}\left[F_{Y|X}\left( y\, |\, X_{(\ceil{\alpha m})} \right)\right]. 
\end{equation}

Notice that, by Assumption \ref{assu:dist}, for any  $y\in\mathcal{Y}$, $F_{Y|X}(y|x)$ is a continuous function of $x\in \mathbb{R}$. Furthermore, $X_{(\ceil{\alpha m})}\rightarrow {\rm VaR}_\alpha(X)$ w.p.1 as $m\rightarrow \infty$ \citep{Serfling1980}. Then, by the continuous mapping theorem \citep{van2000}, we have $F_{Y|X}\left( y\, |\, X_{(\ceil{\alpha m})} \right) \rightarrow F_{Y|X}(y\,|\,{\rm VaR}_\alpha(X))$ w.p.1 as $m\rightarrow\infty$. 

Furthermore, because $0\leq F_{Y|X}\left( y\, |\, X_{(\ceil{\alpha m})} \right)\leq 1$ for all $m$. Then, by the dominated convergence theorem \citep{durrett2019probability}, we have 
\begin{equation*}\label{DCT}
	\lim\limits_{m\rightarrow\infty}{\rm E}\left[ F_{Y|X}\left( y\, |\, X_{(\ceil{\alpha m})} \right) \right] = F_{Y|X}(y\,|\,{\rm VaR}_\alpha(X)).
\end{equation*}
Then, by Equation (\ref{eqn:WBC1}), we have ${\rm E}\big[I\{\hat{Y}\leq y\} \big]\rightarrow F_{Y|X}(y\,|\,{\rm VaR}_\alpha(X))$ for any $y\in\mathcal{Y}$ as $m\rightarrow\infty$. 
\end{proof}

\begin{lemma}\label{ABC}
	Suppose that Assumption \ref{assu:dist} holds. Then, we have $\hat{F}_{k}(y)\rightarrow {\rm E}\big[ I\{\hat{Y}\leq y\} \big]$ w.p.1 as $k \rightarrow\infty$. 
\end{lemma}

\begin{proof}
For $m\geq 1$ and $i\geq 1$, we have $0\leq I\{\hat{Y}_i\leq y\}\leq 1$. Then, by Hoeffding's inequality \citep{Serfling1980}, for any $\varepsilon>0$, 
\begin{equation}\label{eq15}
	{\rm Pr}\left\{ \left| \hat{F}_{k}(y)- {\rm E}\left[ I\{\hat{Y}\leq y\} \right]  \right|\geq \varepsilon \right\} 
	\leq 2e^{-2k\varepsilon^2}. 
\end{equation}
Therefore, $\sum_{k=1}^\infty {\rm Pr}\left\{ \big| \hat{F}_{k}(y)- {\rm E}\big[ I\{\hat{Y}\leq y\} \big]  \big|\geq \varepsilon \right\}<\infty$. 
Hence, by the Borel-Cantelli Lemma \citep{Serfling1980}, we conclude the lemma.  
\end{proof}

Lemmas \ref{WBC} and \ref{ABC} basically show that both the within-batch and across-batch errors converge to zero as $m$ and $k$ both go to infinity. In the proofs of both lemmas, we take advantage of the boundedness of both the indicator function and empirical distribution function, which allows us to prove the strong consistency through the dominated convergence theorem and the Hoeffdling's inequality, without any additional assumptions. Combining these two lemmas with Equation (\ref{2parts}), we have the strong consistency of $\hat{F}_k(y)$ in the following theorem. 

\begin{theorem}\label{consistencyofF}
	Suppose that Assumption \ref{assu:dist} holds and $m\rightarrow\infty$ and $k\rightarrow\infty$ as $n\rightarrow\infty$.  Then, for any $y\in\mathcal{Y}$, we have $\hat{F}_{k}(y) \rightarrow F_{Y|X}(y\,|\,{\rm VaR}_\alpha(X))$ w.p.1 as $n\rightarrow\infty$. 
\end{theorem}

Let $\hat{F}^{-1}_{k}(z)= \inf\{y\in\mathbb{R}: \hat{F}_{k}(y)\geq z\}$  for any  $z\in[0,1]$. Based on the property of quantile estimator \citep{Serfling1980}, we have $\hat{Y}^{\rm BE}= \hat{F}^{-1}_{k}(\beta)$. Furthermore, as shown in Equation \eqref{ConQuan}, we have ${\rm CoVaR}_{\alpha, \beta}= F^{-1}_{Y|X}(\beta\,|\,{\rm VaR}_\alpha(X))$. In the following theorem, we use the convergence of $\hat{F}_{k}(y)$ to $F_{Y|X}(y\,|\,{\rm VaR}_\alpha(X))$ established in Theorem \ref{consistencyofF} to show that the inverse $\hat{F}^{-1}_{k}(\beta)$ converges to $F^{-1}_{Y|X}(\beta\,|\,{\rm VaR}_\alpha(X))$ as well, which implies the convergence of the batching estimator to the CoVaR.

\begin{theorem}\label{thm3}
	Suppose that Assumption \ref{assu:dist} holds and $m\rightarrow\infty$ and $k\rightarrow\infty$ as $n\rightarrow\infty$. Then, we have $\hat{Y}^{\rm BE} \rightarrow {\rm CoVaR}_{\alpha,\beta}$ w.p.1 as $n\rightarrow\infty$. 
\end{theorem}

\begin{proof}
From Lemma \ref{WBC}, for any $\varepsilon>0$ and $y\in\mathcal{Y}$, there exists $M>0$ such that when $m>M$ we have 
\begin{equation}\label{eqn:thm3.2.11}
\Big| {\rm E}[I\{\hat{Y}\leq y \}] - F_{Y|X}(y\,|\, {\rm VaR}_\alpha(X)) \Big| < {\varepsilon\over 2}. 
\end{equation}
Then, when $m>M$, we have 
\begin{eqnarray}
	\lefteqn{ {\rm Pr}\left\{ \left|  \hat{F}_{k}(y) - F_{Y|X}(y\,|\, {\rm VaR}_\alpha(X)) \right| \geq \varepsilon \right\} } \nonumber \\
	& \leq & {\rm Pr}\left\{ \left|  \hat{F}_{k}(y) - {\rm E}[I\{\hat{Y}\leq y \}] \right| + \left| {\rm E}[I\{\hat{Y}\leq y \}] - F_{Y|X}(y\,|\, {\rm VaR}_\alpha(X)) \right| \geq \varepsilon \right\} \nonumber \\
	& \leq & {\rm Pr} \left\{ \left| \hat{F}_{k}(y) - {\rm E}[I\{\hat{Y}\leq y \}] \right|\geq \frac\varepsilon2 \right\} \label{eqn:thm3.2.10} \\
	& \leq & 2e^{-\frac12 k\varepsilon^2}, \label{eqn:thm3.2.1}
\end{eqnarray}
where Equation \eqref{eqn:thm3.2.10} follows Equation \eqref{eqn:thm3.2.11} and Equation \eqref{eqn:thm3.2.1} follows Equation \eqref{eq15}. 

For any small enough $\tilde{\varepsilon}>0$, we have both ${\rm CoVaR}_{\alpha,\beta} -\tilde{\varepsilon}$ and ${\rm CoVaR}_{\alpha,\beta} +\tilde{\varepsilon}$ are in $\mathcal{Y}$. By the definition \eqref{defCoVaR} of ${\rm CoVaR}_{\alpha,\beta}$, we have  
\begin{equation}\label{eq16}
	F_{Y|X}({\rm CoVaR}_{\alpha,\beta} -\tilde{\varepsilon}\,|\, {\rm VaR}_\alpha(X))< \beta <   F_{Y|X}({\rm CoVaR}_{\alpha,\beta} +\tilde{\varepsilon}\,|\, {\rm VaR}_\alpha(X)). 
\end{equation}
Let $\varepsilon_1 = \beta-F_{Y|X}({\rm CoVaR}_{\alpha,\beta} -\tilde{\varepsilon}\,|\, {\rm VaR}_\alpha(X))$, $\varepsilon_2 = F_{Y|X}({\rm CoVaR}_{\alpha,\beta} +\tilde{\varepsilon}\,|\, {\rm VaR}_\alpha(X))-\beta$, and $\varepsilon = \min\{\varepsilon_1, \varepsilon_2\}$, there exists $\tilde{M}>0$ such that when $m>\tilde{M}$, we have Equation \eqref{eqn:thm3.2.1} holds for both $y={\rm CoVaR}_{\alpha,\beta} -\tilde{\varepsilon}$ and $y={\rm CoVaR}_{\alpha,\beta} +\tilde{\varepsilon}$. Notice that $\big| \hat{F}_{k}({\rm CoVaR}_{\alpha,\beta} - \tilde{\varepsilon}) - F_{Y|X}({\rm CoVaR}_{\alpha,\beta} - \tilde{\varepsilon}\,|\, {\rm VaR}_\alpha(X)) \big|<\varepsilon$ implies $\hat{F}_{k}({\rm CoVaR}_{\alpha,\beta} -\tilde{\varepsilon})< \beta$, and $\big| \hat{F}_{k}({\rm CoVaR}_{\alpha,\beta} + \tilde{\varepsilon}) - F_{Y|X}({\rm CoVaR}_{\alpha,\beta} + \tilde{\varepsilon}\,|\, {\rm VaR}_\alpha(X)) \big|<\varepsilon$ implies $\hat{F}_{k}({\rm CoVaR}_{\alpha,\beta} +\tilde{\varepsilon})>\beta$. Then, when $m>\tilde{M}$, we have 
\begin{eqnarray}
	\lefteqn{{\rm Pr} \left\{ \hat{F}_{k}({\rm CoVaR}_{\alpha,\beta} -\tilde{\varepsilon})< \beta <   \hat{F}_{k}({\rm CoVaR}_{\alpha,\beta} +\tilde{\varepsilon}) \right\}}\nonumber \\
	&\geq& {\rm Pr} \left\{ \left\{ \left| \hat{F}_{k}({\rm CoVaR}_{\alpha,\beta} - \tilde{\varepsilon}) - F_{Y|X}({\rm CoVaR}_{\alpha,\beta} - \tilde{\varepsilon}\,|\, {\rm VaR}_\alpha(X)) \right|<\varepsilon \right\}\right. \nonumber \\
	&&~~~~ \left. \cap \
	\left\{ \left| \hat{F}_{k}({\rm CoVaR}_{\alpha,\beta} + \tilde{\varepsilon}) - F_{Y|X}({\rm CoVaR}_{\alpha,\beta} + \tilde{\varepsilon}\,|\, {\rm VaR}_\alpha(X)) \right|<\varepsilon \right\} \right\} \nonumber \\
	&\geq & 1-{\rm Pr}\left\{ \left| \hat{F}_{k}({\rm CoVaR}_{\alpha,\beta} - \tilde{\varepsilon}) - F_{Y|X}({\rm CoVaR}_{\alpha,\beta} - \tilde{\varepsilon}\,|\, {\rm VaR}_\alpha(X)) \right|\geq \varepsilon \right\} \nonumber \\
	& &~~~~ -\ {\rm Pr}\left\{ \left| \hat{F}_{k}({\rm CoVaR}_{\alpha,\beta} + \tilde{\varepsilon}) - F_{Y|X}({\rm CoVaR}_{\alpha,\beta} + \tilde{\varepsilon}\,|\, {\rm VaR}_\alpha(X)) \right|\geq \varepsilon \right\} \label{Bonferroni}\\
	&\geq& 1- 4e^{-\frac12 k \varepsilon^2}, \label{citeeq15} 
\end{eqnarray}
where Equation \eqref{Bonferroni} follows the Bonferroni inequality and Equation \eqref{citeeq15} follows Equation \eqref{eqn:thm3.2.1}. 

Moreover, we have $\hat{F}_{k}({\rm CoVaR}_{\alpha,\beta} -\tilde{\varepsilon})< \beta <   \hat{F}_{k}({\rm CoVaR}_{\alpha,\beta} +\tilde{\varepsilon})$ if and only if ${\rm CoVaR}_{\alpha,\beta} -\tilde{\varepsilon}< \hat{F}^{-1}_{k}(\beta) =\hat{Y}^{\rm BE} <  {\rm CoVaR}_{\alpha,\beta} +\tilde{\varepsilon}$, see Lemma 1.1.4 in \cite{Serfling1980}. Hence, we have 
$$
\sum_{k=1}^\infty {\rm Pr} \left\{ |\hat{Y}^{\rm BE} - {\rm CoVaR}_{\alpha,\beta}| \geq \tilde{\varepsilon} \right\} \leq \sum_{k=1}^\infty  4e^{-\frac12 k \varepsilon^2}  <\infty. 
$$
Therefore, we conclude the theorem  by the Borel-Cantelli Lemma. 
\end{proof}

As pointed out in the Introduction, \cite{Hong09} also applies the batching idea to estimate the quantile sensitivity, which is a conditional expectation instead of a conditional quantile. However, their estimator is only weakly consistent instead of strongly consistent. The strong consistency established by Theorem \ref{thm3} depends critically on the facts that the batching estimator is the inverse of an empirical distribution and the empirical distribution is strongly consistent (i.e., Theorem \ref{consistencyofF}) due to its boundedness.

\subsection{Asymptotic Normality}\label{subsec:BE:nor}

The strong consistency established in Theorem \ref{thm3} neither explains how fast is the convergence nor gives guidelines on how to choose $m$ and $k$. To solve these problems we need to analyze the rate of convergence of the batching estimator and to study its asymptotic distributions. We follow the same analysis framework used in Section \ref{subsec:BE:con}, first analyzing the rates of convergence of the two error terms in Equation \eqref{2parts} and then using the inverse empirical distribution function to derive the asymptotic distribution of the batching estimator $\hat{Y}^{\rm BE}$.

In the following two lemmas, we establish the rates of convergence of the within-batch and across-batch errors in Equation \eqref{2parts} respectively. 

\begin{lemma}\label{RCWB}
	Suppose that Assumption \ref{assu:dist} holds and, there exists $M>0$ such that  $|\frac{\partial}{\partial x}F_{Y|X}(y|{\rm VaR}_\alpha(X))|\leq M$ for all $y\in\mathcal{Y}$ and $|\frac{\partial^2}{\partial x^2}F_{Y|X}(y|x)|\leq M$ for all $(x,y)\in  \mathbb{R} \times\mathcal{Y}$. Then, we have $$\sup_{y\in\mathcal{Y}} \left| {\rm E}\left[I\{\hat{Y}\leq y\} \right]- F_{Y|X}(y\,|\,{\rm VaR}_\alpha(X))\right| = O(m^{-1})$$ as $m\rightarrow\infty$. 
\end{lemma}

\begin{proof}
From Assumption \ref{assu:dist}, by Taylor's expansion, for $y\in \mathcal{Y}$, we have 
\begin{eqnarray*}
	\lefteqn{ F_{Y|X}(y|X_{(\ceil{\alpha m})}) - F_{Y|X}(y|{\rm VaR}_\alpha(X)) } \\
	&=& \frac{\partial}{\partial x}F_{Y|X}(y\,|\,{\rm VaR}_\alpha(X)) \cdot  \left[ X_{(\ceil{\alpha m})} -  {\rm VaR}_\alpha(X) \right] + \frac{\partial^2}{\partial x^2}F_{Y|X}(y|Z)\cdot \left[ X_{(\ceil{\alpha m})} -  {\rm VaR}_\alpha(X) \right]^2, 
\end{eqnarray*}
for some random variable $Z$. By Equation \eqref{eqn:WBC1}, we have $ {\rm E}\big[ I\{\hat{Y}\leq y\} \big] = {\rm E}\left[F_{Y|X}\left( y\, |\, X_{(\ceil{\alpha m})} \right)\right]$. Then, by the assumptions $|\frac{\partial}{\partial x}F_{Y|X}(y|{\rm VaR}_\alpha(X))|\leq M$ for all $y\in\mathcal{Y}$ and $|\frac{\partial^2}{\partial x^2}F_{Y|X}(y|x)|\leq M$ for all $(x,y)\in\mathbb{R}\times\mathcal{Y}$, we have 
\begin{eqnarray}
	\lefteqn{  \left|{\rm E}\left[ I\{\hat{Y}\leq y\} \right]- F_{Y|X}(y\,|\,{\rm VaR}_\alpha(X))\right|\ =\ \left|{\rm E}\left[ F_{Y|X}(y\,|\,X_{(\ceil{\alpha m})}) - F_{Y|X}(y\,|\,{\rm VaR}_\alpha(X)) \right] \right|  }\nonumber\\
	&\leq &  \left|\frac{\partial}{\partial x}F_{Y|X}(y\,|\,{\rm VaR}_\alpha(X))\right| \cdot \left| {\rm E}\left[ X_{(\ceil{\alpha m})} -  {\rm VaR}_\alpha(X) \right]\right| +  {\rm E}\left[\left| \frac{\partial^2}{\partial x^2}F_{Y|X}(y|Z) \right| \cdot| X_{(\ceil{\alpha m})} -  {\rm VaR}_\alpha(X)|^2 \right] \nonumber \\
	&\leq &  M \cdot \left| {\rm E}\left[ X_{(\ceil{\alpha m})} -  {\rm VaR}_\alpha(X) \right]\right| + M\cdot {\rm E}\left[| X_{(\ceil{\alpha m})} -  {\rm VaR}_\alpha(X)|^2 \right]. \nonumber
\end{eqnarray}
Because both ${\rm E}\left[ X_{(\ceil{\alpha m})} -  {\rm VaR}_\alpha(X) \right]$ and ${\rm E}\left[| X_{(\ceil{\alpha m})} -  {\rm VaR}_\alpha(X)|^2 \right]$ are of $O(m^{-1})$ (see Lemma 2 of \cite{Hong09}), so we have  $\sup_{y\in\mathcal{Y}} \left|{\rm E}\left[I\{\hat{Y}\leq y\} \right]- F_{Y|X}(y\,|\,{\rm VaR}_\alpha(X))\right|$  is of $O(m^{-1})$ as well. 
\end{proof}

Lemma \ref{RCWB} shows that the within-batch error converges to $0$ as $m\rightarrow\infty$ uniformly on $y\in\mathcal{Y}$, and the rate of convergence is of order of $m^{-1}$. 

\begin{lemma}\label{AsymAB}
	Suppose that Assumption \ref{assu:dist} holds. Then, we have 
	\begin{equation*}\label{eq222}
		\sup_{t\in\mathbb{R}} \left| {\rm Pr}\left\{ \frac{ \sqrt{k}  }{\hat\sigma(y)}\left( \hat{F}_{k}(y)- {\rm E}\left[I\{\hat{Y}\leq y\} \right] \right) \leq t \right\} -\Phi(t) \right| \leq \frac{33}{4}\cdot \frac{1}{\hat\sigma^3(y) k^{1/2}}, 
	\end{equation*}
	where $\hat\sigma(y) = \sqrt{{\rm Var}(I\{\hat{Y}\leq y \})}$ and $\Phi(t)$ is the cumulative distribution function of the standard normal distribution. 
\end{lemma}

\begin{proof}
From Berry-Ess\'een Theorem \citep{Serfling1980}, we have 
\begin{eqnarray}
	\lefteqn{\sup_{t\in\mathbb{R}} \left| {\rm Pr}\left\{ \frac{ \sqrt{k} }{\hat\sigma(y)} \left( \hat{F}_{k}(y)- {\rm E}\left[I\{\hat{Y}\leq y\} \right] \right) \leq t \right\} -\Phi(t) \right| } \nonumber \\
	& = & \sup_{t\in\mathbb{R}} \left| {\rm Pr}\left\{ \frac{ \sum_{i=1}^k I\{ \hat{Y}_i\leq y \} - {\rm E}[\sum_{i=1}^k I\{ \hat{Y}_i\leq y \}] }{\sqrt{{\rm Var}(\sum_{i=1}^k I\{ \hat{Y}_i\leq y \})}} \leq t \right\} -\Phi(t) \right| \nonumber \\
	& \leq & \frac{33}{4} \frac{\sup_m {\rm E}[| I\{ \hat{Y}\leq y \}- {\rm E}[ I\{ \hat{Y}\leq y \}]|^3]}{\hat\sigma^3(y) k^{1/2}}\ \leq \ \frac{33}{4} \frac{1}{\hat\sigma^3(y) k^{1/2}}. \nonumber
\end{eqnarray}
This concludes the proof of the lemma. 
\end{proof}

Lemma \ref{AsymAB} is developed based on Berry-Ess\'een Theorem. It directly implies that the across-batch error follows an asymptotic normal distribution when scaled by $\sqrt{k}$ and, therefore, its rate of convergence is $k^{-1/2}$. However, Lemma \ref{AsymAB} presents a result that is much stronger than the convergence in distribution. The probability bound established in the lemma is critical in establishing the asymptotic normality and the rate of convergence of the batching estimator $\hat{Y}^{\rm BE}$, stated in following theorem. The proof of the theorem is long and we include it in the appendix.

\begin{theorem}\label{CLTBEE}
	Suppose that Assumption \ref{assu:dist} holds, there exists $M>0$ such that $|\frac{\partial}{\partial x}F_{Y|X}(y|{\rm VaR}_\alpha(X))|\leq M$ for all $y\in\mathcal{Y}$ and $|\frac{\partial^2}{\partial x^2}F_{Y|X}(y|x)|\leq M$ for all  $(x,y)\in\mathbb{R}\times\mathcal{Y}$, and $k\rightarrow\infty$ and $m\rightarrow\infty$ as $n\rightarrow\infty$. When $\sqrt{k}/m \to c$ as $n\to\infty$ for some constant $c\ne 0$, 
	\[\hat{Y}^{\rm BE} - {\rm CoVaR}_{\alpha,\beta} = O_{\rm Pr}\left(n^{-1/3}\right)\]
	as $n\rightarrow\infty$. When $\sqrt{k}/m \to 0$ as $n\to\infty$,
	\begin{equation}\label{eq8}
		\sqrt{k} \left(\hat{Y}^{\rm BE} - {\rm CoVaR}_{\alpha,\beta} \right) \Rightarrow \frac{\sqrt{\beta(1-\beta)}}{f_{Y|X}({\rm CoVaR}_{\alpha,\beta}\, |\,{\rm VaR}_\alpha(X))} \cdot N(0,1)
	\end{equation}
	as $n\rightarrow\infty$. 
\end{theorem}

Theorem \ref{CLTBEE} is an interesting result. First, it shows that the optimal rate of convergence of the batching estimator is $n^{-1/3}$, which is slower than $n^{-1/2}$ of typical quantile estimators. This is because ${\rm CoVaR}_{\alpha,\beta}$ is the $\beta$-quantile of a conditional distribution that conditions on a probability-zero event $\{X={\rm VaR}_\alpha(X)\}$ and ${\rm VaR}_\alpha(X)$ needs to be estimated. Second, it shows that, if $\sqrt{k}/m \to 0$ as $n\to\infty$, the asymptotic normal distribution has mean zero and has the exactly same form as if ${\rm VaR}_\alpha(X)$ is known. This is because, when $\sqrt{k}/m \to 0$ as $n\to\infty$, the bias converges faster than the variance and the bias caused by the ${\rm VaR}_\alpha(X)$ estimator may be ignored.

The asymptotic normal distribution established in Theorem \ref{CLTBEE} is useful in developing a confidence interval of the batching estimator.  Notice that the condition $\sqrt{k}/m \to 0$ as $n\to\infty$ implies that the we may ignore the variation of $X_{(\ceil{\alpha m})}$ and treat it as ${\rm VaR}_\alpha(X)$. Then, by Section 2.6 in \cite{Serfling1980}, we can use the distribution-free approach to build a $100(1-\gamma)\%$ ($0<\gamma<1$) confidence interval $ \left( \hat{Y}_{(\floor{K_1})},\hat{Y}_{(\ceil{K_2})} \right)$ based on two order statistics, where
\begin{equation*}\label{APCI}
	K_1 = k\left( \beta - \frac{z_{1-\gamma/2}[\beta(1-\beta)]^{1/2}}{k^{1/2}}\right), \quad
	K_2 = k\left( \beta + \frac{z_{1-\gamma/2}[\beta(1-\beta)]^{1/2}}{k^{1/2}}\right), 
\end{equation*}
where $z_{1-\gamma/2}$ is the $(1-\gamma/2)$-quantile of the standard normal distribution. The distribution-free confidence interval does not need to estimate the conditional density $f_{Y|X}$ on the right-hand side of Equation \eqref{eq8} and, therefore, is easy to use in practice.

\section{Importance-Sampling Inspired Estimation}\label{sec4}

Financial institutions' holdings are typically not only complicated but also large, e.g., their portfolios have hundreds or even more of financial assets. For such large portfolios, simulating their losses may require re-valuations of a large number of financial assets, e.g., derivatives, and it is known that such simulation may be very time consuming (\citealt{gordy2010nested}, \citealt{Hong2017}). Furthermore, the batching estimator of CoVaR has a slower rate of convergence than the canonical rate of  $n^{-1/2}$, indicating that it may need a large number of simulation observations to achieve a desired accuracy. To obtain a fast estimator of CoVaR, we adopt two ideas. First, we use the delta-gamma approximations to approximate the portfolio losses. The delta-gamma approximation is essentially a second-order Taylor expansion and it is commonly used to approximate portfolio losses to speed up the simulation  in financial risk management (\citealt{hull2012risk} and \citealt{glasserman2004monte}). Second, we propose an importance-sampling inspired estimator to further improve the efficiency of the estimation. Importance sampling has also been used widely in estimation of risk measures  (\citealt{glasserman2000variance}, \citealt{sun2010asymptotic}, \citealt{chu2012confidence}). In most of these works, importance sampling reduces the variances of the estimators but does not improve the rates of the convergence. In this section, however, we show that the IS-inspired estimator of ours not only has a smaller variance but also achieves a better rate of convergence than the batching estimator.  We briefly introduce the delta-gamma approximation in Section \ref{subsec:IS:app}, describe the estimator in Section \ref{ISalg} and then prove its consistency and asymptotic normality in Sections \ref{C&ANofIS} and \ref{subsec:IS:nor}, respectively. 

\subsection{Delta-Gamma Approximation}\label{subsec:IS:app}

The following introduction of the delta-gamma approximation and its simplification is based on Chapter 9 of \cite{glasserman2004monte}. Suppose that we have a portfolio whose value is determined by a vector of risk factors, such as stock prices, commodity prices or index values. The delta-gamma approximation uses the changes of the risk factors to approximate the changes of the portfolio value through a second-order Taylor expansion rooted in It\^o's Lemma (\citealt{hull2012risk}). Let $V(t)$ denote the value of a portfolio at time $t$, and let ${\rm S}(t)= (S_1(t), \ldots, S_d(t))^\top$ denote the values of the $d$ risk factors at time $t$. Then, the delta-gamma approximation approximates $V(\Delta t)$ with a small $\Delta t>0$ by
\[
V(\Delta t) \approx V(0)+\bar\Theta \Delta  t + \bar{\rm \delta}^{\top} {\rm \Delta S}+\frac{1}{2} {\rm \Delta S}^{\top}  \bar{\rm \Gamma} {\rm \Delta S},
\]
where ${\rm \Delta S} = {\rm S}(\Delta t)-{\rm S}(0)$,
\[
\bar\Theta =\left. \frac{\partial V(t)}{\partial t}\right|_{t=0},\quad \bar{\rm \delta}_i =\left. \frac{\partial V(t)}{\partial S_i}\right|_{t=0},\quad {\rm and} \quad \bar{\rm \Gamma}_{ij}=\left. \frac{\partial^2 V(t)}{\partial S_i\partial S_j}\right|_{t=0}
\]
for all $i,j=1,2,\ldots,d$. Then, the loss of the portfolio from time $0$ to $\Delta t$, denoted by $L$, can be approximated by
\begin{equation}\label{eqn:dg_loss}
	L = V(0) - V(\Delta t) \approx  -\bar\Theta \Delta  t - \bar{\rm \delta}^{\top} {\rm \Delta S} - \frac{1}{2} {\rm \Delta S}^{\top}  \bar{\rm \Gamma} {\rm \Delta S}.
\end{equation}

Following \cite{glasserman2000variance}, we assume that ${\rm \Delta S}$ follows a multivariate normal distribution with mean ${\bf 0}$ and covariance matrix $\Sigma$, denoted by ${\rm \Delta S} \sim {\bf N}({\bf 0},\Sigma)$. Let ${\rm \tilde C}$ be any matrix that satisfies ${ \tilde C} { \tilde C}^{\top}=\Sigma$. Notice that ${ \tilde C}$ may be obtained through Cholesky factorization. Then, it is easy to see that $A=-{1\over 2} { \tilde C}^{\top}\, \bar{\rm \Gamma}\, { \tilde C}$ is a symmetric matrix. Then, we can represent $A$ by its eigen-decomposition $A=UBU^{\top}$ where $U$ is a matrix formed by the eigenvectors of $A$ with $UU^{\top}=I$, where $I$ is the $d$-dimensional identity matrix, and $B={\rm diag}(\gamma_1,\ldots,\gamma_d)$, where $\gamma_1,\ldots,\gamma_d$ are the eigenvalues of $A$. Because $A$ is a symmetric matrix, all eigenvalues are real numbers. Let $C=\tilde C U$. It is easy to see that $C C^{\top} ={ \tilde C} { \tilde C}^{\top}=\Sigma$. Let ${\rm Z}=(Z_1,\ldots,Z_d)^\top$ be a vector of independent and identically distributed (i.i.d.) standard normal random variables. Then, it is clear that $C\,{\rm Z}\sim {\bf N}({\bf 0},\Sigma)$ and it has the same distribution as ${\rm \Delta S}$.

By Equation (\ref{eqn:dg_loss}), we may write
\begin{eqnarray*}
	L &\approx& -\bar\Theta \Delta  t - (C^\top \bar{\rm \delta})^{\top} {\rm Z} - \frac{1}{2} {\rm Z}^\top C^\top  \bar{\rm \Gamma} C {\rm Z}\\
	&=& -\bar\Theta \Delta  t - (C^\top \bar{\rm \delta})^{\top} {\rm Z} - \frac{1}{2} {\rm Z}^\top U^\top {\tilde C}^\top  \bar{\rm \Gamma} {\tilde C} U {\rm Z}\\
	&=&  -\bar\Theta \Delta  t - (C^\top \bar{\rm \delta})^{\top} {\rm Z} + {\rm Z}^\top B {\rm Z}.
\end{eqnarray*}
Furthermore, let $c=-\bar\Theta \Delta  t$ and let $(\delta_1,\ldots,\delta_d)^\top = -C^\top \bar{\rm \delta}$. Then, we obtain a much simpler form of the delta-gamma approximation of the loss:
\begin{equation}\label{eqn:dg_simple}
	L \approx c + \sum_{j=1}^d \left(\delta_j Z_j + \gamma_j Z_j^2\right).
\end{equation}
Equation (\ref{eqn:dg_simple}) shows that the randomness of the loss comes from the $d$ standard normal random variables, which may be viewed as the driving force behind the risk factors. When there are multiple portfolios underlying the same risk factors, the parameters $c$, $\delta_j$ and $\gamma_j$ may be different, but these portfolios share the same $Z_1,\ldots,Z_d$. Also, compared to Equation (\ref{eqn:dg_loss}), Equation (\ref{eqn:dg_simple}) is much simpler to simulate and it also sets up a stage for an easier understanding of the importance-sampling scheme that we introduce in next subsection. 

\subsection{The IS Representation and the Estimator}\label{ISalg}

Suppose that we have two portfolios underlying the same $d$ risk factors. Following the delta-gamma approximation introduced in Section \ref{subsec:IS:app}, we may write their losses in the following way (here we assume that the approximations are exact):
\begin{eqnarray}
	X &=& c_1 + \sum_{j=1}^d \left(\delta_{1j} Z_j + \gamma_{1j} Z_j^2\right), \label{eqn:is:X}\\
	Y &=& c_2 + \sum_{j=1}^d \left(\delta_{2j} Z_j + \gamma_{2j} Z_j^2\right). \label{eqn:is:Y}
\end{eqnarray}
Notice that the two portfolio losses are dependent through the same risk factors $Z_1,\ldots,Z_d$. Our goal is to estimate the  ${\rm CoVaR}_{\alpha,\beta}$ that satisfies $ {\rm Pr}\left\{ Y\leq {\rm CoVaR}_{\alpha,\beta} | X= {\rm VaR}_\alpha(X)\right\} = \beta$.

\subsubsection{Representation of the Conditional Probability}

We start by analyzing the conditional probability ${\rm Pr}\{Y \leq y\, |\, X=x\}$. Notice that, by Equation (\ref{defofCP}),
\[
{\rm Pr}\{Y \leq y\, |\, X=x\} =\lim\limits_{\varepsilon\rightarrow 0}\frac{{\rm Pr}\{Y\leq y, |X- x|\leq \varepsilon\}}{{\rm Pr}\{ |X- x|\leq \varepsilon\}}. 
\]
This motivates us to think whether we can use importance sampling to land the majority of the observations, if not all, in the set $\{|X-x|\le\varepsilon\}$. To do that, we consider to change the distribution of $Z_d$ after observing $Z_1,\ldots,Z_{d-1}$. Let $ P_{\varepsilon}={\rm Pr}\left\{|X-x| \leq \varepsilon\, |\, Z_1, \ldots, Z_{d-1}\right\}$. Notice that it is possible to satisfy $\{|X-x|\le\varepsilon\}$ only if $P_\varepsilon>0$ after observing $Z_1,\ldots,Z_{d-1}$. Then, we have $\{|X-x|\le\varepsilon\}=\{|X-x|\le\varepsilon\}\cap \{P_{\varepsilon}>0\}$ w.p.1 and 
\begin{equation}\label{eq13}
	{\rm Pr}\{Y \leq y\, |\, X=x\} = \lim _{\varepsilon \rightarrow 0} \frac{{\rm E}\left[I\{Y \leq y\} \cdot I\{|X-x| \leq \varepsilon\}\cdot I\{P_{\varepsilon} > 0\}\right]}{{\rm E}\left[I\{|X-x| \leq \varepsilon\}\cdot I\{P_{\varepsilon}>0\}\right]}.
\end{equation}

We apply an importance sampling to change the measure of the last dimension $Z_d$. Let $f_d(z)$ denote the density function of $Z_d$ conditional on $Z_1,\ldots,Z_{d-1}$. Because $Z_d$ is independent of $Z_1,\ldots,Z_{d-1}$, $f_d(z)$ is the density of the standard normal random variable. Let $\tilde{f}_{d,\varepsilon}(z)$ denote the importance-sampling distribution of $Z_d$. Conditional on $Z_1,\ldots,Z_{d-1}$, we let
\begin{equation*}
	\tilde{f}_{d,\varepsilon}(z)=
	\begin{cases}
		\frac{f_{d}(z)}{P_{\varepsilon}}\cdot I\{|X-x| \leq \varepsilon\},& \text{if $P_\varepsilon > 0$}\\
		f_{d}(z),& \text{if $P_\varepsilon = 0$}
	\end{cases},
\end{equation*}
where notice that $P_{\varepsilon}$ is a function of $Z_1,\ldots,Z_{d-1}$ and $X$ is a function of $Z_1,\ldots,Z_{d-1},z$. It is easy to verify that $\tilde{f}_{d,\varepsilon}$ is a density function and  all the simulation observations of $X$ will fall in the important region $\{|X-x| \leq \varepsilon\}$ when $P_\varepsilon> 0$. Therefore, by Equation (\ref{eq13}), we have
\begin{equation}\label{changemeasure}
	{\rm Pr}\{Y \leq y\, |\, X=x\} 
	\ =\ \lim _{\varepsilon \rightarrow 0} \frac{\tilde{\rm E}\left[I\{Y \leq y\} \cdot P_{\varepsilon}\cdot I\{P_{\varepsilon}> 0\}\right]}{\tilde{\rm E}\left[P_{\varepsilon}\cdot I\{P_{\varepsilon}> 0\}\right]}
	\ =\ \lim _{\varepsilon \rightarrow 0} \frac{\tilde{\rm E}\left[I\{Y \leq y\} \cdot {P_{\varepsilon}\over 2\varepsilon}\right]}{\tilde{\rm E}\left[{P_{\varepsilon}\over 2\varepsilon}\right]}, 
\end{equation}
where $\tilde{\rm E}$ denotes the expectation under the importance-sampling distribution and the last equality holds because $P_{\varepsilon}\cdot I\{P_{\varepsilon}> 0\}= P_{\varepsilon}$ w.p.1. 

Conditional on $Z_1,\ldots,Z_{d-1}$, by Equation (\ref{eqn:is:X}), we have
\begin{equation}\label{quadraticeq}
	X =\xi_{1}+ \delta_{1d}  Z_{d}+\gamma_{1d} Z_d^2,
\end{equation} 
where $\xi_1=c_1+\sum_{j=1}^{d-1}\big(\delta_{1j}Z_j+\gamma_{1j}Z_j^2\big)$. To study the event $\{|X-x| \leq \varepsilon\}$, we define 
\begin{equation}\label{eqn:is:g}
	g(z)=\xi_{1}+ \delta_{1d}z +\gamma_{1d} z^2.
\end{equation}

\begin{figure}[ht]
	\captionsetup{labelfont=bf}
	\caption{The illustration of $|g(z)-x|\le\varepsilon$ when $x>g^*$}
	\centering    
	\includegraphics[scale = 0.65]{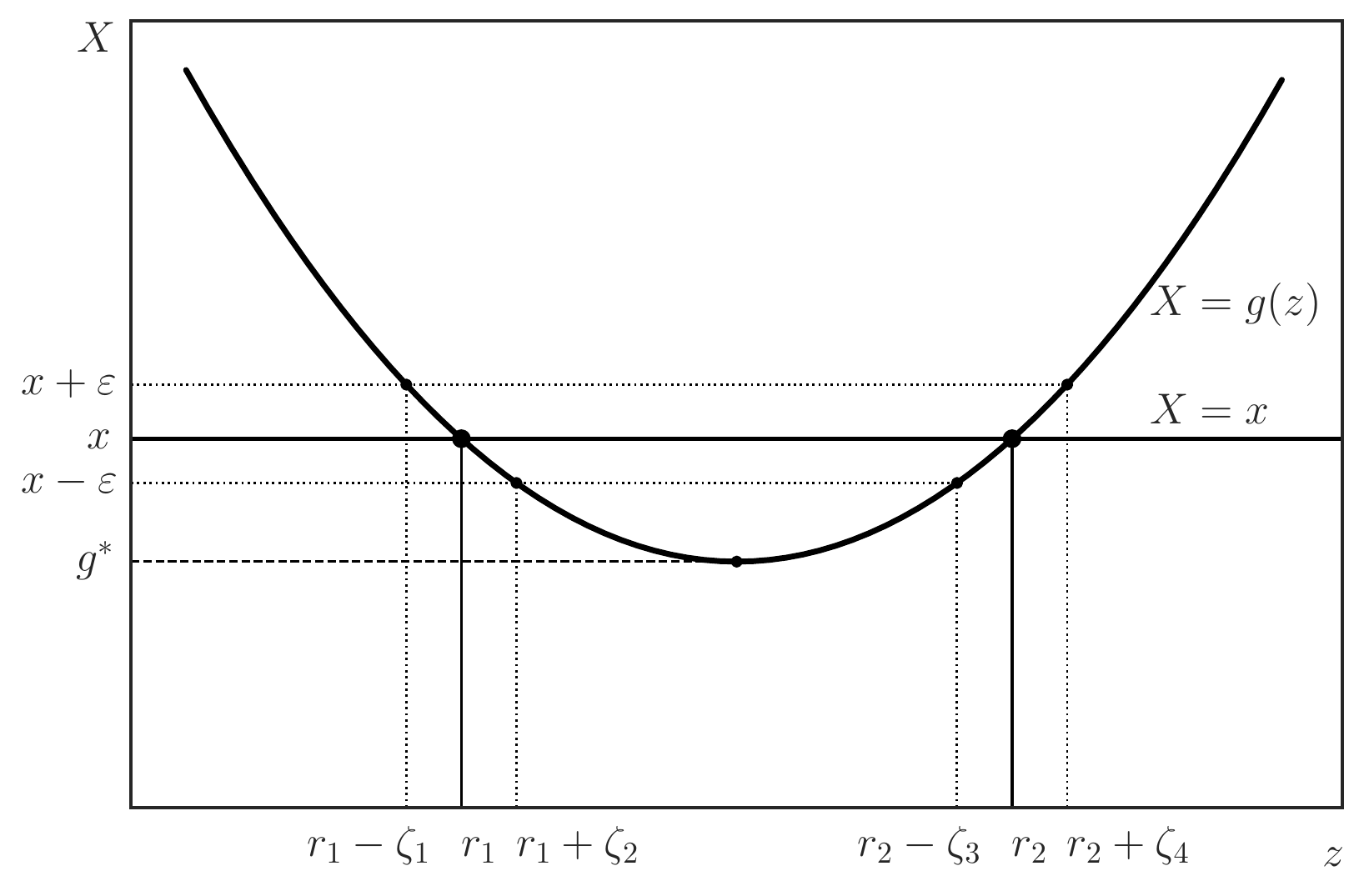}
	\captionsetup{labelfont=bf}
	\label{quadraticfunction}
\end{figure}

Assuming $\gamma_{1d}>0$ (which will be discussed in Remark \ref{remark1}), $g(z)$ is a quadratic function as plotted in Figure \ref{quadraticfunction}. Let 
\[
g^\ast=\min_{z\in\mathbb{R}} g(z)=\xi_1- {\delta_{1d}^2\over 4\gamma_{1d}}
\] 
be the minimum of the function $g$. If $x<g^*$, then there exists no real root that can make $g(z)=x$ and thus $P_\varepsilon=0$ for $\varepsilon$ is small enough. If $x=g^*$, then there exists only one real root. This is a probability-zero event and we may ignore it. If $x>g^*$ (see Figure \ref{quadraticfunction}), there are two real roots to $g(z)=x$ and we denote them as $r_{1}$ and $r_{2}$ with $r_{1} < r_{2}$. Let $\lambda_{1}=g^{\prime}\left(r_{1}\right)$ and $\lambda_{2}=g^{\prime}\left(r_{2}\right)$ be the slopes of the function $g(z)$ at $r_1$ and $r_2$ respectively, where $g^{\prime}(z)=\delta_{1d}+2 \gamma_{1d} z$. Notice that $\lambda_1=-\lambda_2$ and we denote $\lambda=|\lambda_1|=|\lambda_2|>0$. By taking a close look at Figure \ref{quadraticfunction}, we find that $|g(z)-x|\le\varepsilon$ is equivalent to $z\in [r_1-\zeta_1, r_1+\zeta_2]\cup[r_2-\zeta_3, r_2+\zeta_4]$, where $r_1-\zeta_1$, $r_1+\zeta_2$, $r_2-\zeta_3$ and $r_2+\zeta_4$ are the four roots of $g(z)=x+\varepsilon$ and $g(z)=x-\varepsilon$. When $\varepsilon$ is small enough, by Taylor's first-order approximation, $[r_1-\zeta_1, r_1+\zeta_2]\cup[r_2-\zeta_3, r_2+\zeta_4]$ is approximately $\left[r_1-{\varepsilon\over\lambda}, r_1+{\varepsilon\over\lambda}\right]\cup \left[r_2-{\varepsilon\over\lambda}, r_2+{\varepsilon\over\lambda}\right]$. Then, based on this intuition, we have the following lemma on the convergence of $P_\varepsilon/(2\varepsilon)$. 

\begin{lemma}\label{q_derivation}
	Suppose $\gamma_{1d}>0$. Let $q_i=f_d(r_i)/\lambda$ for $i=1,2$ when $x>g^*$, and $q_i=0$ for $i=1,2$ when $x\le g^*$. Then, we have
	\begin{equation}\label{PE/E}
	\lim _{\varepsilon \rightarrow 0}\frac{P_\varepsilon}{2\varepsilon} = q_1+q_2 \quad {\rm w.p.1}.
	\end{equation}
\end{lemma}

\begin{proof}
 Notice that the randomness of both sides of Equation \eqref{PE/E} comes from $Z_1, \ldots, Z_{d-1}$. Because $g^\ast$ is a continuous random variable, so $\{g^\ast=x\}$ is a probability zero event and we can ignore this case. When $x< g^\ast$, $P_\varepsilon=0$ for $\varepsilon$ is small enough, and the right hand side of Equation \eqref{PE/E} is also zero, hence Equation \eqref{PE/E} holds. When $x> g^\ast$, there are two real roots $r_1<r_2$. In this case, $\lambda>0$, and then when $\varepsilon$ is small enough, we have
\begin{eqnarray}
	P_\varepsilon &=& {\rm Pr}\left\{|g(Z_d)-x|\leq \varepsilon\,|\,Z_1, \ldots, Z_{d-1}\right\} \nonumber\\
	&=& {\rm Pr}\left\{ \left|g'(r_1)(Z_d-r_1)+o(|Z_d-r_1|)\right|\leq \varepsilon\,|\,Z_1, \ldots, Z_{d-1}\right\}  \nonumber\\
	&& +\ {\rm Pr}\left\{ \left|g'(r_2)(Z_d-r_2)+o(|Z_d-r_2|)\right|\leq \varepsilon\,|\,Z_1, \ldots, Z_{d-1}\right\}  \nonumber\\
	&=& {\rm Pr}\left\{ \left|Z_d-r_1\right|\leq {\varepsilon\over\lambda}+o(\varepsilon)\,|\,Z_1, \ldots, Z_{d-1}\right\}  \nonumber\\
	&& +\ {\rm Pr}\left\{ \left|Z_d-r_2\right|\leq {\varepsilon\over\lambda}+o(\varepsilon)\,|\,Z_1, \ldots, Z_{d-1}\right\} \nonumber\\
	&=& \left[ \frac{f_d(r_1)}{\lambda}2\varepsilon + \frac{f_d(r_2)}{\lambda}2\varepsilon+o(\varepsilon)\right]. \nonumber
\end{eqnarray}
Hence, Equation \eqref{PE/E} holds. 
\end{proof}

Furthermore, when $x>g^\ast$ and as $\varepsilon\to 0$, $Z_d$ basically only has two choices, $r_1$ and $r_2$, each with probability $q_1/(q_1+q_2)$ and  $q_2/(q_1+q_2)$, respectively. Then, by Equation (\ref{eqn:is:Y}), $Y$ can only take two values $Y_1$ and $Y_2$ conditioned on $Z_1,\ldots,Z_{d-1}$ and $x> g^\ast$, i.e.,
\begin{equation*}\label{computeY}
	Y_{\ell}= c_2 + \sum_{j=1}^{d-1}\left(\delta_{2j} Z_j + \gamma_{2j}Z_j^2\right) +  \delta_{2d} r_\ell + \gamma_{2d} r_\ell^2,~ \ell=1,2, 
\end{equation*}
each with probability $q_1/(q_1+q_2)$ and  $q_2/(q_1+q_2)$, respectively. When $x\leq g^\ast$, we define $Y_\ell=+\infty$, $\ell=1,2$. Then, we have the following theorem, which is the main result of this subsection, and its proof is included in the appendix. 

\begin{theorem}\label{thm:IS}
	Suppose $\gamma_{1d}>0$ and ${\rm E}\left[|g^\ast-x|^{-1/2}\right]<\infty$. Then,
	\begin{equation*}
		{\rm Pr}\{ Y \leq y\, |\, X=x\}  = \frac{{\rm E}\left[I\left\{Y_{1} \leq y\right\} q_1 + I\left\{Y_{2} \leq y\right\} q_2\right]}{{\rm E}\left[q_1+q_2\right]},
	\end{equation*}
	where $\rm E[\cdot]$ is the expectation with respect to $(Z_1,\ldots,Z_{d-1})^\top$.
\end{theorem}

\begin{remark}
    Notice that Theorem \ref{thm:IS} does not follow naturally from Lemma \ref{q_derivation}, because $P_\varepsilon/(2\varepsilon)$ may not be uniformly integrable due to the complication caused by the situation where $x$ is in the neighborhood of $g^\ast$. Therefore, we need a more careful handling of the set $\{|X-x|\le\varepsilon\}$ by breaking it into three pieces, i.e., $\{x> g^\ast +\varepsilon\}$, $\{|x-g^\ast|\le \varepsilon\}$ and $\{x< g^\ast -\varepsilon\}$, and analyze each term individually. The details can be found in the proof of the theorem in the appendix.
\end{remark}

\begin{remark}\label{remark1}
	Notice that we can rearrange the order of $Z_1,\ldots,Z_d$ so that the $d$-th dimension has the highest value of $\gamma_{11},\ldots,\gamma_{1d}$, i.e., $\gamma_{1d}=\max\{\gamma_{11},\ldots,\gamma_{1d}\}$. Then, the assumption $\gamma_{1d}>0$ basically implies at least one of $\gamma_{11},\ldots,\gamma_{1d}$ is positive. Notice that in Equation (\ref{eqn:is:X}), $c_1$ is typically a small deterministic loss and $\delta_{11},\ldots,\delta_{1d}$ are typically zero or very close to zero (due to the delta-hedging strategies). Thus, if $\gamma_{11},\ldots,\gamma_{1d}$ are all negative or zero, the portfolio becomes almost riskless, which contradicts to what we see in practice. Therefore, it is reasonable to assume $\gamma_{1d}>0$. 
	
	There are three more reasons why we let the dimension with the largest $\gamma_{1j}$, $j=1,\ldots,d$, to be the $d$-th dimension. First, $\max\{\gamma_{11},\ldots,\gamma_{1d}\}$ is the largest potential contributor of the loss $X$. Choosing it as the $d$-th dimension typically provides the highest probability to ensure $x> g^\ast$. Second, it is easy to show that $\lambda = 2\sqrt{\gamma_{1d}(x-\xi_1)+\delta_{1d}^2}$. Notice that in our problem $x={\rm VaR}_\alpha(X)$ and it is typically significantly larger than $\xi_1$. Then, a large $\gamma_{1d}$ typically implies a large $\lambda$, which reduces the chance of $\lambda$ near 0 and thus prevents $q_1$ and $q_2$ from blowing up. Third, a larger $\gamma_{1d}$ typically reduces the distance between the two real roots $r_1$ and $r_2$ and thus balance the values of $q_1$ and $q_2$, which prevents one side of the roots from dominating the estimation and reduces the variance of the estimator.
\end{remark}

\begin{remark}
Notice that the random variable $g^\ast-x$ is generalized chi-squared distributed. In fact, 
$$
g^\ast -x \ =\ c_1- x- \sum_{j=1}^d \frac{\delta_{1j}^2}{4\gamma_{1j}} + \sum_{j=1}^{d-1} \gamma_{1j} \tilde{Z}_j^2, 
$$
where $\tilde{Z}_j^2= [ Z_j + \delta_{1j}/(2\gamma_{1j}) ]^2$ is noncentral chi-squared distributed, so $g^\ast-x$ is generalized chi-squared distributed. The generalized chi-squared random variable does not have a simple closed-form probability density function, so it is difficult to analytically evaluate ${\rm E}\left[|g^\ast-x|^{-1/2}\right]$. As a special case, for a chi-squared distributed random variable $\tilde{g}$ with degree of freedom $d-1$, i.e., $\tilde{g}\sim \chi^2(d-1)$, we can directly prove that, when $d>2$, ${\rm E}[\tilde{g}^{-1/2}]= 2^{-1/2}\cdot \Gamma[(d-2)/2]/\Gamma[(d-1)/2]<\infty$, where $\Gamma[\cdot]$ is the Gamma function. Therefore, when $d$ is large, the condition ${\rm E}\left[|g^\ast-x|^{-1/2}\right]<\infty$ is likely to hold. The same argument also applies to the similar conditions used in Lemmas 7 and 9 and Theorems 5, 6 and 7. These conditions are likely to hold when $d$ is large.
\end{remark}

Theorem \ref{thm:IS} is an interesting result. First, it is derived under the importance-sampling distribution, but is ended up in expectations under the original distribution. Hence, no change of measure is needed to compute it. Therefore, we call our estimation approach the ``importance-sampling inspired estimation".  Second, it turns the conditional probability, conditioned on a probability zero event $\{X=x\}$, into the ratio of two unconditional expectations. Suppose there are $n$ observations of $(Z_1,\ldots,Z_{d-1})^\top$ to compute $n$ observations of $Y_1,~Y_2,~q_1,~q_2$, denoted by $Y_{1,k},~Y_{2,k},~q_{1,k},~q_{2,k}$ for $k=1,\ldots,n$. Then, the conditional probability ${\rm Pr}\{Y \leq y\, |\, X=x\}$, i.e., the conditional distribution function $F_{Y|X}(y|x)$, may be estimated by
\begin{equation}\label{eqn:is:dist_est}
	\tilde F_n(y,x) = \frac{{1\over n} \sum_{k=1}^{n} \left( I\{Y_{1,k}\le y\} q_{1,k} + I\{Y_{2,k}\le y\} q_{2,k}\right)}{ {1\over n} \sum_{k=1}^{n} \left( q_{1,k} +  q_{2,k}\right)}
\end{equation}
for all $(x,y)\in\mathbb{R}^2$. It is worthwhile noting that, in Equation (\ref{eqn:is:dist_est}), $Y_{1,k},Y_{2,k},q_{1,k},q_{2,k}$, $k=1,\ldots,n$, are all functions of $x$ as well. Because $\tilde F_n(y,x)$ is a ratio estimator and its rate of convergence is $n^{-1/2}$ \citep{Law2015}, in the following Section \ref{subsubsec:IS_est} we derive a CoVaR estimator that may achieve a rate of convergence of $n^{-1/2}$. 

Furthermore, it is interesting to notice that, even though the derivation of Theorem \ref{thm:IS} depends critically on the property of the quadratic function $g(z)$ introduced by the delta-gamma approximation of $X$, the IS-inspired estimation approach may be applicable to more general situations where a conditional expectation, conditioned on a probability-zero event, needs to be estimated. As long as the condition may be turned into a general equation in the form of $g(z)=x$ and all the roots of the equation may be calculated either through closed-form expressions or through numerical root-finding algorithms, the approach is applicable.

\subsubsection{The IS-Inspired Estimator} \label{subsubsec:IS_est}

Based on Theorem \ref{thm:IS}, we propose the following procedure to estimate ${\rm CoVaR}_{\alpha,\beta}$, where the $d$-th dimension satisfies $\gamma_{1d}=\max\{\gamma_{11},\ldots,\gamma_{1d}\}>0$. Notice that Theorem \ref{thm:IS} holds for all $x\in\mathbb{R}$. Then, we may treat all of $r_1,r_2,Y_1,Y_2,\lambda,q_1,q_2$ as random functions of $x$.
\begin{description}
	\item[Step 1.] Let $n_1$ and $n_2$ be positive integers such that $n_1+n_2=n$. Generate $n_1$ observations of $\left(Z_{1}, \ldots, Z_{d}\right)^{\top}$ to obtain $n_1$ observations of $X$ and estimate ${\rm VaR}_\alpha(X)$ by the order statistic $\tilde{v}_\alpha=X_{(\ceil{\alpha n_1})}$.
	
	\item[Step 2.] Generate $n_2$ observations of $\left(Z_{1}, \ldots, Z_{d-1}\right)^{\top}$. For each observation $\left(Z_{1,k}, \ldots, Z_{d-1,k}\right)^{\top}$, $k=1,\ldots,n_2$, let $\xi_{1,k}= c_1+\sum_{j=1}^{d-1} \big(\delta_{1j} Z_{j,k}+\gamma_{1j}Z_{j,k}^2 \big)$ and $g_k(z)=\xi_{1,k}+\delta_{1d}z+\gamma_{1d}z^2$. If $\tilde{v}_\alpha> g_k^\ast$, i.e., $\tilde{v}_\alpha>\min_{z\in\mathbb{R}} g_k(z)=\xi_{1,k}- \delta_{1d}^2/(4\gamma_{1d})$, let
	\begin{eqnarray*}
		\tilde r_{1,k} &=& r_{1,k}(\tilde{v}_\alpha)\ =\  \frac{1}{2\gamma_{1d}} \left[ -\delta_{1d}- \sqrt{\delta_{1d}^2 +4\gamma_{1d}(\tilde{v}_\alpha-\xi_{1,k})} \right],\\
		\tilde r_{2,k} &=& r_{2,k}(\tilde{v}_\alpha)\ =\   \frac{1}{2\gamma_{1d}} \left[ -\delta_{1d} + \sqrt{\delta_{1d}^2 +4\gamma_{1d}(\tilde{v}_\alpha-\xi_{1,k})} \right],\\
		\tilde Y_{\ell,k} &=& Y_{\ell,k}(\tilde{v}_\alpha)\ =\   c_2 + \sum_{j=1}^{d-1}\Big(\delta_{2j} Z_{j,k} + \gamma_{2j}Z_{j,k}^2\Big) +  \delta_{2d} \cdot \tilde r_{\ell,k} + \gamma_{2d} \cdot \tilde r_{\ell,k}^2, \quad\ell=1,2,\\
		\tilde \lambda_k &=& \lambda_k(\tilde{v}_\alpha) ~ \ =\   \sqrt{\delta_{1d}^2 + 4\gamma_{1d}(\tilde{v}_\alpha-\xi_{1,k})},\\
		\tilde q_{\ell,k} &=& q_{\ell,k}(\tilde{v}_\alpha)\ =\   {1\over\tilde\lambda_k\sqrt{2\pi}} \cdot e^{-{1\over 2}\tilde r_{\ell,k}^2}, \quad\ell=1,2.
	\end{eqnarray*}
	If $\tilde{v}_\alpha\leq g_k^\ast$, i.e., $\tilde{v}_\alpha\leq \xi_{1,k}- \delta_{1d}^2/(4\gamma_{1d})$, let $\tilde Y_{\ell,k}=+\infty$ and $\tilde{q}_{\ell,k}=0$, $\ell=1,2$. 
	
	\item[Step 3.] For each $k=1,\ldots,n_2$, let
	\[
	\tilde w_{\ell,k} = \frac{\tilde q_{\ell,k} }{ \sum_{k=1}^{n_2}\left(\tilde q_{1,k}+\tilde q_{2,k} \right) },\quad \ell=1,2.
	\]
	We organize the observations as follows: 
	\begin{align*}
		(\tilde Y_{1,1}, \tilde w_{1,1}),\quad & (\tilde Y_{2,1}, \tilde w_{2,1}), \\
		(\tilde Y_{1,2}, \tilde w_{1,2}), \quad &(\tilde Y_{2,2}, \tilde w_{2,2}), \\
		\vdots \\
		(\tilde Y_{1,n_2}, \tilde w_{1,n_2}), \quad & (\tilde Y_{2,n_2}, \tilde w_{2,n_2}).
	\end{align*}
	Sort $\tilde Y_{1,1}, ~\tilde Y_{2,1}, ~\tilde Y_{1,2}, ~\tilde Y_{2,2},\ldots, ~\tilde Y_{1,n_{2}}, ~\tilde Y_{2,n_{2}}$ from lowest to highest, denoted by
	$$
	\tilde Y_{(1)} \leq \tilde Y_{(2)} \leq \cdots \leq\tilde  Y_{(2 n_{2})}, 
	$$
	and denote the corresponding $\tilde w$ values by $\tilde w_{(1)}, ~\tilde w_{(2)},\ldots, ~\tilde w_{(2 n_{2})}$. Furthermore, let $\pi_i =\sum_{j=1}^{i} \tilde w_{(j)}$, $i=1, \ldots, 2 n_{2}$. Suppose that there is $m \leq 2 n_{2}$ such that $\pi_{m-1} \leq \beta$ and $\pi_{m}>\beta$. Then, we let $\tilde{Y}^{\rm IS}=\tilde Y_{(m)}$, which is the IS-inspired estimator of ${\rm CoVaR}_{\alpha,\beta}$.
\end{description}

\begin{remark}
	Notice that the conditional distribution function ${\rm Pr}\left\{ Y\le y\,|\,X= {\rm VaR}_\alpha(X)\right\}$ may be approximated by ${\rm Pr}\left\{ Y\le y\,|\,X= \tilde{v}_\alpha\right\}$, where $\tilde{v}_\alpha$ is the {\rm VaR} estimator calculated from $n_1$ observations of $X$. Furthermore, by Equation (\ref{eqn:is:dist_est}), we may estimate ${\rm Pr}\left\{ Y\le y\,|\,X= \tilde{v}_\alpha\right\}$ by 
	\begin{eqnarray}\label{con_prob_est}
		\tilde F_{n_2}(y,\tilde{v}_\alpha) &=& \frac{{1\over n_2} \sum_{k=1}^{n_2} \left( I\{\tilde Y_{1,k}\le y\} \tilde q_{1,k} + I\{\tilde Y_{2,k}\le y\} \tilde q_{2,k}\right)}{ {1\over n_2} \sum_{k=1}^{n_2} \left( \tilde q_{1,k} +  \tilde q_{2,k}\right)}\\
		&=& \sum_{k=1}^{n_2} \left( I\{\tilde Y_{1,k}\le y\}\tilde  w_{1,k} + I\{\tilde Y_{2,k}\le y\} \tilde w_{2,k}\right). \nonumber
	\end{eqnarray}
	Then, $\tilde{Y}^{\rm IS}$ can be viewed as a direct estimator of the inverse function of $\tilde F_{n_2}(y,\tilde{v}_\alpha)$ at $\beta$. 	
\end{remark}

\subsection{Consistency}\label{C&ANofIS}

In this subsection we prove the consistency of the IS-inspired estimator $\tilde{Y}^{\rm IS}$. To analyze the estimator, we start with the the conditional distribution function $F_{Y|X}(y\,|\,{\rm VaR}_\alpha(X))= {\rm Pr}\{ Y\leq y\,|\,X= {\rm VaR}_\alpha(X)\}$ and its estimator $\tilde F_{n_2}(y,\tilde{v}_\alpha)$ given by Equation \eqref{con_prob_est}. Let $Q(y,x)= I\left\{Y_{1} \leq y\right\}  q_{1} + I\left\{Y_{2} \leq y\right\} q_{2}$ and $Q(x)=q_{1} +q_2$, where $Y_1,Y_2,q_1,q_2$ are all random functions of $x$. Furthermore, let $v_\alpha={\rm VaR}_\alpha(X)$. Then, by  Theorem \ref{thm:IS} and Equation \eqref{con_prob_est}, we have
\begin{eqnarray}
	F_{Y|X}(y|v_\alpha) &=& \frac{{\rm E}[Q(y, v_\alpha)]}{ {\rm E}[Q(v_\alpha)]}, \nonumber\\
	\tilde {F}_{n_2}(y, \tilde{v}_\alpha) &=& \frac{\frac1{n_2} \sum_{k=1}^{n_2} Q_k(y, \tilde{v}_\alpha)}{\frac1{n_2} \sum_{k=1}^{n_2} Q_k(\tilde{v}_\alpha)}, \label{eqn:F_est}
\end{eqnarray}
where $Q_k(y,x)$ and $Q_k(y)$, $k=1,\ldots,n_2$, are the observations of $Q(y,x)$ and $Q(y)$. Therefore, to prove that $\tilde {F}_{n_2}(y,\tilde{v}_\alpha)$ converges to $F_{Y|X}(y|v_\alpha)$ as $n\to\infty$, we need to prove the convergence of both the numerator and the denominator. We make the following assumption on $Q(y,x)$ and $Q(x)$.

\begin{assumption}\label{assu:dist1}
	Suppose ${\rm E}[Q(x)]$ and  ${\rm E}[Q(y, x)]$ are twice differentiable functions of $x$ for all $x\in \mathbb{R}$ and $y\in \mathcal{Y}$. 
\end{assumption}

Similar to the proofs in Section \ref{sec3}, we can also divide the estimation errors of the numerator and denominator of Equation (\ref{eqn:F_est}) into two parts: 
\begin{align*}
 \frac{1}{n_2} \sum_{k=1}^{n_2} Q_k(y, \tilde{v}_\alpha) - {\rm E}[Q(y, v_\alpha)]
	\ &=\ \underbrace{ \frac1{n_2} \sum\nolimits_{k=1}^{n_2} Q_k(y, \tilde{v}_\alpha) - {\rm E}[Q(y, \tilde{v}_\alpha) |\tilde{v}_\alpha]}_{\rm error~I}\ +\ \underbrace{\vphantom{\frac1{n_2}}{\rm E}[Q(y, \tilde{v}_\alpha) |\tilde{v}_\alpha ]- {\rm E}[Q(y, v_\alpha)] }_{\rm error~II}, \label{2parts1} \\ 
	\frac{1}{n_2} \sum_{k=1}^{n_2} Q_k(\tilde{v}_\alpha) - {\rm E}[Q(v_\alpha)]\ &=\ \underbrace{ \frac1{n_2} \sum\nolimits_{k=1}^{n_2} Q_k(\tilde{v}_\alpha) - {\rm E}[Q(\tilde{v}_\alpha) |\tilde{v}_\alpha ]}_{\rm error~III}\ +\ \underbrace{\vphantom{\frac1{n_2}}{\rm E}[Q(\tilde{v}_\alpha) |\tilde{v}_\alpha ]- {\rm E}[Q(v_\alpha)] }_{\rm error~IV}.  
\end{align*}
Notice that, $\{Q_k(y,\tilde{v}_\alpha),k=1,\ldots,n_2\}$ are not independent since they all depends on $\tilde{v}_\alpha$. However, conditional on $\tilde{v}_\alpha$, they are independent. Therefore, to study the estimation error of the numerator, we choose the conditional expectation ${\rm E}[Q(y, \tilde{v}_\alpha) |\tilde{v}_\alpha]$ as a bridge. Similarly, we choose ${\rm E}[Q(\tilde{v}_\alpha) |\tilde{v}_\alpha]$ as a bridge when study the estimation error of the denominator. 

In next two lemmas, we prove that the four errors converge to zero. The convergences of the errors II and IV (i.e., Lemma \ref{WBC14}) are based on the continuous mapping theorem \citep{van2000}, and those of the errors I and III (i.e., Lemma \ref{ABC4}) are based on Chebyshev's inequality \citep{durrett2019probability}. The more detailed proofs are included in the appendix. 

\begin{lemma}\label{WBC14}
	Suppose that Assumptions \ref{assu:dist} and \ref{assu:dist1} hold. Then, we have ${\rm E}[Q(y, \tilde{v}_\alpha)|\tilde{v}_\alpha] \rightarrow {\rm E}[Q(y, v_\alpha)]$ and ${\rm E}[Q(\tilde{v}_\alpha)|\tilde{v}_\alpha] \rightarrow {\rm E}[Q(v_\alpha)]$ w.p.1 as $n_1 \rightarrow\infty$. 
\end{lemma}

\begin{lemma}\label{ABC4}
	Suppose that Assumptions \ref{assu:dist} and \ref{assu:dist1} hold and $\sup_{n_1} {\rm E}[|\tilde{v}_\alpha-g^\ast|^{-1}] < \infty$. Then, we have $\frac1{n_2} \sum\nolimits_{k=1}^{n_2} {Q}_k(y, \tilde{v}_\alpha) - {\rm E}[{Q}(y, \tilde{v}_\alpha)| \tilde{v}_\alpha] \rightarrow 0$ in probability and $\frac1{n_2} \sum\nolimits_{k=1}^{n_2} Q_k(\tilde{v}_\alpha) - {\rm E}[Q(\tilde{v}_\alpha)|\tilde{v}_\alpha] \rightarrow 0$ in probability as $n_2 \rightarrow\infty$. 
\end{lemma}

Combining the above two lemmas and by Slutsky's lemma \citep{van2000}, we obtain directly the following theorem on the consistency of $\tilde{F}_{n_2}(y,\tilde{v}_\alpha)$. 

\begin{theorem}\label{consistencyofFbar}
	Suppose that Assumptions \ref{assu:dist} and \ref{assu:dist1} hold, $\sup_{n_1} {\rm E}[|\tilde{v}_\alpha-g^\ast|^{-1}] < \infty$, and $n_1 \rightarrow\infty$ and $n_2 \rightarrow\infty$ as $n\rightarrow\infty$. Then, we have $\tilde{F}_{n_2}(y,\tilde{v}_\alpha) \rightarrow F_{Y|X}(y|v_\alpha)$ in probability as $n\rightarrow\infty$.
\end{theorem}

Theorem \ref{consistencyofFbar} shows that $\tilde{F}_{n_2}(y,\tilde{v}_\alpha)$ is a consistent estimator to $F_{Y|X}(y|v_\alpha)$. In light of this theorem, we can prove the consistency of the IS-inspired estimator $\tilde{Y}^{\rm IS}$. In fact, it is easy to recognize that $\tilde{Y}^{\rm IS}$ satisfies
\begin{equation*}
	\tilde{Y}^{\rm IS}= \inf\{y\in\mathbb{R}: \tilde{F}_{n_2}(y,\tilde{v}_\alpha)\geq \beta\}. 
\end{equation*}
Let $\tilde{F}^{-1}_{n_2}(z,\tilde{v}_\alpha)= \inf\{y\in\mathbb{R}: \tilde{F}_{n_2}(y,\tilde{v}_\alpha)\geq z\}$. Then, we have $\tilde{Y}^{\rm IS}= \tilde{F}^{-1}_{n_2}(\beta,\tilde{v}_\alpha)$. As shown in Equation \eqref{ConQuan}, we have ${\rm CoVaR}_{\alpha, \beta}= F^{-1}_{Y|X}(\beta|v_\alpha)$. Hence, we can take inverse of both $\tilde{F}_{n_2}(y,\tilde{v}_\alpha)$ and $F_{Y|X}(y|v_\alpha)$ in Theorem \ref{consistencyofFbar} to show that $\tilde{Y}^{\rm IS}$ converges to ${\rm CoVaR}_{\alpha,\beta}$. 

\begin{theorem}\label{thm34}
	Suppose that Assumptions \ref{assu:dist} and \ref{assu:dist1} hold, $\sup_{n_1} {\rm E}[|\tilde{v}_\alpha-g^\ast|^{-1}] < \infty$, and $n_1 \rightarrow\infty$ and $n_2 \rightarrow\infty$ as $n\rightarrow\infty$. Then, we have $\tilde{Y}^{\rm IS} \rightarrow {\rm CoVaR}_{\alpha,\beta}$ in probability, as $n\rightarrow\infty$. 
\end{theorem}

\begin{proof}
Similar to the proof between Equations \eqref{eq16} and \eqref{Bonferroni}, by replacing $\hat{F}_k(\cdot)$ by $\tilde{F}_{n_2}(\cdot,\tilde{v}_\alpha)$, replacing $m$ by $n_1$, and replacing $k$ by $n_2$, we have, for any $\tilde{\varepsilon}>0$, 
\begin{equation*}
	{\rm Pr} \left\{ \tilde{F}_{n_2}({\rm CoVaR}_{\alpha,\beta} -\tilde{\varepsilon},\tilde{v}_\alpha)< \beta < \tilde{F}_{n_2}({\rm CoVaR}_{\alpha,\beta} +\tilde{\varepsilon},\tilde{v}_\alpha) \right\} \rightarrow 1 \label{citeeq154}
\end{equation*}
as $n\to\infty$. Moreover, we have $\tilde {F}_{n_2}({\rm CoVaR}_{\alpha,\beta} -\tilde{\varepsilon},\tilde{v}_\alpha)< \beta <   \tilde {F}_{n_2}({\rm CoVaR}_{\alpha,\beta} +\tilde{\varepsilon},\tilde{v}_\alpha)$ if and only if ${\rm CoVaR}_{\alpha,\beta} -\tilde{\varepsilon}< \tilde {F}^{-1}_{n_2}(\beta,\tilde{v}_\alpha) =\tilde{Y}^{\rm IS} <  {\rm CoVaR}_{\alpha,\beta} +\tilde{\varepsilon}$. Then, we have ${\rm Pr} \big\{{\rm CoVaR}_{\alpha,\beta} -\tilde{\varepsilon}< \tilde{Y}^{\rm IS} <  {\rm CoVaR}_{\alpha,\beta} +\tilde{\varepsilon} \big\} \rightarrow 1$ as $n\to\infty$. This concludes the proof of the theorem. 
\end{proof}

Theorem \ref{thm34} shows that $\tilde{Y}^{\rm IS}$ is a consistent estimator of ${\rm CoVaR}_{\alpha,\beta}$. Notice that in the proof of Lemma \ref{ABC4}, we use Chebyshev's inequality to prove that $\frac1{n_2} \sum\nolimits_{k=1}^{n_2} {Q}_i(y, \tilde{v}_\alpha) \rightarrow {\rm E}[{Q}(y, \tilde{v}_\alpha)]$ in probability, and $\frac1{n_2} \sum\nolimits_{k=1}^{n_2} Q_k(\tilde{v}_\alpha) \rightarrow {\rm E}[Q(\tilde{v}_\alpha)]$ in probability as $n\rightarrow\infty$. Therefore, in Theorem \ref{thm34}, we are only able to prove the weak convergence of the IS-inspired estimator $\tilde{Y}^{\rm IS}$, while in Section \ref{subsec:BE:con} we are able to use Hoeffding's inequality to prove the strong consistency of the batching estimator $\hat{Y}^{\rm BE}$.

\subsection{Asymptotic Normality}
\label{subsec:IS:nor}

The consistency established in Theorem \ref{thm34} neither explains how fast is the convergence nor gives guidelines on how to choose $m$ and $k$. To solve these problems we need to analyze the rate of convergence of the IS-inspired estimator and study its asymptotic distribution. In the next two lemmas, we prove  the rates of convergence of the expectation of error I and the normalized error II, respectively. The rate of convergence of the expectation of error I (i.e., Lemma \ref{RCWB4}) is based on the rates of convergence of ${\rm E}\left[ \tilde{v}_\alpha -  v_\alpha \right]$ and ${\rm E}\left[| \tilde{v}_\alpha -  v_\alpha|^2 \right]$ \citep{Hong09}, and the rate of convergence of the normalized error II (i.e., Lemma \ref{AsymAB4}) is based on Berry-Ess\'een Theorem \citep{Serfling1980}. The details of proofs are included in the appendix. 

\begin{lemma}\label{RCWB4}
	Suppose that Assumptions \ref{assu:dist} and \ref{assu:dist1} hold, and there exists $M>0$ such that $|\frac{\partial}{\partial x} {\rm E}[{Q}(y, v_\alpha)]|\leq M$ for all $y\in\mathcal{Y}$ and $|\frac{\partial^2}{\partial x^2} {\rm E}[{Q}(y, x)]|\leq M$ for all $(x,y)\in \mathbb{R}\times\mathcal{Y}$. Then, we have
	$$
	\sup_{y\in\mathcal{Y}} \Big| {\rm E}[Q(y, \tilde{v}_\alpha)]- {\rm E}[Q(y, v_\alpha)] \Big|  = O(n_1^{-1})
	$$
	as $n_1\to\infty$. 
\end{lemma}

Lemma \ref{RCWB4} shows that the expectation of the error I converges to zero uniformly in $y\in\mathcal{Y}$, and the rate of convergence is $n_1^{-1}$. 

\begin{lemma}\label{AsymAB4}
	Suppose that Assumptions \ref{assu:dist} and \ref{assu:dist1} hold and $\sup_{n_1} {\rm E}[|\tilde{v}_\alpha-g^\ast|^{-3/2}] < \infty$. Let $c(\cdot)$ be any deterministic function. Then, we have  
		\begin{align}
			&\left| {\rm Pr}\left\{ \frac{\sqrt{n_2}}{\tilde\sigma(y,\tilde{v}_\alpha)} \Big\{ \frac1{n_2} \sum_{k=1}^{n_2} Q_k(y, \tilde{v}_\alpha) - {\rm E}\big[Q(y, \tilde{v}_\alpha) |\tilde{v}_\alpha \big] \Big\} \leq c(\tilde{v}_\alpha) \right\} - {\rm E}\big[\Phi\big(c(\tilde{v}_\alpha)\big)\big] \right|\ =\ O\left( n_2^{-1/2} \right) \label{eq2}
		\end{align}
as $n_2\to\infty$, where $\tilde\sigma^2(y,\tilde{v}_\alpha) = {\rm Var}\left(Q(y, \tilde{v}_\alpha) | \tilde{v}_\alpha\right)$ and $\Phi(t)$ is the cumulative distribution function of the standard normal distribution. 
\end{lemma}

Notice that $c(\tilde{v}_\alpha)$ may be random since it depends on $\tilde{v}_\alpha$. When the function $c(\cdot)$ is a single-valued function, i.e., its range has only one value, $c(\tilde{v}_\alpha)$ degenerates to a constant, Lemma \ref{AsymAB4} implies that the error II follows an asymptotic normal distribution when scaled by $\sqrt{n_2}$ and, therefore, its rate of convergence is $n_2^{-1/2}$. However, Lemma \ref{AsymAB4} presents a stronger result since it holds for random $c(\tilde{v}_\alpha)$. Combining with the above two lemmas, we can prove the following theorem on the rate of convergence and asymptotic normality of the IS-inspired estimator $\tilde{Y}^{\rm IS}$. The proof of the theorem is long and we include it in the appendix. 

\begin{theorem}\label{CLTBEE4}
	Suppose that Assumptions \ref{assu:dist} and \ref{assu:dist1} hold, there exists $M>0$ such that $|\frac{\partial}{\partial x} {\rm E}[Q(y, v_\alpha) ]|\leq M$ for all $y\in\mathcal{Y}$ and $|\frac{\partial^2}{\partial x^2} {\rm E}[Q(y, x) ]|\leq M$ for all $(x,y)\in\mathbb{R}\times\mathcal{Y}$, $\sup_{n_1} {\rm E}[|\tilde{v}_\alpha-g^\ast|^{-3/2}] < \infty$, $\sigma(y,x)$ is a continuous function of $(x,y)$ in $\mathbb{R}\times\mathcal{Y}$, and $n_1 \rightarrow\infty$ and $n_2 \rightarrow\infty$ as $n\rightarrow\infty$. When $\sqrt{n_2}/n_1\rightarrow c$ as $n\rightarrow\infty$ for some constant $c\neq 0$, 
	$$
	\tilde{Y}^{\rm IS} - {\rm CoVaR}_{\alpha,\beta} = O_{\rm Pr}\left( n^{-1/2} \right)
	$$
	as $n\rightarrow\infty$. When $\sqrt{n_2}/n_1\rightarrow 0$ as $n\rightarrow\infty$, 
	\begin{equation*}\label{eq84}
		\sqrt{n_2} \left(\tilde{Y}^{\rm IS} - {\rm CoVaR}_{\alpha,\beta} \right) \Rightarrow  \frac{\sigma({\rm CoVaR}_{\alpha,\beta},v_\alpha)}{ {\rm E}[Q(v_\alpha)] f_{Y|X}({\rm CoVaR}_{\alpha,\beta} |v_\alpha )} \cdot N(0,1)
	\end{equation*}
	as $n\rightarrow\infty$. 
\end{theorem}

Theorem \ref{CLTBEE4} shows that the IS-inspired estimator $\tilde{Y}^{\rm IS}$ is asymptotic normally distributed and the optimal rate of convergence is $n^{-1/2}$, and the optimal rate is achieved when we use the sample allocation rule $\sqrt{n_2}/{n_1}\to c \neq 0$. This result shows that, by using the IS-inspired estimation approach, we are able to improve the rate of convergence from $n^{-1/3}$ of the batching estimator $\hat{Y}^{\rm BE}$ to the canonical rate of $n^{-1/2}$, significantly improving the large-sample efficiency of the CoVaR estimation. 

The asymptotic normal distribution established in Theorem \ref{CLTBEE4}, under the condition $\sqrt{n_2}/n_1 \to 0$ as $n\to\infty$, is useful in developing a confidence interval of the IS-inspired estimator $\tilde{Y}^{\rm IS}$. Notice that the condition implies that the we may ignore the variation of $\tilde{v}_\alpha$ and treat it as $v_\alpha$. Then, by \cite{Nakayama2014}, we can use the sectioning approach to build a confidence interval. The approach divides the $n_2$ second-stage observations into $b\ge 2$ batches and each batch has $m=n_2/b$ observations, and applies the IS-inspired estimator on each batch with the same first-stage estimated $\tilde{v}_\alpha$, denoted by $\tilde{Y}^{\rm IS}_j,j=1,\ldots,b$. Let
\[
S^2={1\over b-1}\sum_{j=1}^b \left(\tilde{Y}^{\rm IS}_j-\tilde{Y}^{\rm IS} \right)^2.
\]
Then, an approximate $100(1-\gamma)\%$ ($0<\gamma<1$) confidence interval of ${\rm CoVaR}_{\alpha,\beta}$ is
\[
\left(\tilde{Y}^{\rm IS}-t_{b-1,1-\gamma/2}{S\over\sqrt{b}},\ \tilde{Y}^{\rm IS}+t_{b-1,1-\gamma/2}{S\over\sqrt{b}}\right),
\]
where $t_{b-1,1-\gamma/2}$ is the $(1-\gamma/2)$-quantile of the t distribution with $b-1$ degrees of freedom. According to \cite{Nakayama2014}, $b$ is recommended to be chosen from $10\le b\le 30$.

\section{Numerical Study}\label{sec5}

In this section, we study the performances of the batching estimator (BE) and the IS-inspired estimator (ISE) through four examples based on simulated datasets. In the first two examples, we consider two portfolios whose losses have a linear and a nonlinear relation, respectively. We use these two examples to compare the BE and the quantile-regression estimator (QRE) proposed by \cite{CoVaR2016}. In the last two examples, we consider a large portfolio problem with normal and heavy-tailed risk factors respectively. To compare the performance of  BE and ISE, we use the estimated bias and root mean-squared error (RMSE) to compare the point estimators, and use the observed coverage probability and width to compare the confidence intervals. We also study the empirical rates of convergence and compare them to the theoretical results developed in the paper.  All experiments are coded in Python and conducted on a computer with two Intel Xeon Gold 6248R CPUs (each with 24 cores) and 256GB RAM.

\subsection{Linear Portfolio}\label{experiment1}

Suppose there are two portfolios whose losses are denoted as $X$ and $Y$, which are both normally distributed with means $\mu_x$ and $\mu_y$, variances $\sigma_x^2$ and $\sigma_y^2$ and their correlation is $\rho$. Notice that we may write 
\[
Y\ =\ \mu_y+\sigma_y\left(\rho\, {X-\mu_x\over\sigma_x} + \sqrt{1-\rho^2} Z\right),
\]
where $Z$ is a standard normal random variable that is independent of $X$. Therefore, there is a linear relation between $X$ and $Y$. It is easy to derive that 
\begin{equation}\label{eqn:lin_covar}
	{\rm CoVaR}_{\alpha,\beta} \ =\  \mu_y+\sigma_y\left[\rho\Phi^{-1}(\alpha)+\sqrt{1-\rho^2}\Phi^{-1}(\beta)\right], 
\end{equation}
where $\Phi^{-1}(\cdot)$ is the inverse distribution function of the standard normal distribution. Following \cite{CoVaR2016} we set $\mu_x=-0.005$, $\mu_y=-0.00286$, $\sigma_x=0.08$ and $\sigma_y=0.06111$. We calculate the correlation coefficients of the 30 stocks of Dow Jones Industrial Average (DJIA), from 1/1/2020 to 31/6/2021 (to include the large volatility in the US stock market in early 2020), and the results shows that the  correlation coefficients are in the range [-0.81,0.98]. So we take extreme value of $\rho$ into consideration and conduct the experiments with $\rho=-0.95, -0.5, 0.5, 0.95$. Furthermore, we set $\alpha=\beta=0.95$. Notice that when $X$ and $Y$ have a linear relation, the QRE works well. We use this example to understand the performance of the BE when it is compared to the QRE.

The BEs and their confidence intervals are calculated using the procedures developed in Section \ref{sec3}. The QREs and their confidence intervals are calculated using the {\tt quantreg} package in Python. To verify the consistency and asymptotic normality, we increase the sample size from $4.0\times 10^4$ to $3.6 \times 10^5$ and construct the $95\%$ confidence intervals. The biases, the RMSEs and the coverage probabilities are reported in Table \ref{table1}, and all the results are based on 100 independent replications. From these results, we see that the BE is a valid estimator of CoVaR. As the sample size increases, its bias and RMSE both reduce and the coverage probability (CP) of its confidence interval becomes close to the nominal level of 0.95. The results also show that the QRE has better performance than the BE in this example, which is expected because the QRE is developed under the assumption of linear portfolios \citep{CoVaR2016}.

\begin{table}[t]
	\footnotesize
	\centering
	\captionsetup{labelfont=bf}
	\caption{The Comparison between the BE and the QRE in Linear Portfolios}
	\label{table1}
	\begin{tabular}{ccccrrrcrrc}
		\toprule
		\multicolumn{5}{c}{Setting} &   \multicolumn{3}{c}{BE} &  \multicolumn{3}{c}{QRE} \\ \cmidrule{1-5}  \cmidrule(l){6-8} \cmidrule(l){9-11} 
		$\rho$ & $k$ & $m$ & $n$ & ${\rm CoVaR}$ &  \multicolumn{1}{c}{Bias} & \multicolumn{1}{c}{RMSE} & CP & \multicolumn{1}{c}{Bias} & \multicolumn{1}{c}{RMSE} & CP \\ \midrule
		& 200   & 200   & $4.0\times10 ^4$ &          &    $4.23\times 10 ^{-3}$ & $5.18\times 10 ^{-3}$ & 0.68  &  $1.10\times 10 ^{-4}$ & $7.96\times 10 ^{-4}$ & 0.66  \\
		-0.95      & 400   & 400   & $1.6\times 10 ^5$  &-0.0670& $2.02\times 10 ^{-3}$ & $2.98\times 10 ^{-3}$ & 0.84     & $-4.72\times 10 ^{-5}$ & $3.45\times 10 ^{-4}$ & 0.73  \\
		& 600   & 600   & $3.6\times 10 ^5$ & & $1.21\times 10 ^{-3}$ & $2.13\times 10 ^{-3}$ & 0.84  &  $-3.97\times 10 ^{-5}$ & $2.37\times 10 ^{-4}$ & 0.71  \\
		\midrule
		& 200   & 200   & $4.0\times 10 ^4$ &        & $1.28\times 10 ^{-3}$ & $7.67\times 10 ^{-3}$ & 0.98  & $-5.19\times 10 ^{-5}$ & $1.30\times 10 ^{-3}$ & 0.89  \\
		-0.50     & 400   & 400   & $1.6\times 10 ^5$       & 0.0339  & $-4.65\times 10 ^{-4}$ & $5.85\times 10 ^{-3}$ & 0.97  & $-3.13\times 10 ^{-5}$ & $5.51\times 10 ^{-4}$ & 0.92  \\
		& 600   & 600   & $3.6\times 10 ^5$ &  & $-3.02\times 10 ^{-4}$ & $4.52\times 10 ^{-3}$ & 0.98      & $-3.08\times 10 ^{-5}$ & $3.54\times 10 ^{-4}$ & 0.96  \\
		\midrule
		& 200   & 200   & $4.0\times 10 ^4$  &       & $-6.52\times 10 ^{-4}$ & $7.31\times 10 ^{-3}$ & 1.00  & $6.15\times 10 ^{-5}$ & $1.24\times 10 ^{-3}$ & 0.91  \\
		0.50      & 400   & 400   & $1.6\times 10 ^5$        & 0.1344  & $-1.15\times 10 ^{-3}$ & $5.92\times 10 ^{-3}$ & 0.99  & $1.17\times 10 ^{-5}$ & $5.73\times 10 ^{-4}$ & 0.97  \\
		& 600   & 600   & $3.6\times 10 ^5$ & & $-1.17\times 10 ^{-3}$ & $4.67\times 10 ^{-3}$ & 0.95  &  $-1.07\times 10 ^{-7}$ & $3.69\times 10 ^{-4}$ & 0.94  \\
		\midrule
		& 200   & 200   & $4.0\times 10 ^4$&       & $1.27\times 10 ^{-3}$ & $2.85\times 10 ^{-3}$ & 0.96 & $1.08\times 10 ^{-4}$ & $7.30\times 10 ^{-4}$ & 0.66  \\
		0.95    & 400   & 400   & $1.6\times 10 ^5$       & 0.1240 & $2.66\times 10 ^{-4}$ & $2.34\times 10 ^{-3}$ & 0.95  & $3.32\times 10 ^{-5}$ & $3.68\times 10 ^{-4}$ & 0.65  \\
		& 600   & 600   & $3.6\times 10 ^5$ &  & $2.33\times 10 ^{-4}$ & $1.76\times 10 ^{-3}$ & 0.94  & $1.90\times 10 ^{-5}$ & $2.52\times 10 ^{-4}$ & 0.66\\
		\bottomrule
	\end{tabular}
	\label{tab:addlabel}
\end{table}

\subsection{Nonlinear Portfolio}\label{experiment2}

When there are derivatives in the portfolios, their losses in general have a nonlinear relationship. Thus, we consider a simple nonlinear example where there is a delta-gamma approximation in the loss of the second portfolio. Suppose there are two portfolios whose losses are denoted as $X$ and $Y$, where $X$ is normally distributed with mean $\mu_x$ and variance $\sigma_x^2$ and $Y$ is consist of a quadratic form of $X$ and a mean-zero normal random variable $\xi$ with variance $\sigma_y^2$ and ${\rm corr}(X, \xi)=\rho$, i.e.,
\begin{eqnarray*}
	Y &=& \delta X + \frac{1}{2}\gamma X^2 + \xi,\\
	\xi &=& \sigma_y\left(\rho \frac{X-\mu_x}{\sigma_x}+\sqrt{1-\rho^2}Z\right),
\end{eqnarray*}
where $Z$ is a standard normal random variable independent of $X$.
Therefore, there is a nonlinear relation between $X$ and $Y$. We can furthermore derive that 
\begin{eqnarray}
	{\rm CoVaR_{\alpha, \beta}}&=&\delta{\rm VaR_{\alpha}}(X) + \frac{1}{2}\gamma{\rm VaR_{\alpha}}(X)^2+\sigma_y\left[\rho\Phi^{-1}(\alpha)+\sqrt{1-\rho^2}\Phi^{-1}(\beta)\right], \label{eqn:nonlin_covar}\\
	{\rm VaR_{\alpha}}(X) &=& \mu_x +\Phi^{-1}(\alpha)\sigma_x,\nonumber
\end{eqnarray}
where $\Phi^{-1}(\cdot)$ is the inverse distribution function of standard normal distribution. We set $\mu_x = -0.03$, $\sigma_x = 0.2$, $\sigma_y = 0.3$, $\delta = 0.2$ and $\gamma = 0.8$ respectively. Again, we estimate ${\rm CoVaR}_{0.95,0.95}$ based on 100 replications with $\rho = -0.95, -0.5, 0.5, 0.95$.
Except for the loss model, the procedures to conduct the numerical experiments in this subsection are same as that of Section \ref{experiment1}.  The observed biases, the RMSEs and the coverage probabilities of $95\%$ confidence intervals are reported in Table \ref{table2}. 

Compared with Table \ref{table1} of the linear portfolios, we see that the BE continues to deliver good performance. However, the QRE, which assumes linear portfolios, has a significant bias that cannot be reduced by increasing the sample size, and the bias causes the confidence intervals to have a nearly zero coverage probability, thus missing the true value entirely. This example demonstrates the advantage of the BE with respect to the QRE. In practice, because systemic risks are in general measured at institution level, portfolios typically include complicated derivative products and thus display nonlinear relationships. In such situations, the BE avoids the model error and may deliver better performance than the QRE. 

\begin{table}[t]
	\footnotesize
	\centering
	\captionsetup{labelfont=bf}
	\caption{The Comparison between the BE and the QRE in Noninear Portfolios}
	\label{table2}
	\begin{tabular}{ccccrrrcrrc}
		\toprule
		\multicolumn{5}{c}{Setting} &   \multicolumn{3}{c}{BE} &  \multicolumn{3}{c}{QRE} \\ \cmidrule{1-5}  \cmidrule(l){6-8} \cmidrule(l){9-11} 
		$\rho$ & $k$ & $m$ & $n$ & ${\rm CoVaR}$ &  \multicolumn{1}{c}{Bias} & \multicolumn{1}{c}{RMSE} & CP & \multicolumn{1}{c}{Bias} & \multicolumn{1}{c}{RMSE} & CP \\ \midrule
		& 200 & 200 & $4.0\times 10^4$ &   & $1.22\times10 ^{-2}$ & $1.86\times10 ^{-2}$ & 0.91 & $-2.28\times10 ^{-2}$ & $2.30\times10 ^{-2}$ & 0.00 \\
		-0.95 & 400 & 400 & $1.6\times10^5$ & -0.2192  & $5.02\times10 ^{-3}$ & $1.13\times10 ^{-2}$ & 0.92 & $-2.22\times10 ^{-2}$ & $2.23\times10 ^{-2}$ & 0.00 \\
		& 600 & 600 & $3.6\times10^5$ && $2.71\times10 ^{-3}$ & $8.99\times10 ^{-3}$ & 0.95 & $-2.22\times10 ^{-2}$ & $2.23\times10 ^{-2}$ & 0.00 \\ \midrule
		& 200 & 200 & $4.0\times10^4$ & & $3.22\times10 ^{-3}$ & $3.71\times10 ^{-2}$ & 0.99 & $-2.56\times10 ^{-2}$ & $2.63\times10 ^{-2}$ & 0.00 \\
		-0.50 & 400 & 400 & $1.6\times10^5$ & 0.2762 & $-3.98\times10 ^{-3}$ & $2.89\times10 ^{-2}$ & 1.00 & $-2.55\times10 ^{-2}$ & $2.57\times10 ^{-2}$ & 0.00 \\
		& 600 & 600 & $3.6\times10^5$ & & $-2.79\times10 ^{-3}$ & $2.22\times10 ^{-2}$ & 0.98 & $-2.56\times10 ^{-2}$ & $2.56\times10 ^{-2}$ & 0.00 \\ \midrule
		& 200 & 200 & $4.0\times10^4$ &  & $-2.81\times10 ^{-3}$ & $3.58\times10 ^{-2}$ & 0.99 & $-2.51\times10 ^{-2}$ & $2.58\times10 ^{-2}$ & 0.01 \\
		0.50 & 400 & 400 & $1.6\times10^5$ & 0.7696 & $-6.47\times10 ^{-3}$ & $2.94\times10 ^{-2}$ & 0.97& $-2.53\times10 ^{-2}$ & $2.55\times10 ^{-2}$ & 0.00 \\
		& 600 & 600 & $3.6\times10^5$ & & $-5.94\times10 ^{-3}$ & $2.32\times10 ^{-2}$ & 0.94   & $-2.54\times10 ^{-2}$ & $2.55\times10 ^{-2}$ & 0.00 \\ \midrule
		& 200 & 200 & $4.0\times10^4$ & & $1.45\times10 ^{-2}$ & $1.99\times10 ^{-2}$ & 0.83 & $-2.17\times10 ^{-2}$ & $2.21\times10 ^{-2}$ & 0.00 \\
		0.95 & 400 & 400 & $1.6\times10^5$ & 0.7184 & $5.30\times10 ^{-3}$ & $1.26\times10 ^{-2}$ & 0.89 & $-2.18\times10 ^{-2}$ & $2.19\times10 ^{-2}$ & 0.00 \\
		& 600 & 600 & $3.6\times10^5$ & & $4.17\times10 ^{-3}$ & $1.01\times10 ^{-2}$ & 0.90 & $-2.20\times10 ^{-2}$ & $2.20\times10 ^{-2}$ & 0.00 \\ 
		\bottomrule
	\end{tabular}
	\label{tab:addlabel}
\end{table}

\begin{figure}[ht]
	\captionsetup{labelfont=bf}
	\caption{CoVaR with respect to $\rho$ in both Linear and Nonlinear Portfolios}
	\centering
	\includegraphics[scale = 0.45]{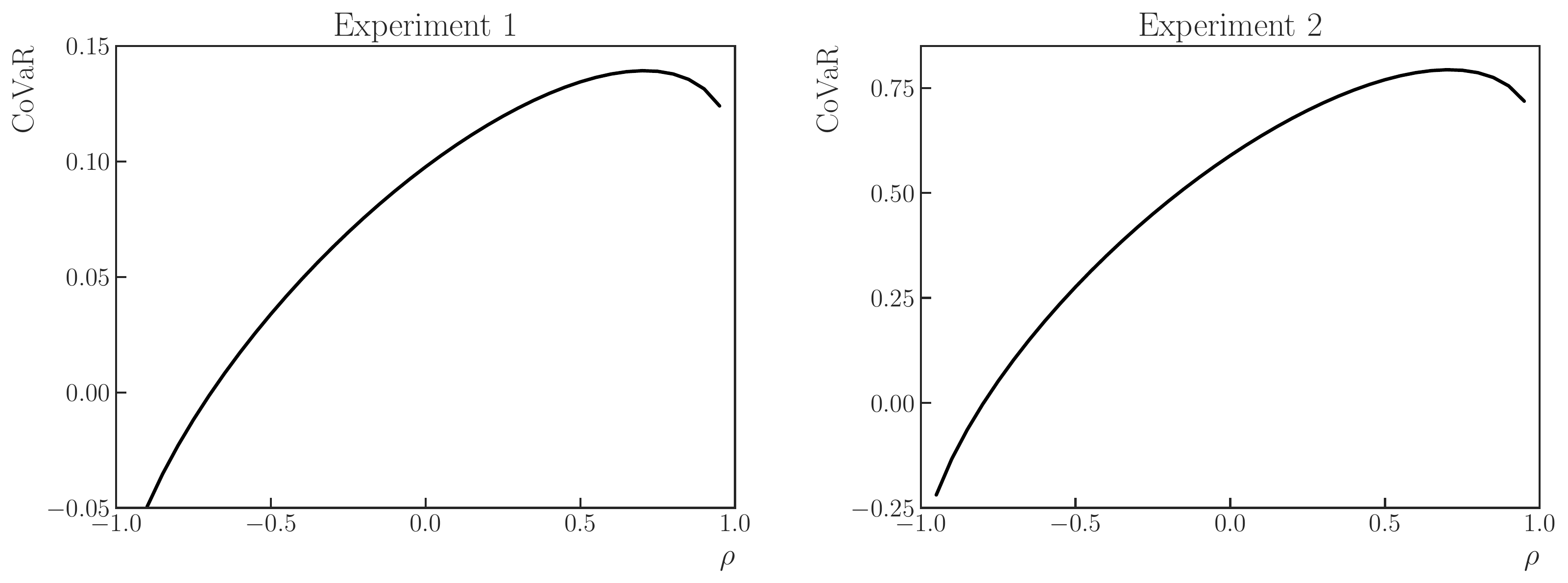}
	\captionsetup{labelfont=bf}
	\label{fig:cor}
\end{figure}
Another interesting observation from both the linear and nonlinear portfolios is that the CoVaR is not a monotone increasing function of the coefficient of correlation $\rho$, which may appear a bit counter-intuitive. In Figure \ref{fig:cor} we plot the CoVaR curves with respect to $\rho$ for both portfolios. Notice that the both curves have the same shape. This is because in both Equations \eqref{eqn:lin_covar} and \eqref{eqn:nonlin_covar}, the shapes of CoVaR with respect to $\rho$ are determined by $\rho\Phi^{-1}(\alpha)+\sqrt{1-\rho^2}\Phi^{-1}(\beta)$, which has the same shape and reaches the optimal at \[\rho^*=\sqrt{\Phi^{-1}(\alpha)^2\over \Phi^{-1}(\alpha)^2+\Phi^{-1}(\beta)^2}.\] Because $\alpha=\beta$ in both examples, therefore, we observe that $\rho^*=1/\sqrt{2}\approx 0.707$.

\subsection{Large Portfolio}\label{experiment3}

In this example we consider two large portfolios whose losses are represented by the delta-gamma approximations in Section \ref{subsec:IS:app}. The portfolios have 50 correlated underlying risk factors, denoted by ${\rm \Delta S}$, and we assume that ${\rm \Delta S}$ follows a multivariate normal distribution with mean vector $\bm 0$ and covariance matrix $\Sigma$. We further simplify them into quadratic forms of independent standard normal random variables according to the approach presented in Section \ref{subsec:IS:app}, and present the detailed parameters of the simplified models in the Appendix \ref{ex3_parameters}. We rearranged the order of $Z_1,\ldots, Z_d$ so that the $d$-th dimension (i.e., the dimension that we conduct importance sampling) has the highest value of $\gamma_{11},\ldots, \gamma_{1d}$, i.e., $\gamma_{1d}=\max\{\gamma_{11},\ldots, \gamma_{1d}\}$. The reason to choose such a dimension has been explained in Remark \ref{remark1}. We set $\alpha=\beta=0.95$. To understand the performance of different estimators, we need the true value of the CoVaR. We compute it by using a very large sample size $n=1.0\times 10^8$ via the IS-inspired estimation and its value is 0.6167.

Before comparing the performance of the BE and ISE, we first highlight the sample-allocation rules. Notice that Sections \ref{subsec:BE:nor} and \ref{subsec:IS:nor} suggest to set $k=n^{2/3-\delta}$ and $m=n/k$ with $\delta\in(0,2/3)$ for the BEs and to set $n_1=n^{1-\delta}$ and $n_2=n-n_1$ with $\delta\in(0,1/2)$ for the ISEs. We use a total sample size of $100,000$ to compare the different sample-allocation rules for these two estimators and report the results in Tables \ref{sample_allocation_BE} and \ref{sample_allocation_IS}, where the nominal coverage probabilities of the confidence intervals are $95\%$. In the experiments of this section, we set the batches of sectioning as $b=10$ to construct the confidence interval of the ISE.

\begin{table}[th]
	\footnotesize
	\captionsetup{labelfont=bf}
	\centering
	\caption{The Performance of the BE under Different Sample-Allocation Rules with $n=100,000$}
	\label{sample_allocation_BE}
	\begin{tabular}{ccrrrcr}
		\toprule
		\multicolumn{1}{c}{$k$} & \multicolumn{1}{c}{$m$} & \multicolumn{1}{c}{bias} & \multicolumn{1}{c}{SD} & \multicolumn{1}{c}{RMSE} & \multicolumn{1}{c}{CP} & \multicolumn{1}{c}{width} \\
		\midrule
		2,000 & 50 & \num{2.64E-02}&\num{9.81E-03}&\num{	2.81E-02}&0.30&\num{4.55E-02}\\
		1,250 & 80 & \num{-2.98E-04}&\num{1.24E-02}&\num{1.24E-02}&0.93&\num{5.58E-02}\\
		1,000 & 100 &\num{-5.87E-04}&\num{1.27E-02}&\num{1.27E-02}&1.00&\num{6.26E-02}\\
		800 & 125 &\num{-4.20E-03}&\num{1.64E-02}&\num{1.70E-02	}&0.99&\num{7.37E-02} \\
		500 & 200 & \num{-2.70E-03}&\num{2.10E-02}&\num{2.11E-02}&0.95&\num{9.34E-02} \\
		400 & 250 &\num{3.56E-03}&\num{2.09E-02}&\num{2.12E-02}&0.96&\num{1.04E-01} \\
		250 & 400 & \num{-6.27E-04}&\num{3.13E-02}&\num{3.13E-02}& 0.96 &\num{	1.44E-01} \\
		200 & 500 &\num{-2.70E-03}&\num{3.31E-02}&\num{	3.32E-02}&0.97&\num{1.82E-01} \\
		125 & 800 & \num{-4.06E-03}&\num{3.83E-02}&\num{3.85E-02}&0.99 &\num{2.28E-01} \\
		100 & 1,000 & \num{-1.54E-02}&\num{	4.48E-02}&\num{4.74E-02}&1.00&\num{ 	3.18E-01} \\
		80 & 1,250 & \num{-1.80E-02}&\num{4.92E-02}&\num{5.24E-02}&0.97&\num{2.86E-01} \\
		50 & 2,000 & \num{-1.11E-02}&\num{7.01E-02}&\num{7.10E-02}&0.92&\num{ 	2.78E-01}\\
		\bottomrule
	\end{tabular}
\end{table}

\begin{table}[th]
	\footnotesize
	\captionsetup{labelfont=bf}
	\centering
	\caption{The Performance of the ISE under Different Sample-Allocation Rules with $n=100,000$}
	\label{sample_allocation_IS}
	\begin{tabular}{ccrrrcr}
		\toprule
		\multicolumn{1}{c}{$n_1$} & \multicolumn{1}{c}{$n_2$} & \multicolumn{1}{c}{bias} & \multicolumn{1}{c}{SD} & \multicolumn{1}{c}{RMSE} & \multicolumn{1}{c}{CP} & \multicolumn{1}{c}{width} \\
		\midrule
		10,000  & 90,000  & \num{3.10E-04}&\num{5.93E-03}&\num{	5.94E-03}&0.49 &\num{8.55E-03}\\
		20,000  & 80,000  & \num{1.47E-04}&\num{3.72E-03}&\num{3.72E-03}&0.74 &\num{8.73E-03}\\
		30,000  & 70,000  &\num{3.19E-04}&\num{4.00E-03}&\num{4.01E-03}&0.79 &\num{9.15E-03}\\
		40,000  & 60,000  &\num{-1.12E-05}&\num{3.39E-03}&\num{3.39E-03	}&0.80 &\num{9.66E-03} \\
		50,000  & 50,000  & \num{1.64E-04}&\num{3.64E-03}&\num{3.65E-03}&0.88 &\num{1.16E-02} \\
		60,000  & 40,000  &\num{2.05E-04}&\num{3.41E-03}&\num{3.42E-03}&0.89 &\num{1.27E-02} \\
		70,000  & 30,000 & \num{9.35E-05}&\num{3.72E-03}&\num{3.72E-03}& 0.96  &\num{1.49E-02} \\
		80,000  & 20,000  &\num{6.68E-05}&\num{4.08E-03}&\num{	4.08E-03}&0.96&\num{1.76E-02} \\
		90,000  & 10,000  & \num{-1.17E-03}&\num{5.54E-03}&\num{5.66E-03}&0.93  &\num{2.43E-02} \\
		\bottomrule
	\end{tabular}
\end{table}

As shown in Table \ref{sample_allocation_BE} for the BE, we observe that the standard deviation increases as $m$ grows larger and becomes the dominant part of the RMSE, which is consistent with the convergence analysis in Section \ref{subsec:BE:nor}. In this example, we observe that the BE performs well when $k$ and $m$ are approximately  $n^{2/3}/2$ and $n/k$. In the rest of this section we use this rule to compute the BE and compare it to the ISE.
	
As shown in Table \ref{sample_allocation_IS} for the ISE, we observe that the bias is always much smaller than the standard deviation under different allocations. If one wants to minimize the RMSE, it is suggested to choose $n_1$ and $n_2$ that are close to each other. However, if one wants to deliver accurate confidence intervals, it is suggested to set $n_1$ large and $n_2$ small. In the rest of this section we use $n_1=n_2=n/2$. Furthermore, we find that we do not lose any samples in the IS-inspired estimation scheme, because the way we choose the $d$-th dimension appears to ensure that $\tilde{v}_\alpha> g_k^\ast$ for all the samples in this example (see Section \ref{subsubsec:IS_est} for the details).

Following the sample-allocation rules above, we increase the sample size from $1\times 10 ^3$ to $1\times 10 ^6$ and report the performance of the estimators and confidence intervals and the total running time (seconds) in Tables \ref{BE_large} and  \ref{IS_large} for the BEs and ISEs, respectively. First, we plot the $\log{\rm RMSE}$ with respect to $\log n$ for both the BE and ISE in Figure \ref{convergencerate} to understand the rates of convergences of the two estimators. The plots show that the empirical rates of convergence of the BE and ISE are $n^{-1/3}$ and $n^{-1/2}$, respectively, which are both consistent with the theoretical rates of convergence developed in Sections \ref{subsec:BE:nor} and \ref{subsec:IS:nor} and demonstrate the advantages of the ISEs for large portfolio problems.

Second, the results of Tables \ref{BE_large} and  \ref{IS_large} shows that the biases and RMSEs decrease drastically for the ISE. With 50,000 samples, the RMSE of the ISE is below $1\%$ of actual CoVaR. Comparing to the BE, the ISE has a smaller bias and RMSE and its confidence interval is also narrower with a fixed sample size $n$. When taking the computation efficiency into consideration, we can see that with similar time budget, the ISE also outperforms the BE. These results suggest that the ISE is a better estimator than the BE for large portfolio problems.

\begin{figure}[t]
	\captionsetup{labelfont=bf}
	\caption{The Rate of Convergence of the BE and the ISE}
	\centering
	\includegraphics[scale = 0.5]{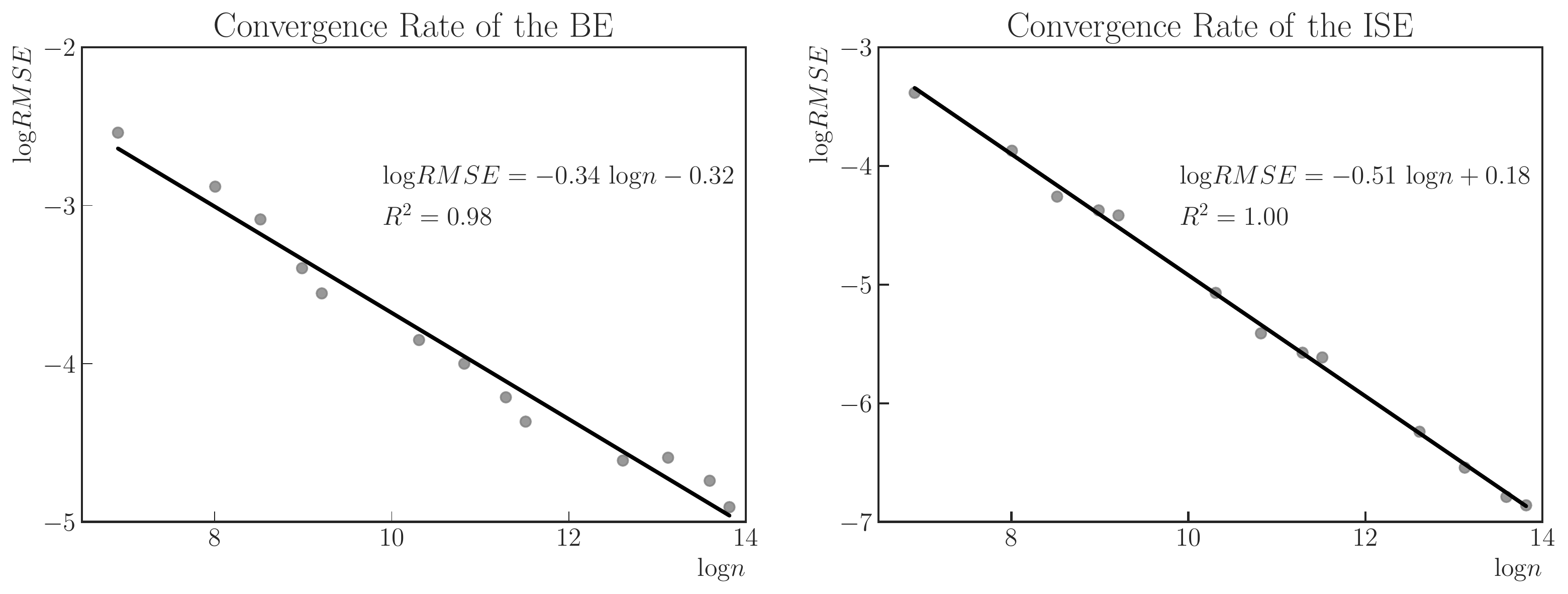}
	\centering
	\captionsetup{justification=centering}
	\captionsetup{labelfont=bf}
	\label{convergencerate}
\end{figure}

\begin{table}[t]
	\footnotesize
	\captionsetup{labelfont=bf}
	\centering
	\caption{ The Performance of the BE in Large Portfolios}
	\label{BE_large}
	
	\begin{tabular}{cccrrrcrc}
		\toprule
		\multicolumn{1}{c}{$n$} & \multicolumn{1}{c}{$k$} & \multicolumn{1}{c}{$m$} & \multicolumn{1}{c}{bias} & \multicolumn{1}{c}{SD} & \multicolumn{1}{c}{RMSE} & \multicolumn{1}{c}{CP} & \multicolumn{1}{c}{width} & \multicolumn{1}{c}{time} \\
		\midrule
		\num{1.E+03}&	50 &	20 	&\num{-4.94E-03} &\num{	7.87E-02} &\num{7.89E-02} & 0.89 &\num{2.78E-01} &\num{1.47 } \\
		\num{3.E+03} &	100 &	30 &\num{2.99E-02} &\num{4.74E-02} &\num{5.61E-02} & 0.90 &\num{3.18E-01} &\num{1.75 } \\
		\num{5.E+03} & 	125 &	40 &\num{-2.44E-03} &\num{4.55E-02} &\num{4.56E-02} &0.99 &\num{2.51E-01} & 	1.87  \\
		\num{8.E+03} & 200 & 40 & \num{-3.96E-03} &\num{3.32E-02} &\num{3.35E-02} &0.99&\num{	1.86E-01} & 1.83 \\
		\num{1.E+04} & 	250 & 40 &\num{3.73E-03} &\num{	2.83E-02} &\num{2.86E-02}&0.99&\num{1.58E-01}&	1.83  \\
		\num{3.E+04} & 	500 &	60 &\num{	1.26E-04}&\num{2.13E-02}&\num{	2.13E-02}& 0.97 &\num{9.12E-02} & 2.34  \\
		\num{5.E+04}&	625 &	80 &\num{8.50E-04}&\num{	1.83E-02}&\num{1.83E-02}&	0.97 &\num{8.14E-02} &	2.03  \\
		\num{8.E+04} &	1000 &	80&\num{ -1.20E-03}&\num{1.48E-02}&\num{	1.48E-02}&	0.97&\num{6.37E-02} &2.37  \\
		\num{1.E+05} & 1000 &100 &\num{	-5.87E-04}&\num{1.27E-02}&\num{	1.27E-02}&	1.00 &\num{6.26E-02}&3.20  \\
		\num{3.E+05} &	2400 &	125 &\num{	3.64E-03}&\num{	9.24E-03}&\num{	9.93E-03}&	0.96 &\num{3.92E-02}&	6.60  \\
		\num{5.E+05} &	2500 &	200 &\num{-2.48E-03}&\num{	9.80E-03} &\num{	1.01E-02}&0.96 &\num{	4.00E-02}&	9.19  \\
		\num{8.E+05} &	4000 &	200 &\num{-1.85E-03}&\num{8.54E-03}&\num{8.74E-03}&	0.92 &\num{3.10E-02}&13.37  \\
		\num{1.E+06} &5000 &200 &\num{-7.16E-04}&\num{	7.36E-03}&\num{7.40E-03}&0.95&\num{2.75E-02}&15.31  \\   
		\bottomrule
	\end{tabular}
\end{table}

\begin{table}[t]
	\footnotesize
	\captionsetup{labelfont=bf}
	\centering
	\caption{ The Performance of the ISE in Large Portfolios}
	\label{IS_large}
	\begin{tabular}{cccrrrcrc}
		\toprule
		\multicolumn{1}{c}{$n$} & \multicolumn{1}{c}{$n_1$} & \multicolumn{1}{c}{$n_2$} & \multicolumn{1}{c}{bias} & \multicolumn{1}{c}{SD} & \multicolumn{1}{c}{RMSE} & \multicolumn{1}{c}{CP} & \multicolumn{1}{c}{width} & \multicolumn{1}{c}{time} \\
		\midrule
		\num{1.E+03}& \num{5.0E+02}& \num{5.0E+02}& \num{-4.48E-03}& \num{3.36E-02}& \num{3.39E-02}& 0.86   & \num{1.07E-01} &	1.37   \\
		\num{3.E+03}& \num{1.5E+03}& \num{1.5E+03}& \num{-1.33E-04}& \num{2.08E-02}& \num{2.08E-02}& 0.85    &\num{6.59E-02}& 1.74   \\
		\num{5.E+03} & \num{2.5E+03} & \num{2.5E+03} & \num{1.59E-03} & \num{1.40E-02} & \num{1.41E-02} & 0.89   &\num{5.15E-02}&	1.80   \\
		\num{8.E+03} & \num{4.0E+03} & \num{4.0E+03} & \num{-1.75E-03} & \num{1.25E-02} & \num{1.26E-02}&	0.87    & \num{3.95E-02}&	2.53  \\
		\num{1.E+04} & \num{5.0E+03} & \num{5.0E+03} & \num{-9.79E-04} & \num{1.20E-02} & \num{1.21E-02} & 0.82   & \num{3.58E-02} & 	1.81   \\
		\num{3.E+04} & \num{1.5E+04} & \num{1.5E+04} & \num{-7.77E-05} & \num{6.28E-03} & \num{6.28E-03} & 0.91   	& \num{2.07E-02}&2.61    \\
		\num{5.E+04} & \num{2.5E+04} & \num{2.5E+04} & \num{3.06E-05} & \num{4.47E-03} & \num{4.47E-03} & 0.92  & \num{1.61E-02}&3.01   \\
		\num{8.E+04} & \num{4.0E+04} & \num{4.0E+04} & \num{-5.42E-05} & \num{3.79E-03} & \num{3.79E-03}&0.90    & \num{1.24E-02}& 4.07  \\
		\num{1.E+05} & \num{5.0E+04} & \num{5.0E+04} & \num{1.64E-04} & \num{3.64E-03} & \num{3.65E-03} & 0.88   & \num{1.16E-02}&	5.16    \\
		\num{3.E+05} & \num{1.5E+05} & \num{1.5E+05} & \num{-2.33E-04} & \num{1.94E-03} & \num{1.95E-03}& 0.87     & \num{6.39E-03}&	14.75  \\
		\num{5.E+05} & \num{2.5E+05} & \num{2.5E+05} & \num{-4.58E-05} & \num{1.44E-03} & \num{1.44E-03} & 0.91   & \num{4.91E-03}&20.75   \\
		\num{8.E+05} & \num{4.0E+05} & \num{4.0E+05} & \num{-2.85E-05} & \num{1.13E-03} & \num{1.13E-03} &0.91   & \num{4.17E-03}&26.83  \\
		\num{1.E+06} & \num{5.0E+05} & \num{5.0E+05} & \num{-1.21E-05} & \num{1.05E-03} & \num{1.05E-03}&0.90  & \num{3.58E-03} & 27.45   \\   
		\bottomrule
	\end{tabular}
\end{table}

\subsection{Large Portfolio with Heavy-Tailed Risk Factors}

In Section \ref{sec4} we assume the risk factors follow a multivariate normal distribution ${\rm \Delta S} \sim {\bf N}({\bf 0},\Sigma)$ and develop an IS-inspired CoVaR estimator. The approach may be extended to certain types of heavy-tailed distributions to account for the heavy-tailed behaviors that are often observed in financial data (see \citealt{bradley2003financial} and \citealt{DuffiePan}). In this subsection we assume that the risk factors follow a multivariate t distribution ${\rm \Delta S} \sim {\bf t}_\nu({\bf 0},\Sigma)$ where $\nu$ is the degrees of freedom. Then, if $C^\top C=\Sigma$, by the definition of the multivariate t distribution \citep{glasserman2004monte}, we have $C^\top {\rm Z}/W \sim {\bf t}_\nu({\bf 0},\Sigma)$, where ${\rm Z}=(Z_1,\ldots,Z_d)^\top$ is a vector of independent standard normal random variables and $W=\sqrt{\chi_\nu^2/\nu}$ where $\chi_\nu^2$ is a chi-squared random variable with $\nu$ degrees of freedom. Notice that we may view $W$ as a common shock to all risk factors and it not only introduces heavy-tailed behaviors to individual risk factors but also leads to extremal dependence among all risk factors, which is commonly observed in financial crisis \citep{bassamboo2008portfolio}. Then, following the derivations in Section \ref{subsec:IS:app}, we may replace the delta-gamma approximations of the losses, i.e., Equations (\ref{eqn:is:X}) and (\ref{eqn:is:Y}), by
\begin{eqnarray*}
	X &=& c_1 + \sum_{j=1}^d \left(\delta_{1j} {Z_j\over W} + \gamma_{1j} {Z_j^2\over W^2}\right), \\
	Y &=& c_2 + \sum_{j=1}^d \left(\delta_{2j} {Z_j\over W} + \gamma_{2j} {Z_j^2\over W^2}\right),
\end{eqnarray*}
where $c_i$, $\delta_{ij}$ and $\gamma_{ij}$, $i=1,2$ and $j=1,\ldots,d$, are same as the ones in Equations (\ref{eqn:is:X}) and (\ref{eqn:is:Y}). We can then apply the same IS-inspired approach to estimate CoVaR by conditioning on $W$ in addition to $Z_1,\ldots,Z_{d-1}$.

We use the same example settings as Section \ref{experiment3} except adding a common shock $W=\sqrt{\chi_\nu^2/\nu}$ with different degrees of freedom $\nu$ to understand how the BE and ISE work under heavy-tailed risk factors and how the tail heaviness, measured by the degrees of freedom $\nu$, affects the VaR and CoVaR. We first repeat the experiments reported in Tables \ref{BE_large} and \ref{IS_large} with $\nu=6$, which is a commonly observed degrees of freedom in financial data (\citealt{vovsvrda2004application}, \citealt{wilhelmsson2006garch}), and report them in Tables \ref{BE_t} and \ref{IS_t}. Similarly, the actual value of CoVaR, 1.4421, is estimated by ISE with  sample size $n=1\times 10^8$. We also plot of the empirical rates of convergence of both estimators in Figure \ref{fig:convergencerate_t}. From the tables and the figure, we see that both estimators have similar performances under the t-distributed risk factors. 

\begin{figure}[ht]
	\captionsetup{labelfont=bf}
	\caption{The Rate of Convergence of the BE and the ISE with Heavy-Tailed Risk Factors}
	\centering
	\includegraphics[scale = 0.5]{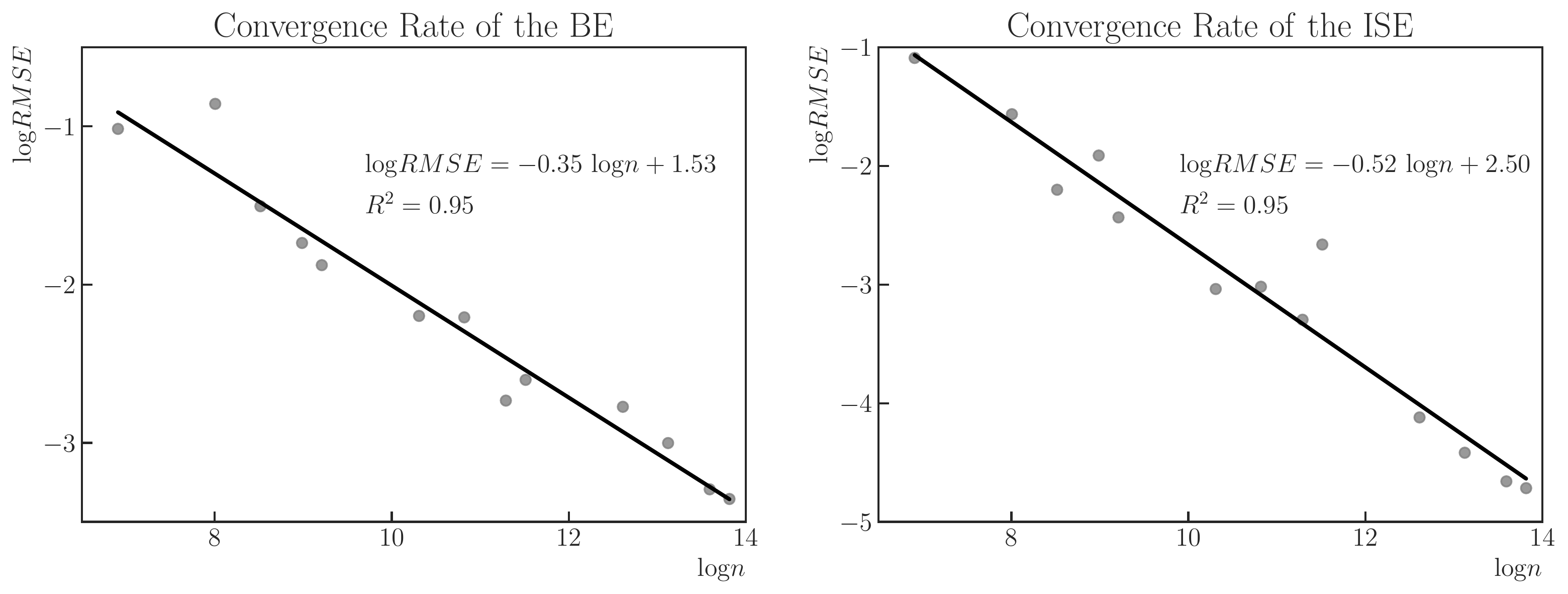}
	\captionsetup{labelfont=bf}\vspace{-11pt}
	\label{fig:convergencerate_t}
\end{figure}

\begin{table}[th]
	\footnotesize
	\captionsetup{labelfont=bf}
	\centering
	\caption{ The Performance of the BE in Large Portfolios with Heavy-Tailed Risk Factors}
	\label{BE_t}
	\begin{tabular}{cccrrrclc}
		\toprule
		\multicolumn{1}{c}{$n$} & \multicolumn{1}{c}{$k$} & \multicolumn{1}{c}{$m$} & \multicolumn{1}{c}{bias} & \multicolumn{1}{c}{SD} & \multicolumn{1}{c}{RMSE} & \multicolumn{1}{c}{CP} & \multicolumn{1}{c}{width} & \multicolumn{1}{c}{time} \\
		\midrule
		\num{1.E+03}&	50 &	20 	&\num{4.33E-02} &\num{	3.59E-01} &\num{3.62E-01} & 0.95  &$1.49\times10^0$ &\num{1.44 } \\
		\num{3.E+03} &	100 &	30 &\num{2.62E-01} &\num{3.33E-01} &\num{4.24E-01} & 0.87  &$2.46\times10^0$ &\num{1.94 } \\
		\num{5.E+03} & 	125 &	40 &\num{2.29E-02} &\num{2.20E-01} &\num{2.22E-01} &0.99 &$1.37\times10^0$ & 1.74  \\
		\num{8.E+03} & 200 & 40 & \num{2.05E-02} &\num{1.74E-01} &\num{1.76E-01} &0.96&\num{	9.97E-01} & 1.74  \\
		\num{1.E+04} & 	250 & 40 &\num{3.13E-02} &\num{	1.50E-01} &\num{1.53E-01}&0.98&\num{7.60E-01}&	1.80  \\
		\num{3.E+04} & 	500 &	60 &\num{1.67E-02}&\num{1.10E-01}&\num{1.11E-01}& 0.97 &\num{4.91E-01} & 	1.92  \\
		\num{5.E+04}&	625 &	80 &\num{3.96E-03}&\num{	1.10E-01}&\num{1.10E-01}&	0.90 &\num{4.16E-01} &	2.34  \\
		\num{8.E+04} &	1000 &	80&\num{ 1.13E-02}&\num{6.40E-02}&\num{6.50E-02}&	0.98&\num{3.28E-01} &	3.22  \\
		\num{1.E+05} & 1000 &100 &\num{	5.00E-03}&\num{7.39E-02}&\num{7.41E-02}&0.98 &\num{3.25E-01}&3.99  \\
		\num{3.E+05} &	2400 &	125 &\num{2.88E-02}&\num{	5.55E-02}&\num{6.25E-02}&0.87 &\num{1.96E-01}&	9.29  \\
		\num{5.E+05} &	2500 &	200 &\num{7.84E-04}&\num{	4.97E-02} &\num{4.97E-02}&0.96 &\num{1.97E-01}&	11.28 \\
		\num{8.E+05} &	4000 &	200 &\num{6.59E-03}&\num{3.65E-02}&\num{3.71E-02}&	0.93 &\num{1.56E-01}&	13.05  \\
		\num{1.E+06} &5000 &200 &\num{5.65E-03}&\num{	3.44E-02}&\num{3.49E-02}&0.96&\num{1.37E-01}&	17.63  \\   
		\bottomrule
	\end{tabular}
\end{table}

\begin{table}[th]
	\footnotesize
	\captionsetup{labelfont=bf}
	\centering
	\caption{ The Performance of the ISE in Large Portfolios with Heavy-Tailed Risk Factors}
	\label{IS_t}
	\begin{tabular}{cccrrrcrc}
		\toprule
		\multicolumn{1}{c}{$n$} & \multicolumn{1}{c}{$n_1$} & \multicolumn{1}{c}{$n_2$} & \multicolumn{1}{c}{bias} & \multicolumn{1}{c}{SD} & \multicolumn{1}{c}{RMSE} & \multicolumn{1}{c}{CP} & \multicolumn{1}{c}{width} & \multicolumn{1}{c}{time} \\
		\midrule
		\num{1.E+03}& \num{5.0E+02}& \num{5.0E+02}& \num{-3.45E-02}& \num{3.34E-01}& \num{3.36E-01}& 0.77  & \num{8.24E-01} &	1.37  \\
		\num{3.E+03}& \num{1.5E+03}& \num{1.5E+03}& \num{1.38E-03}& \num{2.09E-01}& \num{2.09E-01}& 0.87   &\num{5.77E-01}& 1.70  \\
		\num{5.E+03} & \num{2.5E+03} & \num{2.5E+03} & \num{-1.49E-02} & \num{1.10E-01} & \num{1.11E-01} & 0.89   &\num{4.73E-01}&	1.73  \\
		\num{8.E+03} & \num{4.0E+03} & \num{4.0E+03} & \num{2.78E-02} & \num{1.45E-01} & \num{1.48E-01}&	0.94   & \num{4.35E-01}&	4.44 \\
		\num{1.E+04} & \num{5.0E+03} & \num{5.0E+03} & \num{1.45E-03} & \num{8.77E-02} & \num{8.77E-02} & 0.95   & \num{3.93E-01} & 	2.25  \\
		\num{3.E+04} & \num{1.5E+04} & \num{1.5E+04} & \num{-3.37E-03} & \num{4.78E-02} & \num{4.79E-02} & 0.88  	& \num{1.89E-01}&2.46   \\
		\num{5.E+04} & \num{2.5E+04} & \num{2.5E+04} & \num{3.17E-03} & \num{4.88E-02} & \num{4.89E-02} & 0.96   & \num{1.89E-01}&	2.89  \\
		\num{8.E+04} & \num{4.0E+04} & \num{4.0E+04} & \num{-2.43E-03} & \num{3.69E-02} & \num{3.70E-02}&0.95   & \num{1.46E-01}& 4.26   \\
		\num{1.E+05} & \num{5.0E+04} & \num{5.0E+04} & \num{6.37E-03} & \num{6.95E-02} & \num{6.98E-02} & 0.91   & \num{1.34E-01}&	6.05   \\
		\num{3.E+05} & \num{1.5E+05} & \num{1.5E+05} & \num{-5.15E-04} & \num{1.62E-02} & \num{1.63E-02}& 0.97    & \num{7.47E-02}&	14.95   \\
		\num{5.E+05} & \num{2.5E+05} & \num{2.5E+05} & \num{-1.05E-03} & \num{1.20E-02} & \num{1.21E-02} & 0.96  & \num{5.04E-02}&19.75  \\
		\num{8.E+05} & \num{4.0E+05} & \num{4.0E+05} & \num{-2.51E-03} & \num{9.14E-03} & \num{9.48E-03} &0.93  & \num{4.02E-02}&22.70  \\
		\num{1.E+06} & \num{5.0E+05} & \num{5.0E+05} & \num{-4.88E-04} & \num{8.94E-03} & \num{8.95E-03}&0.96 & \num{3.86E-02} & 32.00  \\   
		\bottomrule
	\end{tabular}
\end{table}

We then use the same example with different degrees of freedom, i.e., $\nu=3,4,\ldots,10$ and $\nu=\infty$ (which is the normal distribution), to understand how the tail heaviness affect the VaR of $Y$ and the CoVaR of $Y$ (given $X$ is at risk). Notice that smaller degrees of freedom represent heavier tails. We plot the results in Figure \ref{fig:tail}. There are several findings from the figure. First, as expected, both the VaR and the CoVaR increase as the tail heaviness increases. Second, the CoVaR increases at a faster rate than the VaR as the tail heaviness increases. Third, the difference between the CoVaR and the VaR may be quite significant under heavily tailed distributions, for instance, ${\rm CoVaR}_{0.95,0.95}$ is more than $80\%$ higher than ${\rm VaR}_{0.95}$ when $\nu=6$, indicating that the potential loss of the portfolio $Y$ is significantly higher when the portfolio $X$ is at risk and highlighting the importance of systemic risk.

\begin{figure}[ht]
	\captionsetup{labelfont=bf}
	\caption{VaR and CoVaR with respect to Tail Heaviness}
	\centering
	\includegraphics[scale = 0.6]{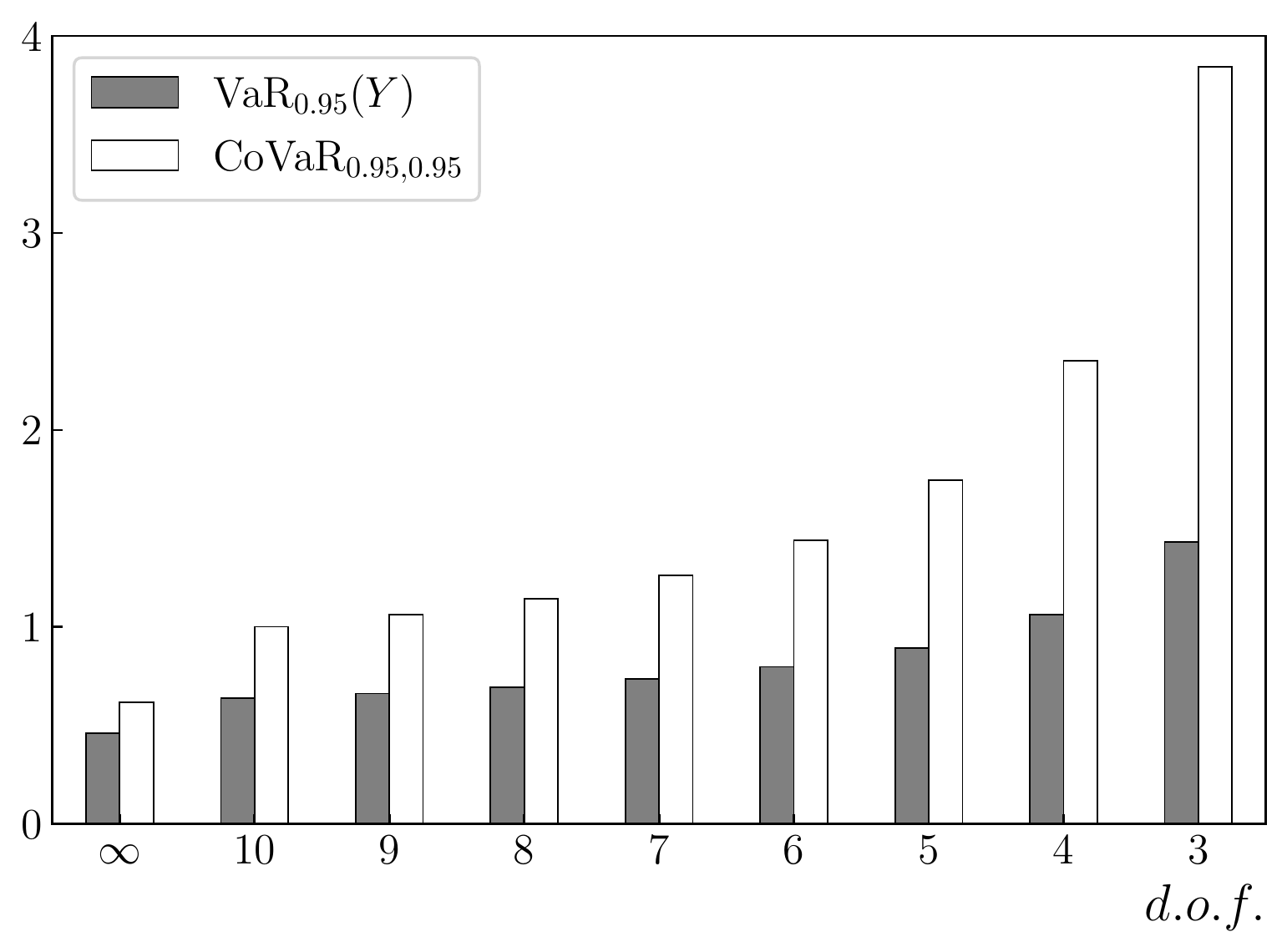}
	\captionsetup{labelfont=bf}\vspace{-11pt}
	\label{fig:tail}
\end{figure}
\vspace{-11pt}

\section{Conclusions}\label{sec6}

In this paper, we study the estimation of CoVaR based on Monte-Carlo simulation. We first develop a batching estimator and show that it is consistent and asymptotically normally distributed, and its optimal rate of convergence is $n^{-1/3}$. The batching estimator takes advantage of the modeling flexibility and is capable of handling complicated portfolios. Moreover, we introduce an IS-inspired estimator to improve the rate of convergence for large portfolios under the delta-gamma approximations. We show that it is consistent and asymptotically normally distributed, and the optimal rate of convergence can be improved to be $n^{-1/2}$.  Numerical experiments support our theoretical findings and show that both estimators work well.

\appendix

\section{Proof of Theorem \ref{CLTBEE}}

\begin{proof}
Let $A=\sqrt{\beta(1-\beta)}/f_{Y|X} ({\rm CoVaR}_{\alpha,\beta}|{\rm VaR}_\alpha(X))$, $y_k= {\rm CoVaR}_{\alpha,\beta} + tAk^{-1/2}$, and 
$$
c_{m,k}(t) = \frac{\sqrt{k} }{\hat\sigma(y_k)}\left(\beta - {\rm E}[I\{ \hat{Y}\leq y_k \}] \right).
$$
Then, for any given $t\in\mathbb{R}$, we have 
\begin{eqnarray}
	\lefteqn{ {\rm Pr}\left\{ \frac{\sqrt{k}}{A} \left(\hat{Y}^{\rm BE}- {\rm CoVaR}_{\alpha,\beta}\right)\leq t \right\} }\nonumber\\
	&=& {\rm Pr}\left\{ \hat{F}_k (y_k ) \geq \beta \right\} 
	\ =\ {\rm Pr} \left\{ \frac{\sqrt{k}}{\hat\sigma(y_k)} \left(\hat{F}_k(y_k )- {\rm E}[I\{ \hat{Y}\leq y_k \}] \right) \geq c_{m,k}(t) \right\}. \label{eqn:Phi} 
\end{eqnarray}
Notice that, by Lemma \ref{AsymAB}, we have
\[
\left|{\rm Pr} \left\{ \frac{\sqrt{k}}{\hat\sigma(y_k)} \left(\hat{F}_k(y_k )- {\rm E}[I\{ \hat{Y}\leq y_k \}] \right) \geq c_{m,k}(t) \right\}- \Phi(-c_{m,k}(t))\right|\ \leq\ \frac{33}{4 \hat\sigma^3(y_k) k^{1/2}}.
\]
Furthermore, notice that $y_k\rightarrow {\rm CoVaR}_{\alpha,\beta}$ as $k \rightarrow\infty$, so when $k$ is large enough, we have $y_k\in\mathcal{Y}$. Then, by Lemma \ref{RCWB}, we have 
\[
\lim_{n\rightarrow\infty} \hat\sigma(y_k )
\ = \ \lim_{n\rightarrow\infty} \sqrt{  {\rm E}[F_{Y|X}(y_k| X_{(\ceil{\alpha m})})]- \left(  {\rm E}[F_{Y|X}(y_k | X_{(\ceil{\alpha m})})]\right)^2} 
\ = \ \sqrt{\beta (1-\beta)}.
\]
Then, by Equation (\ref{eqn:Phi}), we have
\begin{equation}\label{eqn:target}
	{\rm Pr}\left\{ \frac{\sqrt{k}}{A} \left(\hat{Y}^{\rm BE}- {\rm CoVaR}_{\alpha,\beta}\right)\leq t \right\}\ =\ O\left(k^{-1/2}\right)+\Phi(-c_{m,k}(t)).
\end{equation}

By the definition of $F_{Y|X}$ and Equation \eqref{eqn:WBC1}, we have 
\begin{eqnarray}
	c_{m,k}(t) &=& \frac{tA }{\hat \sigma(y_k)}\cdot\frac{\beta - {\rm E}[I\{ \hat{Y}\leq y_k \}] }{t A k^{-1/2}} \nonumber\\
	& = &  \frac{tA }{\hat \sigma(y_k)}\cdot \frac{F_{Y|X}({\rm CoVaR}_{\alpha,\beta}|{\rm VaR}_\alpha(X)) -  F_{Y|X}(y_k|{\rm VaR}_\alpha(X)) }{tAk^{-1/2}} \nonumber\\
	&& \ +\  \frac{tA }{\hat\sigma(y_k)}\cdot \frac{F_{Y|X}(y_k|{\rm VaR}_\alpha(X)) - {\rm E}[ F_{Y|X}(y_k | X_{(\ceil{\alpha m})}) ]}{tAk^{-1/2}}. \label{eqn:m&k}
\end{eqnarray}
Notice that
\[
\frac{F_{Y|X}({\rm CoVaR}_{\alpha,\beta}|{\rm VaR}_\alpha(X)) -  F_{Y|X}(y_k|{\rm VaR}_\alpha(X)) }{tAk^{-1/2}} \to -f_{Y|X} ({\rm CoVaR}_{\alpha,\beta}|{\rm VaR}_\alpha(X))
\]
as $n\to\infty$. Then, the first term of Equation \eqref{eqn:m&k} converges to $-t$ as $n\to\infty$. By Lemma \ref{RCWB}, the second term is $O(\sqrt{k}/m)$ as $n\rightarrow\infty$. Therefore, we have $c_{m,k}(t) = -t + o(1) + O(\sqrt{k}/m)$ as $n\to\infty$.

When $\sqrt{k}/m\rightarrow 0$ as $n\rightarrow\infty$, it is clear that $c_{m,k}(t)\to -t$ as $n\to\infty$. By Equation \eqref{eqn:target},
\begin{equation*}
	\lim_{n\to\infty} {\rm Pr}\left\{ \frac{\sqrt{k}}{A} \left(\hat{Y}^{\rm BE}- {\rm CoVaR}_{\alpha,\beta}\right)\leq t \right\}\ =\ \Phi(t)
\end{equation*}
for any $t\in\mathbb{R}$. Therefore, $\frac{\sqrt{k}}{A} \left(\hat{Y}^{\rm BE}- {\rm CoVaR}_{\alpha,\beta}\right)\Rightarrow N(0,1)$ as $n\to\infty$.

When $\sqrt{k}/m \to c$ as $n\to\infty$ for some constant $c\ne 0$, it is clear that $c_{m,k}(t)=-t+O(1)$. Therefore, there exists a constant $M>0$ such that $c_{m,k}(t)\in(-t-M,-t+M)$. Then, by Equation \eqref{eqn:target},
\begin{eqnarray*}
	\liminf_{n\to\infty} {\rm Pr}\left\{ \frac{\sqrt{k}}{A} \left(\hat{Y}^{\rm BE}- {\rm CoVaR}_{\alpha,\beta}\right)\leq t \right\} \ \geq\ \Phi(t-M),\\
	\limsup_{n\to\infty} {\rm Pr}\left\{ \frac{\sqrt{k}}{A} \left(\hat{Y}^{\rm BE}- {\rm CoVaR}_{\alpha,\beta}\right)\leq t \right\} \ \leq\ \Phi(t+M),
\end{eqnarray*}
and
\[
\limsup_{n\to\infty} {\rm Pr}\left\{ \left|\frac{\sqrt{k}}{A} \left(\hat{Y}^{\rm BE}- {\rm CoVaR}_{\alpha,\beta}\right)\right|\ge t \right\} \ \leq\ 2\Phi(-t+M)
\]
for any $t\in\mathbb{R}$. Therefore, for any $\varepsilon>0$, there exists $t\in\mathbb{R}$ such that $
{\rm Pr}\left\{ \left|\frac{\sqrt{k}}{A} \left(\hat{Y}^{\rm BE}- {\rm CoVaR}_{\alpha,\beta}\right)\right|\ge t \right\}\leq \varepsilon
$
for $n$ is large enough, i.e., $\hat{Y}^{\rm BE}- {\rm CoVaR}_{\alpha,\beta}=O_{\rm Pr}(k^{-1/2})= O_{\rm Pr}(n^{-1/3})$. 
\end{proof}

\section{Proof of Theorem \ref{thm:IS}}\label{appendix2}

\begin{proof}
As shown in Equations \eqref{quadraticeq} and \eqref{eqn:is:g}, we can divide $P_\varepsilon$ into the following three terms:
\begin{eqnarray}
P_\varepsilon
&=& \operatorname{Pr}\{|g(Z_d)-x|\leq \varepsilon|Z_1, \ldots, Z_{d-1}\} \nonumber\\
&=& \operatorname{Pr}\{|g(Z_d)-x|\leq \varepsilon|Z_1, \ldots, Z_{d-1}\}\cdot I\{x>g^\ast+\varepsilon\} \label{term_1}\\
&&+\ \operatorname{Pr}\{|g(Z_d)-x|\leq \varepsilon|Z_1, \ldots, Z_{d-1}\}\cdot I\{|x-g^\ast|\leq\varepsilon\} \label{term_2}\\
&&+\ \operatorname{Pr}\{|g(Z_d)-x|\leq \varepsilon|Z_1, \ldots, Z_{d-1}\}\cdot I\{x<g^\ast-\varepsilon\}. \label{term_3}
\end{eqnarray}
Then, we can analyze the three terms separately. 

For the first term \eqref{term_1}, we know $x>g^\ast+\varepsilon$ implies that the function $g(\cdot)$ intersects with $x\pm \varepsilon$. Let $P^\prime_\varepsilon= \operatorname{Pr}\left\{ \left|g(Z_d)-x\right| \leq \varepsilon \big| Z_1, \ldots, Z_{d-1}\right\} \cdot I\{x>g^\ast+\varepsilon\}$, then 
\begin{eqnarray}
P^\prime_\varepsilon&=& \operatorname{Pr}\left\{ -\varepsilon \leq \xi_1-x+\delta_{1d}Z_d+\gamma_{1d}Z_d^2  \leq \varepsilon \big| Z_1, \ldots, Z_{d-1}\right\} \cdot I\{x>g^\ast+\varepsilon\} \nonumber \\
&=& \operatorname{Pr}\left\{ Z_d\in [r_1-\zeta_1, r_1+\zeta_2]\cup[r_2-\zeta_3, r_2+\zeta_4]\big| Z_1, \ldots, Z_{d-1}\right\} \cdot I\{x>g^\ast+\varepsilon\} \nonumber\\
&=& \left[\int_{r_1-\zeta_1}^{r_1+\zeta_2}f_d(z){\rm d}z+\int_{r_2-\zeta_3}^{r_2+\zeta_4}f_d(z){\rm d}z \right]\cdot I\{x>g^\ast+\varepsilon\} \nonumber \\
&\leq& \left[(\zeta_1+\zeta_2)\cdot f_d(0)+(\zeta_3+\zeta_4)\cdot f_d(0)\right]\cdot I\{x>g^\ast+\varepsilon\} \label{eqn:normpdf} \\
&=& \frac{\sqrt{\Delta_1}-\sqrt{\Delta_2}}{\gamma_{1d}} \cdot f_d(0)\cdot I\{x>g^\ast+\varepsilon\}\nonumber\\
&=&\frac{8\varepsilon}{\sqrt{\Delta_1}+\sqrt{\Delta_2}} \cdot f_d(0)\cdot I\{x>g^\ast+\varepsilon\},\nonumber
\end{eqnarray}
where
\[
r_1-\zeta_1 = \frac{-\delta_{1d}-\sqrt{\Delta_1}}{2\gamma_{1d}},\quad r_1+\zeta_2 = \frac{-\delta_{1d}-\sqrt{\Delta_2}}{2\gamma_{1d}},\quad r_2-\zeta_3 = \frac{-\delta_{1d}+\sqrt{\Delta_2}}{2\gamma_{1d}},\quad r_2+\zeta_4 = \frac{-\delta_{1d}+\sqrt{\Delta_1}}{2\gamma_{1d}}
\]
and $\Delta_1 = \delta_{1d}^2 - 4\gamma_{1d}(\xi_1-x-\varepsilon)=4\gamma_{1d}(x+\varepsilon-g^\ast)\geq 0$, $\Delta_2 = \delta_{1d}^2 - 4\gamma_{1d}(\xi_1-x+\varepsilon)=4\gamma_{1d}(x-\varepsilon-g^\ast)\geq 0$. 
Notice that Equation (\ref{eqn:normpdf}) holds because $f_d(t)$ is the density of the standard normal distribution and it reaches the maximum at $f_d(0)=(\sqrt{2\pi})^{-1}$. Furthermore, notice that $\sqrt{\Delta_1}+\sqrt{\Delta_2}$ is nonincreasing in $\varepsilon$ and achieve its minimum at $x-g^\ast$ when $\varepsilon< x-g^\ast$. Then, 
$
P^\prime_\varepsilon \leq 2\varepsilon/\sqrt{\gamma_{1d}\pi(x-g^\ast)} \cdot I\{x>g^\ast+\varepsilon\}, 
$
when $\varepsilon< x-g^\ast$. 
Therefore, for any $\varepsilon\in(0,x-g^\ast)$, we have 
\begin{equation}
\frac{P^\prime_\varepsilon}{2\varepsilon} \ \le\  \frac{1}{\sqrt{\gamma_{1d}\pi |x-g^\ast|}}. \label{Pp}
\end{equation}
Notice that when $\varepsilon\geq x-g^\ast$, by the definition of $P^\prime_\varepsilon$, we have $P^\prime_\varepsilon=0$, so Equation \eqref{Pp} also holds. 
By the condition that $\operatorname{E}\Big[|g^\ast-x|^{-{1\over2}}\Big]<\infty$, we have $\operatorname{E}\Big[1/\sqrt{\gamma_{1d}\pi |x-g^\ast| }\Big] <\infty$. 
By Lemma \ref{q_derivation} and the dominated convergence theorem, we have
\begin{equation}\label{eqn:P1}
\lim_{\varepsilon\rightarrow0}{1\over 2\varepsilon}\operatorname{E}\left[P^\prime_\varepsilon\right]
\ =\ \operatorname{E}\left[\lim_{\varepsilon\rightarrow0}{1\over 2\varepsilon}P^\prime_\varepsilon \right]
\ =\ \operatorname{E}\left[q_1+q_2\right].
\end{equation}

For the second term \eqref{term_2}, we know  $|x-g^\ast|\leq\varepsilon$ implies that the function $g(\cdot)$ only intersects with $x+ \varepsilon$ but not $x-\varepsilon$. Let $P_\varepsilon^{\prime\prime}=  \operatorname{Pr}\{|g(Z_d)-x|\leq\varepsilon|Z_1,\ldots,Z_{d-1}\}\cdot I\{|x-g^\ast|\leq\varepsilon\}$. Similar to the first term, we have 
\begin{align*}
 P_\varepsilon^{\prime\prime} & \ \leq\  \operatorname{Pr}\{Z_d\in[r_1-\zeta_1,r_2+\zeta_4]|Z_1,\ldots,Z_{d-1}\}\cdot I\{|x-g^\ast|\leq\varepsilon\} \\
&\ \leq\  \frac{\sqrt{4\gamma_{1d}(x-g^\ast+\varepsilon)}f_d(0)}{\gamma_{1d}}\cdot I\{|x-g^\ast|\leq\varepsilon\}\\
&\ \leq\ \frac{2\sqrt{\varepsilon}}{\sqrt{\pi\gamma_{1d}}}\cdot I\{|x-g^\ast|\leq\varepsilon\}.
\end{align*}
Then, we have 
$$
\operatorname{E}[P_\varepsilon^{\prime\prime}] 
\ \leq\ \frac{2\sqrt{\varepsilon}}{\sqrt{\pi\gamma_{1d}}}\cdot \operatorname{Pr}\{|x-g^\ast|\leq\varepsilon\}\ \leq\  \frac{2\sqrt{\varepsilon}}{\sqrt{\pi\gamma_{1d}}}\cdot f^\ast_{g^\ast}\cdot 2\varepsilon,
$$
where $f^\ast_{g^\ast}$ is the maximum of the density of $g^\ast$. Hence, we have 
\begin{equation}
    0\ \leq\ \lim_{\varepsilon\rightarrow0}{1\over2\varepsilon}\operatorname{E}[P_\varepsilon^{\prime\prime}]
\ \leq\ \lim_{\varepsilon\rightarrow0}\frac{2\sqrt{\varepsilon}}{\sqrt{\pi\gamma_{1d}}}\cdot f^\ast_{g^\ast}\ =\ 0.
\label{eqn:P2}
\end{equation}

For the third term \eqref{term_3}, we know $x<g^\ast-\varepsilon$ implies that the function $g(\cdot)$ neither intersects with $x+ \varepsilon$ nor $x-\varepsilon$. Let $P_\varepsilon^{\prime\prime\prime}=\operatorname{Pr}\{|g(Z_d)-x|\leq\varepsilon|Z_1,\ldots,Z_{d-1}\}\cdot I\{x<g^\ast-\varepsilon\}$. Therefore, $P_\varepsilon^{\prime\prime\prime}=0$ and 
\begin{equation}
    \lim_{\varepsilon\rightarrow0}{1\over2\varepsilon}\operatorname{E}[P_\varepsilon^{\prime\prime\prime}] \ =\ 0. 
\label{eqn:P3}
\end{equation}

In summary, we have 
\begin{eqnarray}
\lim_{\varepsilon\rightarrow0}{1\over2\varepsilon}\operatorname{\tilde{E}}[P_\varepsilon]
&=& \lim_{\varepsilon\rightarrow0}{1\over2\varepsilon} \operatorname{\tilde{E}} [P_\varepsilon^{\prime} +P_\varepsilon^{\prime\prime} +P_\varepsilon^{\prime\prime\prime}] \nonumber \\
&=& \lim_{\varepsilon\rightarrow0}{1\over2\varepsilon} \operatorname{E} [P_\varepsilon^{\prime} +P_\varepsilon^{\prime\prime} +P_\varepsilon^{\prime\prime\prime}] \label{measure_restore} \\
&=& \operatorname{E}[q_1+q_2],\label{eqn:PP}
\end{eqnarray}
where Equation \eqref{measure_restore} holds because the randomness of the three terms $P_\varepsilon^\prime$, $P_\varepsilon^{\prime\prime}$ and $P_\varepsilon^{\prime\prime\prime}$ comes from $Z_1,\ldots,Z_{d-1}$ and does not depend on $Z_d$, so the importance sampling distribution is the original distribution, and Equation \eqref{eqn:PP} holds by Equations \eqref{eqn:P1}, \eqref{eqn:P2} and \eqref{eqn:P3}. 

Similarly, we can also prove that
\[
\lim_{\varepsilon\rightarrow0}{1\over2\varepsilon}\operatorname{\tilde{E}}[P_\varepsilon\cdot I\{Y\leq y\}]\ =\ \operatorname{E}[I\{Y_1\leq y\}q_1+I\{Y_2\leq y\}q_2].
\]
Then, the conclusion of the theorem follows directly from Equation (\ref{changemeasure}).
\end{proof}

\section{Proof of Lemma \ref{WBC14}}

\begin{proof}
By Assumption \ref{assu:dist1}, we know that ${\rm E}[Q(y, x)]$ is a continuous function of $x \in\mathbb{R}$. 
We also know that $\tilde{v}_\alpha
\rightarrow v_\alpha$, w.p.1, as $n_1 \rightarrow \infty$, see \cite{Serfling1980}. 
Then, by the continuous mapping theorem \citep{van2000}, we have 
$$
{\rm E}[Q(y, \tilde{v}_\alpha)|\tilde{v}_\alpha] \ \rightarrow\ {\rm E}[Q(y, v_\alpha)], 
$$
w.p.1, as $n_1 \rightarrow\infty$. 
Similarly, we can also prove that 
$
{\rm E}[Q(\tilde{v}_\alpha)|\tilde{v}_\alpha] \ \rightarrow\ {\rm E}[Q(v_\alpha)], 
$
w.p.1, as $n_1 \rightarrow\infty$. 
\end{proof}

\section{Proof of Lemma \ref{ABC4}}

\begin{proof}
As shown in Lemma \ref{q_derivation} and Appendix \ref{appendix2}, we know that 
$$
\lim _{\varepsilon \rightarrow 0} \left[ \frac{P_\varepsilon^\prime (\tilde{v}_\alpha)}{2\varepsilon} \right]^2 \ =\ Q^2(\tilde{v}_\alpha) \quad {\rm w.p.1}, 
$$
where $P_\varepsilon^\prime(x)$ is defined as the term \eqref{term_1} in Appendix \ref{appendix2}. By Equation \eqref{Pp}, we have 
$
\left[ P_\varepsilon^\prime (\tilde{v}_\alpha)/(2\varepsilon) \right]^2 \leq (\gamma_{1d}\pi |\tilde{v}_\alpha -g^\ast|)^{-1}. 
$
Then, we have 
$$
\sup_{n_1}{\rm E}[Q^2(\tilde{v}_\alpha)] \ \leq\ \sup_{n_1}{\rm E} \left[{1\over \gamma_{1d}\pi |\tilde{v}_\alpha -g^\ast|}\right] \ < \ \infty. 
$$

For any $\varepsilon>0$, we have 
\begin{eqnarray}
	\lefteqn{ {\rm Pr}\left\{ \left| \frac1{n_2} \sum_{k=1}^{n_2} Q_k(y, \tilde{v}_\alpha) - {\rm E}[Q(y, \tilde{v}_\alpha) |\tilde{v}_\alpha]  \right|\geq \varepsilon \right\} } \nonumber \\
	& = & {\rm E}\left[ {\rm Pr}\left\{ \left| \frac1{n_2} \sum_{k=1}^{n_2} Q_k(y, \tilde{v}_\alpha) - {\rm E}[Q(y, \tilde{v}_\alpha) |\tilde{v}_\alpha]  \right|\geq \varepsilon \Bigg| \tilde{v}_\alpha \right\} \right] \label{LTEinE} \\
	&\leq & {1\over\varepsilon^2} {\rm E}\left[ {\rm Var}\left( \frac1{n_2} \sum_{k=1}^{n_2} Q_k(y, \tilde{v}_\alpha) \Bigg| \tilde{v}_\alpha \right) \right] \label{Chebyshev} \\
	&=& {1\over n_2\varepsilon^2} {\rm E}\left[ {\rm Var}\left( Q(y, \tilde{v}_\alpha) | \tilde{v}_\alpha  \right) \right] \label{eqn:conditionalonhataIndep} \\
	&\leq & {1\over n_2\varepsilon^2}  {\rm Var}\left( Q(y, \tilde{v}_\alpha) \right) 
	\ \leq \ {1\over n_2\varepsilon^2} {\rm E}\left[Q^2(y, \tilde{v}_\alpha) \right]
	\ \leq \  {1\over n_2\varepsilon^2} \sup\limits_{n_1} {\rm E}\left[Q^2(\tilde{v}_\alpha)\right] \ \rightarrow \ 0,  \nonumber
\end{eqnarray}
as $n_2 \rightarrow\infty$. 
Notice that Equation \eqref{LTEinE} holds by the law of total expectation, Equation \eqref{Chebyshev} holds by the Chebyshev's inequality, and Equation \eqref{eqn:conditionalonhataIndep} holds because conditional on $\tilde{v}_\alpha$, $\{ Q_k(y, \tilde{v}_\alpha) \}_{k=1}^{n_2}$ is independent. Hence, we obtain that $\frac1{n_2} \sum\nolimits_{k=1}^{n_2} Q_k(y, \tilde{v}_\alpha) - {\rm E}[Q(y, \tilde{v}_\alpha)|\tilde{v}_\alpha]\rightarrow 0$ in probability, as $n_2 \rightarrow\infty$. Similarly, we can also prove that $\frac1{n_2} \sum\nolimits_{k=1}^{n_2} Q_k(\tilde{v}_\alpha)- {\rm E}[Q(\tilde{v}_\alpha)|\tilde{v}_\alpha]\rightarrow 0$ in probability, as $n_2 \rightarrow\infty$.  
\end{proof}

\section{Proof of Lemma \ref{RCWB4}}

\begin{proof}
From Assumption \ref{assu:dist1}, by Taylor's expansion, for $y\in \mathcal{Y}$, we have 
$$
	 {\rm E}[Q(y, \tilde{v}_\alpha)| \tilde{v}_\alpha]- {\rm E}[Q(y, v_\alpha)] \ = \ \frac{\partial}{\partial x}{\rm E}[Q(y, v_\alpha)] \cdot  ( \tilde{v}_\alpha -v_\alpha ) + \frac{\partial^2}{\partial x^2} {\rm E}[Q(y, x) | x=Z] \cdot ( \tilde{v}_\alpha -v_\alpha )^2, 
$$
for some random variable $Z$. 
By the law of total expectation, we have 
$
{\rm E}[Q(y, \tilde{v}_\alpha)] = {\rm E}\big[{\rm E}[Q(y, \tilde{v}_\alpha)| \tilde{v}_\alpha]\big]. 
$
Then, by the assumptions $|\frac{\partial}{\partial x} {\rm E}[{Q}(y, v_\alpha)]|\leq M$ for all $y\in\mathcal{Y}$ and $|\frac{\partial^2}{\partial x^2} {\rm E}[{Q}(y, x)]|\leq M$ for all $(x,y)\in \mathbb{R}\times\mathcal{Y}$, we have 
\begin{eqnarray}
	\lefteqn{  \Big| {\rm E}[Q(y, \tilde{v}_\alpha)] - {\rm E}[Q(y, v_\alpha)] \Big| }\nonumber\\
	& =& \Big|{\rm E}\Big[ {\rm E}[Q(y, \tilde{v}_\alpha)| \tilde{v}_\alpha]- {\rm E}[Q(y, v_\alpha)] \Big] \Big| \nonumber \\
	&\leq &  \left| \frac{\partial}{\partial x}{\rm E}[Q(y, v_\alpha) ] \right| \cdot \big| {\rm E}[\tilde{v}_\alpha - v_\alpha] \big| +  {\rm E}\left[\left| \frac{\partial^2}{\partial x^2} {\rm E}[Q(y, x) | x=Z] \right| \cdot (\tilde{v}_\alpha -v_\alpha )^2 \right] \nonumber \\
	&\leq &  M \cdot \big| {\rm E}[ \tilde{v}_\alpha-v_\alpha ]\big| + M\cdot {\rm E}\left[| \tilde{v}_\alpha-v_\alpha|^2 \right]. \nonumber 
\end{eqnarray}
Because both ${\rm E}[ \tilde{v}_\alpha-v_\alpha]$ and ${\rm E}\left[| \tilde{v}_\alpha- v_\alpha|^2 \right]$ are of $O(n_1^{-1})$ (see Lemma 2 of \cite{Hong09}), so we have $
\sup_{y\in\mathcal{Y}} \big| {\rm E}[Q(y, \tilde{v}_\alpha)]- {\rm E}[Q(y, v_\alpha)] \big|
$
is of $O(n_1^{-1})$ as well. 
\end{proof}

\section{Proof of Lemma \ref{AsymAB4}}

\begin{proof}
As shown in Lemma \ref{q_derivation} and Appendix \ref{appendix2}, we know that 
$$
\lim _{\varepsilon \rightarrow 0} \left[ \frac{P_\varepsilon^\prime (\tilde{v}_\alpha)}{2\varepsilon} \right]^3 \ =\ Q^3(\tilde{v}_\alpha) \quad {\rm w.p.1}, 
$$
where $P_\varepsilon^\prime(x)$ is defined as the term \eqref{term_1} in Appendix \ref{appendix2}. By Equation \eqref{Pp}, we have 
$
\left[ P_\varepsilon^\prime (\tilde{v}_\alpha)/(2\varepsilon) \right]^3 \leq (\gamma_{1d}\pi |\tilde{v}_\alpha -g^\ast|)^{-3/2}. 
$
Then, we have 
$$
\sup_{n_1}{\rm E}[Q^3(\tilde{v}_\alpha)] \ \leq\ \sup_{n_1}{\rm E} \left[{1\over (\gamma_{1d}\pi |\tilde{v}_\alpha -g^\ast|)^{\frac32} } \right] \ < \ \infty. 
$$

Notice that conditional on $\tilde{v}_\alpha$, we have $\{Q_k(y, \tilde{v}_\alpha)\}_{k=1}^{n_2}$ is independent. Then, 
\begin{eqnarray}
	\lefteqn{ \left| {\rm Pr}\left\{ \frac{\sqrt{n_2}}{\tilde\sigma(y,\tilde{v}_\alpha)} \Big\{ \frac1{n_2} \sum_{k=1}^{n_2} Q_k(y, \tilde{v}_\alpha) - {\rm E}\big[Q(y, \tilde{v}_\alpha) |\tilde{v}_\alpha \big] \Big\} \leq c(\tilde{v}_\alpha) \right\} - {\rm E}\big[\Phi\big(c(\tilde{v}_\alpha)\big)\big] \right| } \nonumber \\
	& = & \left| {\rm E}\left[ {\rm Pr}\left\{ \frac{\sqrt{n_2}}{\tilde\sigma(y,\tilde{v}_\alpha)} \Big\{ \frac1{n_2} \sum_{k=1}^{n_2} Q_k(y, \tilde{v}_\alpha) - {\rm E}\big[Q(y, \tilde{v}_\alpha) |\tilde{v}_\alpha \big] \Big\} \leq c(\tilde{v}_\alpha) \Big| \tilde{v}_\alpha \right\} - \Phi\big(c(\tilde{v}_\alpha)\big)\right] \right| \nonumber \\
	& \leq & {\rm E}\left[ \left| {\rm Pr}\left\{ \frac{\sqrt{n_2}}{\tilde\sigma(y,\tilde{v}_\alpha)} \Big\{ \frac1{n_2} \sum_{k=1}^{n_2} Q_k(y, \tilde{v}_\alpha) - {\rm E}\big[Q(y, \tilde{v}_\alpha) |\tilde{v}_\alpha \big] \Big\} \leq c(\tilde{v}_\alpha) \Big| \tilde{v}_\alpha \right\} - \Phi\big(c(\tilde{v}_\alpha)\big) \right| \right] \nonumber \\
	& \leq & {\rm E}\left[ \sup_{t\in\mathbb{R}} \left| {\rm Pr}\left\{ \frac{ \sum_{k=1}^{n_2} Q_k(y, \tilde{v}_\alpha) - {\rm E}[\sum_{k=1}^{n_2} Q_k(y, \tilde{v}_\alpha) |\tilde{v}_\alpha ] }{\sqrt{{\rm Var}(\sum_{k=1}^{n_2} Q_k(y, \tilde{v}_\alpha) |\tilde{v}_\alpha )}} \leq t \Big| \tilde{v}_\alpha \right\} -\Phi(t) \right| \right] \nonumber \\
	& \leq & \frac{33}{4} {\rm E}\left[  \frac{ |  Q(y, \tilde{v}_\alpha) - {\rm E}[ Q(y, \tilde{v}_\alpha) |\tilde{v}_\alpha ]|^3 }{\tilde\sigma^3(y,\tilde{v}_\alpha) n_2^{1/2}}  \right] \label{eqn:BEthm3} \\
	& \leq & O\left( \frac{\sup_{n_1} {\rm E}[Q^3(\tilde{v}_\alpha)] }{ n_2^{1/2}}\right) \ =\ O(n_2^{-1/2}). \nonumber
\end{eqnarray}
Notice that Equation \eqref{eqn:BEthm3} holds by the Berry-Ess\'een Theorem \citep{Serfling1980}. 
Therefore, we conclude the proof of the lemma. 
\end{proof}

\section{Proof of Theorem \ref{CLTBEE4}}

\begin{proof}
The idea of the proof is using the asymptotic distribution of the sample distribution $\tilde{F}_{n_2}(y,\tilde{v}_\alpha)$ to prove that of the estimator $\tilde{Y}^{\rm IS}$. We first notice that 
$$
\tilde{F}_{n_2}(y, \tilde{v}_\alpha)\ =\ \frac{\frac1{n_2} \sum_{k=1}^{n_2} Q_k(y, \tilde{v}_\alpha)}{\frac1{n_2} \sum_{k=1}^{n_2} Q_k(\tilde{v}_\alpha) }  \ =\ \frac{\frac1{n_2} \sum_{k=1}^{n_2} Q_k(y, \tilde{v}_\alpha)}{{\rm E}[ Q( v_\alpha)]} \cdot \frac{{\rm E}[ Q( v_\alpha)]}{\frac1{n_2} \sum_{k=1}^{n_2} Q_k(\tilde{v}_\alpha)},   
$$
and, by Lemmas \ref{WBC14} and \ref{ABC4}, we have ${\rm E}[ Q( v_\alpha)]/(\frac1{n_2} \sum_{k=1}^{n_2} Q_k(\tilde{v}_\alpha)) \rightarrow 1$ in probability as $n\rightarrow\infty$. 
Then, by Slutsky's lemma \citep{van2000}, we know that 
$$
\lim\limits_{n\rightarrow\infty} {\rm Pr}\left\{ \tilde{F}_{n_2}(y, \tilde{v}_\alpha) \geq \beta \right\}\ =\ \lim\limits_{n\rightarrow\infty} {\rm Pr}\left\{ \frac{\frac1{n_2} \sum_{k=1}^{n_2} Q_k(y, \tilde{v}_\alpha)}{{\rm E}[ Q( v_\alpha)]} \geq \beta\right\}.
$$ 

Let $A=\tilde\sigma({\rm CoVaR}_{\alpha,\beta},v_\alpha)/\{{\rm E}[Q(v_\alpha)] \cdot f_{Y|X} ({\rm CoVaR}_{\alpha,\beta}|v_\alpha) \}$, $y_{n_2}= {\rm CoVaR} + tAn_2^{-1/2}$, and $$c_{n_1,n_2}(t)= \frac{\sqrt{n_2} \left(\beta {\rm E}[ Q(v_\alpha) ] - {\rm E}[ Q(y_{n_2}, \tilde{v}_\alpha) |\tilde{v}_\alpha] \right)}{\tilde\sigma(y_{n_2}, \tilde{v}_\alpha)}.$$
Notice that $c_{n_1,n_2}(t)$ is a random variable since it depends on $\tilde{v}_\alpha$. 
Then, for any given $t\in\mathbb{R}$, we have 
\begin{eqnarray}
	\lefteqn{ \lim_{n\rightarrow\infty} {\rm Pr}\left\{ \frac{\sqrt{n_2}}{A}(\tilde{Y}^{\rm IS}- {\rm CoVaR}_{\alpha,\beta}) \leq t \right\} 
		\ =\ \lim_{n\rightarrow\infty} {\rm Pr}\left\{ \tilde{F}_{n_2} (y_{n_2} ,\tilde{v}_\alpha) \geq \beta \right\} } \nonumber \\
	& = & \lim_{n\rightarrow\infty} {\rm Pr}\left\{ \frac1{n_2} \sum_{k=1}^{n_2} Q_k(y_{n_2}, \tilde{v}_\alpha)   \geq \beta {\rm E}[ Q( v_\alpha) ] \right\}  \nonumber \\
	& = & \lim_{n\rightarrow\infty} {\rm Pr} \left\{ \frac{\sqrt{n_2} }{\tilde\sigma(y_{n_2}, \tilde{v}_\alpha)}\left(\frac1{n_2} \sum_{k=1}^{n_2} Q_k(y_{n_2}, \tilde{v}_\alpha)  - {\rm E}[ Q(y_{n_2}, \tilde{v}_\alpha) |\tilde{v}_\alpha] \right) \geq c_{n_1,n_2}(t) \right\}. \nonumber
\end{eqnarray}
Notice that, by Lemma \ref{AsymAB4}, we have 
$$
 {\rm Pr} \left\{ \frac{\sqrt{n_2} }{\tilde\sigma(y_{n_2}, \tilde{v}_\alpha)}\left(\frac1{n_2} \sum_{k=1}^{n_2} Q_k(y_{n_2}, \tilde{v}_\alpha)  - {\rm E}[ Q(y_{n_2}, \tilde{v}_\alpha) |\tilde{v}_\alpha] \right) \geq c_{n_1,n_2}(t) \right\} \ =\ {\rm E}\left[ \Phi(-c_{n_1,n_2}(t)) \right] + O(n_2^{-1/2}). 
$$
Then, we have 
\begin{equation}
	\label{eqn:BEthm2}
	{\rm Pr}\left\{ \frac{\sqrt{n_2}}{A}(\tilde{Y}^{\rm IS}- {\rm CoVaR}_{\alpha,\beta}) \leq t \right\} \ =\ {\rm E}\left[ \Phi(-c_{n_1,n_2}(t)) \right] + O(n_2^{-1/2}) +o(1). 
\end{equation}
By the definition of $F_{Y|X}$, we have 
\begin{eqnarray}
	\lefteqn{c_{n_1,n_2}(t)\ =\ \frac{tA}{\tilde\sigma(y_{n_2}, \tilde{v}_\alpha)}\frac{\beta {\rm E}[ Q(v_\alpha) ] - {\rm E}[ Q(y_{n_2}, \tilde{v}_\alpha) |\tilde{v}_\alpha] }{t A n_2^{-1/2}}}  \nonumber \\
	&=& \frac{tA{\rm E}[ Q( v_\alpha) ]}{\tilde\sigma(y_{n_2}, \tilde{v}_\alpha)}\frac{F_{Y|X}({\rm CoVaR}_{\alpha,\beta}|v_\alpha) - F_{Y|X}(y_{n_2} |v_\alpha)  }{tA n_2^{-1/2}} + \frac{tA}{\tilde\sigma(y_{n_2}, \tilde{v}_\alpha)}\frac{ {\rm E}[ Q(y_{n_2}, v_\alpha) ]  - {\rm E}[ Q(y_{n_2}, \tilde{v}_\alpha) ] }{tA n_2^{-1/2} }.~~~~~ \label{eqn:cc}  
\end{eqnarray}
Furthermore, notice that $y_{n_2}\rightarrow{\rm CoVaR}_{\alpha,\beta}$ as $n_2\rightarrow\infty$, so when $n_2$ is large enough, we have $y_{n_2} \in\mathcal{Y}$. Then, by the assumption $\tilde\sigma(y,x)$ is a continuous function of $(x,y)$ in $\mathbb{R}\times\mathcal{Y}$ and the continuous-mapping theorem, we have $ \tilde\sigma(y_{n_2}, \tilde{v}_\alpha) \rightarrow \tilde\sigma({\rm CoVaR}_{\alpha,\beta}, v_\alpha)$ w.p.1 as $n\rightarrow\infty$. 
Notice that \[
\frac{F_{Y|X}({\rm CoVaR}_{\alpha,\beta}|v_\alpha) -  F_{Y|X}(y_{n_2} |v_\alpha) }{tAn_2^{-1/2}} \ \rightarrow\ - f_{Y|X} ({\rm CoVaR}_{\alpha,\beta}|v_\alpha)
\]
as $n\to\infty$.
Then, the first term of Equation \eqref{eqn:cc} converges to $-t$ w.p.1 as $n\rightarrow\infty$. 
By Lemma \ref{RCWB4}, the second term is $O(\sqrt{n_2}/n_1)$ w.p.1 as $n\rightarrow\infty$. 
Therefore, we have $c_{n_1,n_2}(t) =-t+o(1)+O(\sqrt{n_2}/n_1)$ w.p.1 as $n\rightarrow\infty$. 

When $\sqrt{n_2}/n_1 \rightarrow 0$, as $n\rightarrow\infty$, it is clear that $c_{n_1,n_2}(t)\rightarrow -t $ w.p.1 as $n\rightarrow\infty$. 
By the continuous mapping theorem, we have $\Phi(-c_{n_1,n_2}(t))\rightarrow \Phi(t)$ w.p.1 as $n\rightarrow\infty$. Because $\Phi(-c_{n_1,n_2}(t))$ is bounded by $1$, so, by the dominated convergence theorem, we have ${\rm E}[\Phi(-c_{n_1,n_2}(t))] \rightarrow \Phi(t)$ as $n\rightarrow\infty$. 
By Equation \eqref{eqn:BEthm2}, 
$$
\lim_{n\rightarrow\infty} {\rm Pr}\left\{ \frac{\sqrt{n_2}}{A}(\tilde{Y}^{\rm IS}- {\rm CoVaR}_{\alpha,\beta}) \leq t \right\} \ =\ \Phi(t) 
$$
for any $t\in\mathbb{R}$. Therefore, $\frac{\sqrt{n_2}}{A}(\tilde{Y}^{\rm IS}- {\rm CoVaR}_{\alpha,\beta}) \Rightarrow N(0,1)$ as $n\rightarrow\infty$. 

When $\sqrt{n_2}/n_1 \to c$ as $n\to\infty$ for some constant $c\ne 0$, it is clear that $c_{n_1,n_2}(t)=-t+O(1)$ w.p.1. Therefore, there exists a constant $M>0$ such that $c_{n_1,n_2}(t)\in(-t-M,-t+M)$ w.p.1. Then, by Equation \eqref{eqn:BEthm2},
\begin{eqnarray*}
	\liminf_{n\to\infty} {\rm Pr}\left\{ \frac{\sqrt{n_2}}{A} \left(\tilde{Y}^{\rm IS}- {\rm CoVaR}_{\alpha,\beta}\right)\leq t \right\} \ \geq\ \Phi(t-M),\\
	\limsup_{n\to\infty} {\rm Pr}\left\{ \frac{\sqrt{n_2}}{A} \left(\tilde{Y}^{\rm IS}- {\rm CoVaR}_{\alpha,\beta}\right)\leq t \right\} \ \leq\ \Phi(t+M),
\end{eqnarray*}
and
\[
\limsup_{n\to\infty} {\rm Pr}\left\{ \left|\frac{\sqrt{n_2}}{A} \left(\tilde{Y}^{\rm IS}- {\rm CoVaR}_{\alpha,\beta}\right)\right|\ge t \right\} \ \leq\ 2\Phi(-t+M)
\]
for any $t\in\mathbb{R}$. Therefore, for any $\varepsilon>0$, there exists $t\in\mathbb{R}$ such that $
{\rm Pr}\left\{ \left|\frac{\sqrt{n_2}}{A} \left(\tilde{Y}^{\rm IS}- {\rm CoVaR}_{\alpha,\beta}\right)\right|\ge t \right\}\leq \varepsilon
$
for $n$ is large enough, i.e., $\tilde{Y}^{\rm IS}- {\rm CoVaR}_{\alpha,\beta}=O_{\rm Pr}(n_2^{-1/2})= O_{\rm Pr}(n^{-1/2})$. 
\end{proof}

\section{Parameter Setting in Section \ref{experiment3} }\label{ex3_parameters}

In the experiment of Section \ref{experiment3}, we assume that the portfolio losses, i.e., $X$ and $Y$, have 50 correlated risk factors, denoted by ${\rm \Delta S}$. Suppose ${\rm \Delta S}$ follows a multivariate normal distribution with mean vector $\bm{\mu}$ and covariance matrix $\Sigma$. Furthermore, $X$ and $Y$ can be approximated by a quadratic function with respect to ${\rm \Delta S}$, see Section \ref{subsec:IS:app}. We denote the initial parameters of the delta-gamma approximation as
\begin{align*}
	X &\ =\  -\bar\Theta_1 \Delta  t - \bar{\rm \delta}_1^{\top} {\rm \Delta S} - \frac{1}{2} {\rm \Delta S}^{\top}  \bar{\rm \Gamma}_1 {\rm \Delta S}, \\
	Y &\ =\  -\bar\Theta_2 \Delta  t - \bar{\rm \delta}_2^{\top} {\rm \Delta S} - \frac{1}{2} {\rm \Delta S}^{\top}  \bar{\rm \Gamma}_2 {\rm \Delta S}.
\end{align*}
Following the procedures in Section \ref{subsec:IS:app}, we derive a simpler form of $X$ and $Y$ as \eqref{eqn:is:X} and \eqref{eqn:is:Y} with respect to the standard normal random variable $Z_j$, for $j=1,\ldots, 50$. We denote the parameters of the simplified delta-gamma approximation as $\delta_{1j}$, $\gamma_{1j}$, $\delta_{2j}$, $\gamma_{2j}$, $j=1,\ldots, 50$. In our experiment, the parameters of ${\rm \Delta S}$ and initial delta-gamma approximation are generated randomly fixing the random seed. We provide the details of the parameters as follows.
\bigskip 

\noindent 1. The parameters of ${\rm \Delta S}$
\begin{itemize}
	\item The mean vector of ${\rm \Delta S}$ is $\bm 0$
	
	\item Generate covariance matrix of ${\rm \Delta S}$.
	
	$-$ Firstly, generate the standard deviation vector of ${\rm \Delta S}$. Generate 50 random variables from ${\rm Unif}[0,1]$. Sort from them smallest to largest and denote them as $\sigma_j$, $j=1,\ldots,50$. Let standard deviation vector $\bm{\sigma}=[\sigma_1,\ldots, \sigma_{50}]$. 
	
	$-$ Secondly, generate a vector of eigenvalues and then the correlation matrix. Generate 25 random variables $e_j$, $j=1,\ldots,25$, from ${\rm Unif}[0,2]$ and let $e_{50-j} = 2-e_j$. Thus, we obtain the vector of eigenvalues $\bm{e}=[e_1,\ldots,e_{50}]$. Given $\bm{e}$, generate the correlation matrix $A$ using ${\tt scipy.stats.random\_correlation}$ function in Python.
	
	$-$ Thirdly, compute the covariance matrix by $\Sigma = \bm{\sigma}^\top \bm{\sigma}\odot A$, where $\odot$ denote element-wise product.
\end{itemize}

\noindent 2. The parameters of initial delta-gamma approximation  
\begin{itemize}
    \item  $\bar\Theta_1=\bar\Theta_2=0$. 
    
	\item $\bar{\delta}_1$: 50-dimensional vector whose components are generated from ${\rm Unif}[-0.005,0.005]$.
	
	\item $\bar{\Gamma}_1$: $50\times 50$ matrix  generated by $(\bar{G}_1+\bar{G}_1^\top)/2$ where $\bar{G}_1[50,50]=0.8$(to denote heavy-weighted financial asset in the portfolio) and the other elements of $\bar{G}_1$ are generated from ${\rm Unif}[-0.02,0.02]$. 
	
	\item $\bar{\delta}_2$: 50-dimensional vector whose components are generated from ${\rm Unif}[-0.005,0.005]$.
	
	\item $\bar{\Gamma}_2$: $50\times 50$ matrix  generated by $(\bar{G}_2+\bar{G}_2^\top)/2$ where $\bar{G}_2[49,49]=0.1$, $\bar{G}_2[50,50]=0.05$ and the other elements of $\bar{G}_2$ are generated from ${\rm Unif}[-0.04,0.04]$.
\end{itemize}

\noindent 3. The parameters of the simplified delta-gamma approximations

Based on the parameters of ${\rm \Delta S}$ and initial delta-gamma approximation, we can derive the parameters of the simplified delta-gamma approximation:

\vspace{-11pt}
{\footnotesize
\begin{flalign*}
	&\delta_{1}=\\
	&
	\begin{array}{rrrrrrr}
		\big[\num{-2.04E-03},&\num{-5.56E-04},&\num{-3.06E-04},&\num{1.94E-03},&\num{7.03E-04},&\num{1.32E-04},&\num{-1.75E-03},\\
		\num{-6.79E-04},&\num{9.27E-04},&\num{2.14E-03},&\num{8.05E-04},&\num{-1.58E-03},&\num{4.74E-04},&\num{-2.06E-04},\\
		\num{-1.67E-03},&\num{-4.66E-04},&\num{6.03E-05},&\num{1.46E-03},&\num{3.17E-04},&\num{1.32E-03},&\num{1.96E-03},\\
		\num{-2.95E-03},&\num{-1.13E-03},&\num{-7.05E-04},&\num{-1.14E-03},&\num{-2.91E-03},&\num{-9.88E-04},&\num{5.80E-04},\\
		\num{2.81E-04},&\num{2.67E-03},&\num{2.86E-03},&\num{3.13E-03},&\num{-1.04E-04},&\num{1.03E-03},&\num{5.53E-04},\\
		\num{-1.01E-03},&\num{-3.17E-03},&\num{1.16E-03},&\num{-2.26E-04},&\num{1.87E-03},& \num{6.50E-04},&\num{3.38E-03}\\
		\num{1.88E-03},&\num{-3.42E-04},&\num{-3.97E-03},&\num{1.94E-03},&\num{-1.61E-03},&\num{-6.50E-04},&\num{-3.70E-04},\\
		\num{1.15E-03} \big],& & & & & &
	\end{array}
\end{flalign*}

\vspace{-11pt}
\begin{flalign*}
	&\gamma_{1}=\\
	&
	\begin{array}{rrrrrrr}
		\big[\num{-2.70E-02},&\num{-1.84E-02},&\num{-1.68E-02},&\num{-1.25E-02},&\num{-1.15E-02},&\num{-7.74E-03},&\num{-6.71E-03},\\
		\num{-5.80E-03},&\num{-5.08E-03},&\num{-4.45E-03},&\num{-3.77E-03},&\num{-3.18E-03},&\num{-2.35E-03},&\num{-1.98E-03},\\
		\num{-1.81E-03},&\num{-1.17E-03},&\num{-1.02E-03},&\num{-5.46E-04},&\num{-2.82E-04},&\num{-2.56E-04},&\num{-1.18E-04},\\
		\num{-6.98E-05},&\num{-3.71E-05},&\num{-2.54E-05},&\num{-8.86E-07},&\num{8.07E-06},&\num{2.65E-05},&\num{8.39E-05},\\
		\num{1.21E-04},&\num{1.25E-04},&\num{3.03E-04},&\num{5.26E-04},&\num{7.85E-04},&\num{9.46E-04},&\num{1.51E-03}\\
		\num{1.60E-03},&\num{2.02E-03},&\num{3.13E-03},&\num{3.58E-03},&\num{4.21E-03},&
		\num{5.03E-03},&\num{6.53E-03}\\
		\num{7.33E-03},&\num{7.93E-03},&\num{1.18E-02},&\num{1.43E-02},&\num{1.70E-02},&\num{2.20E-02},&\num{3.28E-02},\\
		\num{3.96E-01}\big],& & & & & &
	\end{array}
\end{flalign*}

\vspace{-11pt}
\begin{flalign*}
	&\delta_{2}=\\
	&
	\begin{array}{rrrrrrr}
		\big[\num{-4.27E-04},&\num{7.47E-05},&\num{-2.28E-03},&\num{-6.62E-04},&\num{-1.85E-03},&\num{-2.44E-03},&\num{-3.18E-03}\\
		\num{1.04E-03},&\num{1.55E-03},&\num{-1.54E-03},&\num{1.09E-03},&\num{-1.18E-03},&\num{-1.03E-03},&\num{2.03E-04}\\
		\num{-3.03E-03},&\num{6.99E-04},&\num{-2.17E-03},&\num{-1.46E-03},&\num{1.47E-03},&\num{-6.34E-04},&\num{7.60E-04}\\
		\num{3.49E-05},&\num{-3.45E-04},&\num{-4.75E-04},&\num{-6.02E-04},&\num{-3.13E-04},&\num{-9.54E-04},&\num{1.49E-03}\\
		\num{-1.65E-03},&\num{1.90E-03},&\num{1.01E-03},&\num{-5.71E-05},&\num{3.75E-04},&\num{1.63E-03},&\num{-9.59E-04}\\
		\num{1.67E-03},&\num{3.34E-03},&\num{3.69E-03},&\num{-2.46E-04},&\num{4.85E-03},&\num{9.56E-04},&\num{1.43E-03}\\
		\num{4.48E-03},&\num{3.79E-03},&\num{2.61E-04},&\num{4.53E-04},&\num{3.03E-03},&\num{3.88E-03},&\num{2.92E-03}\\
		\num{2.96E-03}\big].& & & & & &
	\end{array}
\end{flalign*}

\vspace{-11pt}
\begin{flalign*}
	&\gamma_{2}=\\
	&
	\begin{array}{rrrrrrr}
		\big[\num{-5.50E-02},&\num{-3.85E-02},&\num{-2.79E-02},&\num{-2.47E-02},&\num{-2.31E-02},&\num{-1.59E-02},&\num{-1.42E-02}\\
		\num{-1.22E-02},&\num{-1.01E-02},&\num{-8.10E-03},&\num{-6.67E-03},&\num{-4.35E-03},&\num{-3.85E-03},&\num{-3.40E-03}\\
		\num{-2.78E-03},&\num{-1.91E-03},&\num{-1.67E-03},&\num{-1.02E-03},&\num{-8.50E-04},&\num{-4.63E-04},&\num{-2.59E-04}\\
		\num{-2.42E-04},&\num{-3.56E-05},&\num{-1.69E-05},&\num{-9.58E-06},&\num{1.94E-06},&\num{3.31E-05},&\num{4.87E-05}\\
		\num{2.78E-04},&\num{5.07E-04},&\num{1.02E-03},&\num{1.33E-03},&\num{1.65E-03},&\num{2.21E-03},&\num{3.32E-03}\\
		\num{4.18E-03},&\num{5.12E-03},&\num{6.50E-03},&\num{8.09E-03},&\num{8.95E-03},&\num{1.05E-02},&\num{1.33E-02}\\
		\num{1.45E-02},&\num{1.94E-02},&\num{2.87E-02},&\num{3.32E-02},&\num{3.61E-02},&\num{3.93E-02},&\num{5.78E-02}\\
		\num{6.84E-02}\big].& & & & & &
	\end{array}
\end{flalign*}
}

\bibliographystyle{informs2014}  
\bibliography{Ref}  

\end{document}